%% file: AKSZ_PC-EH.tex
\documentclass[a4paper,reqno]{amsart}

\usepackage[english]{babel} %% English standard typographical conventions
\usepackage{hyperref} %% hyperlinks
\usepackage{graphicx,color,comment}  %% figures, more colors
\usepackage[dvipsnames]{xcolor} %% more colors

\usepackage{amsmath,amssymb,amsfonts,amsthm} %% ams packages
\usepackage{mathtools}  %% add-on to amsmath
\usepackage[color=Apricot]{todonotes}
\usepackage{tikz-cd}  %%For commutative diagrams with tikz
\usepackage[all,cmtip]{xy} %% For diagrams
\usepackage[bbgreekl]{mathbbol} %%For bold greek letters
\theoremstyle{definition} %%upright text, extra space above and below
	\newtheorem{definition}{Definition}
	\newtheorem{remark}[definition]{{Remark}}

\theoremstyle{plain} %% italic text, extra space above and below
	\newtheorem{theorem}[definition]{Theorem}
	\newtheorem{proposition}[definition]{Proposition}
	\newtheorem{lemma}[definition]{Lemma}
	\newtheorem{corollary}[definition]{Corollary}

\usepackage[bibstyle=alphabetic,citestyle=alphabetic,useprefix,giveninits=true, minbibnames=99,maxbibnames=99,sorting=nyt, sortcites=true,backend=biber]{biblatex}
\renewbibmacro{in:}{} %% removes in: before journals
\usepackage{csquotes} %%necessary for biblatex
\bibliography{bibliography}

\newcommand{\pard}[2]{\frac{\delta#1}{\delta#2}}
\newcommand{\intl}{\int\limits}

\newcommand{\tc}{\widetilde{c}}
\newcommand{\tom}{\widetilde{\omega}}
\newcommand{\te}{\widetilde{e}}

\newcommand{\txi}{\widetilde{\xi}}

\newcommand{\tg}{\widetilde{g}}

\newcommand{\be}{\mathbf{e}}
\newcommand{\bom}{\boldsymbol{\omega}}

\newcommand{\bxi}{\boldsymbol{\xi}}
\newcommand{\bc}{\mathbf{c}}

\newcommand{\bg}{{\gamma}}
\newcommand{\bphi}{{\varphi}}
\newcommand{\bpi}{{\Pi}}

\newcommand{\FF}{\mathfrak{F}}

\newcommand{\AKSZ}{\textsf{\tiny AKSZ}}

\newcommand{\res}{\text{res}}
\newcommand{\T}{\mathfrak{T}}
\newcommand{\Fp}[2]{\mathcal{F}_{#2}^\partial(#1)}
\newcommand{\Sp}[2]{S_{#2}^\partial(#1)}
\newcommand{\Qp}[2]{Q_{#2}^\partial(#1)}
\newcommand{\varp}[2]{\varpi_{#2}^\partial(#1)}
\newcommand{\alp}[2]{\alpha_{#2}^\partial(#1)}
\newcommand{\uF}[1]{{\mathcal{F}}_{R}(#1)}
\newcommand{\uFF}[1]{{\mathfrak{F}}_{R}(#1)}
\newcommand{\uS}[1]{{S}_{R}(#1)}
\newcommand{\uV}[1]{{\varpi}_{R}(#1)}
\newcommand{\uQ}[1]{{Q}_{R}(#1)}
\newcommand{\UQ}{{Q}_R}
\newcommand{\calF}{\mathcal{F}}
\DeclareMathOperator{\BFV}{\mathit{BFV}}
\DeclareMathOperator{\BV}{\mathit{AKSZ}}

\makeatletter
\newcommand{\zzlabel}[1]{\ifmeasuring@\else\ltx@label{#1}\fi} %%new label (necessary if amsmath is present)
\makeatother

\newcounter{terms}[equation] %%counter for terms in a equation
\title[General Relativity and the AKSZ construction]{General Relativity and the AKSZ construction}

\author{G. Canepa}
\address{Institut f\"ur Mathematik, Universit\"at Z\"urich, Winterthurerstrasse 190, 8057 Z\"urich, Switzerland}
\email{giovanni.canepa@math.uzh.ch}

\author{A. S. Cattaneo}
\address{Institut f\"ur Mathematik, Universit\"at Z\"urich, Winterthurerstrasse 190, 8057 Z\"urich, Switzerland}
\email{cattaneo@math.uzh.ch}

\author{M. Schiavina}
\address{Institute for Theoretical Physics, ETH Zurich, Wolfgang Pauli strasse 27, 8092, Z\"urich, Switzerland}
\address{Department of Mathematics, ETH Zurich, R\"amistrasse 101, 8092, Z\"urich, Switzerland }
\email{micschia@phys.ethz.ch}

\thanks{This research was (partly) supported by the NCCR SwissMAP, funded by the Swiss National Science Foundation. G.C. and A.S.C. acknowledge partial support of SNF Grant No. 200020\_192080. M.S. acknowledges partial support from Swiss National Science Foundation grants P2ZHP2\_164999 and P300P2\_177862. Part of the research has been carried out while M.S. was a visiting scientist at Max Planck Institute for Mathematics, Bonn, funded by the Max Planck Foundation.}

\begin{document}
\begin{abstract}
    In this note the AKSZ construction is applied to the BFV description of the reduced phase space of the Einstein-Hilbert and of the Palatini--Cartan theories in every space-time dimension greater than two. In the former case one obtains a BV theory for the first-order formulation of Einstein--Hilbert theory, in the latter a BV theory for Palatini--Cartan theory with a partial implementation of the torsion-free condition already on the space of fields. All theories described here are BV versions of the same classical system on cylinders. The AKSZ implementations we present have the advantage of yielding a compatible BV-BFV description, which is the required starting point for a quantization in presence of a boundary.
\end{abstract}

\maketitle

\tableofcontents

\section*{Introduction}
A Lagrangian field theory $\mathsf{F}$ on a cylinder $\Sigma\times I$, where $I$ is a ``time'' interval, can be given a corresponding Hamiltonian description in terms of a symplectic manifold (the phase space) of the possible initial conditions on $\Sigma$ and a Hamiltonian that describes the time evolution. If the Lagrangian is degenerate, its Euler--Lagrange equations yield, in addition to time evolution, some constraints that have to be taken into account when specifying the initial conditions. The true phase space, called the ``reduced phase space'', is typically described as the symplectic reduction of the coisotropic submanifold defined by the constraints (hence the name).

This reduction is often singular, and one possible description is by means of a cohomological resolution: one introduces a complex whose cohomology is the algebra of functions of the reduced phase space. In addition, one wants this resolution to feature also the symplectic/Poisson nature of the phase space, and a solution to this problem is provided by the Batalin--Fradkin--Vilkovisky (BFV) formalism \cite{BV3} (see also \cite{Stasheff1997,SchaetzTH,Schaetz:2008}). We denote by $\FF^\partial$ the collection of data associated to the reduced phase space of a Lagrangian theory $\mathsf{F}$, as a BFV theory (Definition \ref{def:BFV}). 
 
On the other hand, a flexible way to deal with a degenerate Lagrangian is the Batalin--Vilkovisky (BV) formalism \cite{BV2}, which allows a cohomological resolution of the space of solutions to the Euler--Lagrange equations modulo symmetries but is also the starting point for perturbative quantization. We denote with $\FF$ the BV data associated to a classical Lagrangian field theory $\mathsf{F}$, as a BV theory (Definition \ref{def:BV}).

To quantize a Lagrangian field theory $\mathsf{F}$ on a cylinder $\Sigma\times I$, one needs a good relation between its associated BV and BFV data $\FF$ and $\FF^\partial$. 
In \cite{CMR2012b} an explicit procedure was introduced to construct what in \cite{CMR2012} is called a BV-BFV theory (Definition \ref{def:BVBFV}), associating to the BV data $\FF$ certain BFV data denoted by $\BFV(\FF)$ in a way suitable for quantization \cite{CMR2} --- under some regularity assumptions. In regular cases it relates $\FF$ and $\FF^\partial$, so that that $\BFV(\FF)=\FF^\partial$. 

While it is true that both $\FF$ and $\FF^\partial$ depend on $\mathsf{F}$, the relation $\BFV(\FF)=\FF^\partial$ is not guaranteed, and it is a necessary requirement for BV quantisation with boundary \cite{CMR2}. This relation turns out to hold for a large variety of field theories, including general relativity (GR) in the Einstein--Hilbert (EH) formulation in any space--time dimension greater than $2$ \cite{CSEH}. However, the procedure notably fails in the case of GR in the Palatini--Cartan (PC) formulation\footnote{{There appears to be no uniform consensus on the nomenclature to best attribute and label the theory that will be described in Sections \ref{sec:PC-BFV_theory} and \ref{s:AKSZPC}. We discuss our choice ``Palatini--Cartan'' in \cite{CS2019} (see also references therein). Other choices include the names of Einstein, Weyl, Sciama and Kibble, in various combinations.}} in $3+1$ dimensions \cite{CS2017}, as the construction of $BFV(\FF)$ is obstructed. However, $\FF^\partial$ exists and has been presented in \cite{CCS2020}.

Conversely, given a BFV theory $\FF^\partial$ associated to a manifold $\Sigma$, there is a standard way\footnote{To the best of our knowledge, the first explicit application of the AKSZ construction to a BFV target to produce a BV structure in one dimension goes back to \cite{GrigorievDaamgard}.} to produce a BV theory on $\Sigma\times I$ by means of a construction due to Alexandrov, Kontsevich, Schwarz and Zaboronski (AKSZ \cite{AKSZ}). The resulting BV theory, which we temporarily denote here by\footnote{This construction is clarified in Theorem \ref{AKSZtheorem}, and $\BV(\FF^\partial)$ will be denoted  $\FF^{\AKSZ}(I;\FF^\partial)$.} $\BV(\FF^\partial)$, satisfies automatically the regularity assumptions {required by} the BV-BFV formalism, and we also have $\BFV(\BV(\FF^\partial))=\FF^\partial$.

On the other hand, in general $\BV(\BFV(\FF))$ will not be the same as $\FF$. 
In fact, the AKSZ construction produces a theory that is invariant under reparametrization of $I$, which is certainly different from $\FF$ if the latter does not enjoy this invariance. In this case $\BV(\BFV(\FF))$ is a version of $\FF$ with ``frozen time'' and may be used to describe a change in the polarization chosen for the quantization of the reduced phase space (see \cite[Remark 2.38]{CMR2}).
If $\FF$ is reparametrization invariant --- e.g. a topological field theory or GR --- we may wonder whether $\BV(\BFV(\FF))$ and $\FF$ are somehow related. In the case of AKSZ topological field theories, it turns out that $\BV(\BFV(\FF))$ and $\FF$ are actually the same. For more general reparametrization invariant theories we might expect the two to be equivalent, in one of the possible ways presented below.

A BV theory $\FF$ is essentially composed of a (-1)-symplectic manifold $(\calF,\varpi)$ and an action functional $S$ over it. We say that $\FF_1$ and $\FF_2$ are strongly BV-equivalent if there is a symplectomorphism $\phi\colon (\calF_1,\varpi_1)\to(\calF_2,\varpi_2)$ that relates their action functionals, i.e. $S_1=\phi^*S_2$. This in particular implies that their BV cohomology groups are isomorphic. A nontrivial example of strong BV-equivalence is the one betwen PC and $BF$ theory in $3$ space--time dimensions \cite{CSS2017,CaSc2019}.

If $\FF_2$ is obtained from $\FF_1$ by a partial integration of the fields\footnote{This is more appropriately called BV-pushforward or BV fiber integral, see \cite[Section 2.2.2]{CMR2}.} (with some partial gauge fixing), we say that 
$\FF_2$ is an effective theory for $\FF_1$. 
We say that two BV theories are effectively BV-equivalent if
one is (strongly BV-equivalent to) an effective theory for the other.
Typical cases for this are Wilson renormalization or the passage to a second-order theory from its associated first-order formulation. Another important example is given by elimination of so-called auxiliary fields. In that case, one can argue that effective equivalence also preserves the BV cohomology \cite{Henn,BBH} (see Remark \ref{rem:equivalences}).

A third case is when the theories $\FF_1$ and $\FF_2$ have the same space of classical solutions modulo symmetries. We speak in this case of classical equivalence. A typical case of classical equivalence is that between EH and PC. Observe that this is equivalent to just asking that the degree-zero BV cohomologies of the two theories coincide, making this kind of equivalence weaker.

In this paper we study this question for EH and PC models of gravity in any space--time dimension greater than $2$, assuming that the metric encoded in the BFV data $\FF^\partial$ is nondegenerate (i.e. assuming that the manifold $\Sigma$ is either spacelike or timelike but not lightlike). In the case of EH, we show that $\FF$ and $\BV(\BFV(\FF))$ are effectively equivalent, with the former being actually the first-order formulation of %(the ADM version of) 
the latter. 

In the case of PC in three dimensions, where $\BFV(\FF)=\FF^\partial$ holds, we show that $\BV(\BFV(\FF))$ and $\FF$ are strongly BV equivalent, which is not unexpected, since PC is strongly BV equivalent to $BF$ theory \cite{CSS2017,CaSc2019}, and the latter is a topological AKSZ theory. Instead, for higher dimensional PC theory we show that $\BV(\FF^\partial)$ and $\FF$ are classically equivalent\footnote{{On an open subset of the moduli space of solutions}.}, with $\FF^\partial$ the BFV data constructed from the reduced phase space of PC theory \cite{CS2019,CCS2020}. This case is particularly interesting because the BV-BFV construction for PC is obstructed in dimension 4 (and presumably higher).  The data $\BV(\FF^\partial)$ resulting from the AKSZ construction is a new BV theory defined on cylinders that is still classically equivalent to EH, but also compatible with the BV-BFV formalism (by construction via the AKSZ procedure). Classically, it is simply PC on a smaller space of fields, where part of the torsion-free condition is imposed a priori instead of through the Euler--Lagrange equations.

Our result addresses the problem presented in \cite{CS2017}, where it was pointed out that PC theory in dimension greater than three must be complemented with requirements on field configurations at the boundary in order to induce a well-defined BV-BFV structure.  One possible way to construct a BV-BFV structure for PC theory is to assume vector fields generating diffeomorphisms transversal to the boundary to vanish at the boundary. Denote by $\underline{\mathfrak{F}}$ the resulting BV theory. In \cite[Section 5, Remark 34]{CS2017} this was shown to be insufficient to describe the full reduced phase space of GR, as the Hamiltonian constraint is lost in the process: this means that $BFV(\underline{\mathfrak{F}})\not=\mathfrak{F}^\partial$. Alternatively one may require certain components of the Lorentz connection to vanish on the boundary, although this condition is not natural for general manifolds. One way of reading our paper is to make these conditions natural on cylindrical manifolds, in the sense that we present a { version of PC theory, with the compatibility requirements already implemented. In fact, the resulting AKSZ theory has the same equations of motion and the same symmetries, but the AKSZ procedure restricts the moduli space of solutions to an open subset. This is akin to restricting to globally hyperbolic solutions.} One can think of the extra required conditions as imposing part of the equations of motion that fix $\omega$ to be the Levi-Civita connection for the metric induced by a tetrad $e$. This is discussed in Section \ref{sec:Interpretation}.

Let us stress that having a well-defined BV-BFV structure is a necessary requirement for the quantisation of BV theories with boundary \cite{CMR2}. The fact that the boundary-compatible AKSZ version of PC theory is (possibly) only classically equivalent to the original PC formulation reinforces the idea that care must be placed when attempting BV quantisation of the latter.

{A related approach is the `parent formulation' by Barnich and Grigoriev \cite{BGparent1,Grigorievparent} which derives an AKSZ construction of the BV theory from the jet space formalism (trivariational complex). What is crucially different in our construction is that we consider, as a target, a symplectic description of the classical boundary states. This involves a careful symplectic reduction of the naively associated boundary spaces\footnote{This association is the natural restriction of fields and normal jets to the boundary, see \cite{CMR2012}.}. The result of our construction is not only a BV reformulation of the original bulk theory, but a reformulation that is compatible with the boundary as a 1-extended BV-BFV theory (see Definition 5), which is the starting point for quantum (or at least semiclassical) considerations for a theory with boundary \cite{CMR2}.

For the same reason, unlike the presymplectic AKSZ formulation presented by Grigoriev et Alkalaev in \cite{AlkGri} and \cite{grigoriev2016presymplectic}, our BV-BFV description of PC gravity is based on a symplectic structure, which is essential for quantization. This does not arise directly from a reduction of the natural presymplectic BFV structure derived from BV in the bulk, which is impossible for $N\geq4$ as shown in \cite{CS2017}, but it is the symplectic BFV structure  \cite{CCS2020} that resolves the reduced phase space of the theory. 

Finally, note that in this paper we consider two separate applications of the AKSZ ``reconstruction'' of a parametrization-invariant bulk BV theory from its boundary BFV structure, respectively for two formulations of GR (EH and PC). We do not discuss the equivalence between EH and PC, but we investigate the appropriate BV equivalence between each formulation and its own AKSZ ``reconstruction.''

These considerations do not exclude, however, some deeper connection between our construction and the ones mentioned above, which are definitely worth exploring.
}

The paper is organised as follow. In Sections \ref{s:BVformalism} and \ref{s:AKSZbckgnd} we will outline the BV-BFV and AKSZ constructions, while Section \ref{s:BFVgrav} is a brief review of the construction of the BFV data for Einstein--Hilbert and Palatini--Cartan theories of gravity, as presented respectively in \cite{CS2016b} and \cite{CCS2020}.

Finally, in Sections \ref{s:AKSZEH} and \ref{s:AKSZPC} we will apply the AKSZ construction to the BFV data of EH and PC gravity, respectively, and compare it with the BV data for the two formulations as presented in \cite{CS2016b} and \cite{CS2017}.

\subsection*{Relevance and outlook}
This work is intended as a first reaping, as a result of a few years of sowing, in a program directed at an analysis of classical General Relativity seen through the lens of the BV formalism with boundary, an attempt at formalising its quantisation within the BV-BFV formalism \cite{CMR2}. The program was initiated in \cite{ScTH,CS2016b,CS2017}, where a few inconsistencies in the behaviour of GR in the presence of boundaries in dimension $4$ were detected, and was later extended in \cite{CSS2017,CaSc2019,CCS2020}, where the comparison with the three dimensional analogue was made.

This series of works is motivated by the obstruction encountered in defining the BV-BFV data for Palatini--Cartan gravity, a requirement for the BV quantisation program with boundary, which has otherwise provided very reliable and flexible (see \cite{CMR2,CMRtop} for the quantisation of $BF$ theory, \cite{IrMn} for Yang--Mills theory in dimension 2, \cite{CMW,CMW2019} for split Chern--Simons theory and \cite{CaMoWe} for a general approach to a class of AKSZ models, including the Poisson sigma model). No obstruction to BV quantisation with boundary is otherwise present for Einstein--Hilbert theory, and this discrepancy points at the fact that classical equivalence of field theories might be too coarse a classification to have bearing on the respective quantum theories.

The results contained in this paper close the circle, so to speak, in the comparison of classical BV general relativity with boundary, between EH and PC formulations. As a matter of fact, while the AKSZ construction for EH theory is effectively equivalent to the BV theory analysed in \cite{CS2016b}, this is not the case in PC theory analysed in \cite{CS2017} (it is only included within). This fact, together with the equivalence of the reduced phase spaces for EH and PC theories \cite{CS2019,CCS2020}, can be interpreted as a confirmation that BV Palatini--Cartan theory \emph{must} be supplemented with additional requirements on fields, or otherwise restricted, in order to be viable for BV quantisation. {The requirements we find, summarised by Definition \ref{def:PCStructuralConstraints} are conditions on the Lorentz connection and its conjugate variable, which effectively restrict the space of fields. We find that these conditions are somewhat natural on cylinders.

The very ultimate goal of the construction presented in this paper is the grail of quantisation of gravity. We do not attempt doing it here. What we present is the preliminary setting for a perturbative quantisation on cylinders resulting in the quantum evolution operator from the initial to the final quantum space of states. Due to the degeneracy of the actions (related to gauge invariance), one needs a formalism that allows imposing gauge fixings and checking that the results are independent thereof, up to equivalences that are under control. In the absence of free boundaries (i.e., in the computation of partition functions and expectation values), there are several good methods to do this, including BRST and BV. In the presence of boundary, the best developed method is a compatible combination of BV in the bulk and BFV on the boundary. The compatibility is the main issue here, and this paper discusses it in the context of GR theories.

Performing the actual quantisation, which is far beyond the scope of this paper, implies choosing a polarization on the boundary and a gauge fixing in the bulk, computing the resulting propagators, regularizing the theory, and performing renormalisation in a compatible way with the BV-BFV data (the quantum master equation, and its version with boundary \cite{CMR2}). In the case of gravity, one of course expects an infinite number of independent counterterms to be taken care of. Clever or miraculous ways to keep them under control are the same issue as in other treatments (without boundary): we do not claim to have a better recipe for this issue, but just to have a method to incorporate free boundaries. A minimal way to proceed, as, e.g., in \cite{BFR}, is to allow for infinite counterterms (which is algebraically possible and allows for the construction of families of effective theories, even though the predictivity at all energies is missing).

Naturally, since the outlook of this extended program is that of addressing quantisation of General Relativity (with boundary), we wish to stress that without the observations produced in this preliminary phase, an early attempt at directly quantising PC theory might have been thwarted by the very obstructions highlighted by our investigations. 

}

In this sense, we believe the correct preparation of a field theory for its perturbative quantisation to be of crucial importance to drive the scientific effort towards sensible questions, and divert it when evidence is presented of a potential roadblock ahead. This should be of particular interest for the scientific community heavily involved with the study of Palatini--Cartan theory as a fundamental building block for a quantum theory of gravity.

\subsection*{Acknowledgements}
We thank G. Barnich and M. Grigoriev for useful discussions on both content and context, as well as the anonymous referees, whose comments have helped improve our manuscript. M.S. would also like to thank C. Blohmann and A. Weinstein for numerous scientific interactions relevant to this paper.

\section{Background}
One of the goals of this paper is the construction of a BV theory on a cylindrical manifold $ \Sigma \times I $ by means of the AKSZ construction, with target a BFV theory associated to $\Sigma$. In this section we introduce the basic definition of the BV(-BFV) and AKSZ formalisms, together with the relevant notions of equivalence that will allow us to compare theories.
We refer to \cite{ BV1, BV2, BV3,CMR2012,CMR2012b}  for a more detailed introduction and more insight in the meaning and the motivations for the following definitions and theorems. For an introduction of the BV (-BFV) formalism and  gravity see \cite{CS2016a, CS2019, CaSc2019}. Other versions and interpretations of the BV formalism for gravity can be found in \cite{BFR}.

\subsection{The Batalin--Vilkovisky formalism}\label{s:BVformalism}

\begin{definition}\label{def:BV}
A BV theory is a quadruple $\FF= \left( \calF, S, \varpi, Q\right)$ where $\calF$ is a graded manifold (the \emph{space of BV fields}) endowed with a degree $-1$ symplectic form $\varpi$, $S\colon \calF \rightarrow \mathbb{R}$ is a degree 0 functional (the \emph{BV action}) and $Q$ is the (odd) Hamiltonian vector field of $S$ with respect to $\varpi$ satisfying $[Q,Q]=0$. 
\end{definition}

\begin{remark}
Since $Q$ is the Hamiltonian vector field of $S$, i.e. $\iota_{Q} \varpi = \delta S $ where $\delta$ is the de Rham differential on $\calF$ and $\iota_Q$ is the contraction w.r.t. $Q$, we can rewrite the equation $[Q,Q]=0$ as $(S,S)=0$ where $(\cdot,\cdot)$ denotes the Poisson bracket defined by $\varpi$. The latter equation is called the Classical Master Equation (CME). 
\end{remark}

\begin{definition}\label{def:BFV}
An exact BFV theory is a quadruple $\FF^{\partial}= \left( \calF^{\partial}, S^{\partial}, \varpi^{\partial}, Q^{\partial}\right)$ where $\calF^{\partial}$ is a graded manifold (the \emph{space of boundary fields}) endowed with a degree-$0$ exact symplectic form $\varpi^{\partial}= \delta \alpha^{\partial}$, $S^{\partial}: \calF^{\partial} \rightarrow \mathbb{R}$ is a degree 1 functional and $Q^{\partial}$ is the Hamiltonian vector field of $S^{\partial}$ with respect to $\varpi^{\partial}$ such that $[Q^{\partial},Q^{\partial}]=0$.
\end{definition}

\begin{remark}
Typical examples of BV and BFV theories are modeled on sections of bundles over differentiable manifolds, possibly with boundary, with $\varpi^{(\partial)},S^{(\partial)}$ and $Q^{(\partial)}$ respectively a local two-form, functional and vector field. Throughout the paper, when specifying BV theories, we will assume that the equations $\iota_Q\varpi =\delta S$ and $(S,S)=0$ are satisfied only up to boundary terms. The failure of said equations will be controlled by the data of a BV-BFV theory, as follows. { It is often convenient, in this scenario, to define the slightly more general concept of a \emph{relaxed} BV theory, i.e. data $\FF=(\calF,S,\varpi,Q)$ as in Definition \ref{def:BFV}, but without the requirement that $Q$ be the Hamiltonian vector field of $S$. If we are given a BV theory on a closed manifold without boundary, we can consider the same local data as a relaxed BV theory on a manifold with boundary.}
\end{remark}

\begin{subequations}
\begin{definition}[\cite{CMR2012}]\label{def:BVBFV}
A { (relaxed)} BV theory $\FF=\left(\calF,S, \varpi, Q\right)$ is said to be $1$-extended to the BFV theory $\FF^\partial=\left(\calF^{\partial}, S^{\partial}, \varpi^{\partial}, Q^{\partial}\right)$ if there exists a surjective submersion $\pi: \mathcal{F} \rightarrow \mathcal{F}^{\partial}$, such that the following compatibility relation is satisfied:
\label{BVBFVeqts}\begin{equation} \label{rCME}
\iota_{Q} \varpi  = \delta S  + \pi^{*} \alpha^{\partial}  
\end{equation}
The data $\FF^{\uparrow1}=\left(\FF,\FF^\partial,\pi\right)$ will be called $1$-extended BV-BFV theory.
\end{definition}

\begin{remark}
Notice that, from the data above, the following relation follows:
\begin{equation} \label{b-action}
\iota_{Q} \iota_{Q} \varpi= 2 \pi^{*}S^{\partial}.
\end{equation}
\end{remark}
\end{subequations}

The following definitions compare two different BV (or BFV) theories.

\begin{definition}\label{def:strongBVeq}
Two B(F)V theories $\FF_1^{(\partial)}$ and $\FF_2^{(\partial)}$ are said to be strongly B(F)V-equivalent if there exists a symplectomorphism
$$ \Phi : (\calF_1^{(\partial)}, \varpi_1^{(\partial)} ) \rightarrow (\calF_2^{(\partial)}, \varpi_2^{(\partial)})$$
preserving the BV action: $\Phi^* S_2^{(\partial)} = S_1^{(\partial)}$. The map $\Phi$ is called a strong B(F)V-equivalence.
\end{definition}

\begin{definition}\label{def:BV-inclusion}
Let $\FF_1$ and $\FF_2$ be two { (relaxed)} BV theories. A (relaxed) BV-inclusion  $\mathfrak{I}:  \FF_1 \rightarrow \FF_2$ is an inclusion of (super)manifolds $\iota: \calF_1 \rightarrow \calF_2$ such that $\varpi_1 = \iota^* \varpi_2$ and {$\iota^*Q_1 = Q_2\iota^*$. If the two theories are relaxed we will additionally require $\iota^* S_2 = S_1$.} In this case we say that $\FF_1$ is a BV-subspace\footnote{In the math literature, a map with this compatibility between the symplectic structures and the cohomological vector fields is known as a morphism of dg symplectic manifolds.} of $\FF_2$.
\end{definition}

{
\begin{remark}\label{rem:HamiltonianVF}
%Let $(\calF_1,\varpi_1,Q_1)$ and $(\calF_2,\varpi_2,Q_2)$ be two (-1)-symplectic $Q$-manifolds. The most basic of inclusions in this setting is a smooth inclusion of graded symplectic manifolds $\iota\colon \calF_1\to \calF_2$ such that $\iota^* Q_2 = Q_1 \iota^*$, and $\iota^*\varpi_2=\varpi_1$.  When the symplectic $Q$-manifolds are endowed with a degree-$0$ action functional, a \emph{relaxed} BV-inclusion additionally requires $\iota^*S_2 = S_1$. 
Naturally, if $Q_1$ and $Q_2$ are the Hamiltonian vector fields of $S_1$ and $S_2$ respectively, the condition $\iota^* S_2 = S_1$ is equivalent to the condition $\iota^* Q_2 = Q_1 \iota^* $, up to a constant. 
\end{remark}
}

\begin{proposition}\label{prop:equivalence+inclusion}
The composition of a strong BV equivalence and a BV inclusion is in turn a BV inclusion. 
\end{proposition}
\begin{proof}
The map $\Phi \circ \iota$ satisfies trivially the properties of a BV inclusion.
\end{proof}

A notion that we will need to compare theories is that of BV-pushforward. This notion is usually phrased at the quantum level \cite{CMR2,Mnev2017}, where the additional data of a BV Laplacian needs to be provided. However here we are interested mainly in its classical counterpart. The basic setting is the same, although we consider the following simplifying assumptions. Suppose that we have a splitting of a graded symplectic manifold $(\calF, \varpi)$ so that $\calF= \calF' \times \calF''$, with $\varpi= \varpi'+\varpi''$, and let $\mathcal{L}$ be a Lagrangian submanifold of $(\calF'', \varpi'')$ endowed with a half-density $\mu$ on $\calF''$, which thus defines by restriction a density $\mu_{\mathcal{L}}$ on $\mathcal{L}$. Denote coordinates $(z',z'')$ respectively in $\calF',\calF''$, and let $z''\in\{x,x^\dag\}$ be Darboux adapted coordinates such that $x$ parametrises $\mathcal{L}$ and $x^\dag$ are transversal. 

\begin{definition}\label{def:BVL}
We define the Batalin--Vilkovisky--Legendre transform of a functional $S\in C^\infty (\calF)$, with respect to the Lagrangian $\mathcal{L}\subset\calF''$, as $S_{\text{BVL}}\in C^\infty(\mathcal{F'})$:
\begin{equation}\label{BVL}
S_{\text{BVL}} = S(z,x_0, x^\dag=0)
\end{equation}
where $x_0$ is a critical point for $S$ (assumed unique):
$$
\pard{S}{x}\bigg\vert_{x^\dag=0}(x_0)=0.
$$
\end{definition}

Starting from a BV theory on $(\calF, \varpi)$ we build a theory on $(\calF', \varpi')$ by means of a gauge-fixed fiber integral along $\calF''$, with gauge-fixing Lagrangian $\mathcal{L}$. In other words, if $S$ denotes a BV action on $\calF$ we consider the effective result of the BV pushforward (fiber integral) to be 
\begin{equation}\label{Seffdef}
\exp\left(\frac{i}{\hbar}S_{\text{eff}}\right)\coloneqq \int_{\mathcal{L}\subset \calF''} \exp\left(\frac{i}{\hbar}S\right)\bigg\vert_{\mathcal{L}}\mu_{\mathcal{L}}
\end{equation}
where the integral is defined perturbatively as a power series in $\hbar$. Note that $S_\text{BVL}$ is the dominant term of $S_\text{eff}$. When $S$ depends only quadratically on the variables on $\mathcal L$, the only correction is $i\hbar/2$ times the logarithm of the determinant of the quadratic form.

{
\begin{remark}\label{rem:equivalences}
Let us comment briefly on the notion of equivalence of theories in the BV formalism. When the moduli spaces of solutions of the equations of motion for two theories coincide, the theories are said to be classically equivalent. A finer notion of equivalence requires that the BV cohomologies be isomorphic\footnote{Classical equivalence is the less restrictive requirement that only their $0$-cohomologies be isomorphic.}, and it allows for an extension to the case with boundary. One looks at the bicomplex given by the BV operator and the de Rham differential $(\Omega^{\bullet,\bullet}(\mathcal{F}_i,M), Q_i-d)$, and equivalence in this sense requires two theories to have quasi-isomorphic local de Rham/BV complexes. The typical argument for equivalence involves the elimination of so-called generalised auxiliary fields and trivial pairs \cite{BBH,Henn}. This notion, however, is suboptimal because it relies on homological perturbation theory, which potentially could output an infinite tower of ghosts and antighosts in the process of constructing an equivalent theory. In other words, that approach ---albeit somewhat well-established--- does not provide a direct answer to the question of whether two \emph{given} BV theories (in the form of two Hamiltonian dg-manifolds) are equivalent, provided their degree-$0$ sectors are classically equivalent. For this, instead, a spectral sequence argument would be more appropriate.\footnote{In an optimal scenario, to auxiliary fields one can associate a subcomplex with trivial cohomology. More generally, one has a filtration of the original BV complex, so that the associated spectral sequence converges. However, this is in general quite hard to prove.}

%it is only an equivalence of chain complexes, and not of Hamiltonian dg manifolds\footnote{By this we mean that it does not explicitly require any condition on the BV spaces as shifted symplectic manifolds with a distinguished function.}. Another potential issue that arises with the elimination of auxiliary fields is that the general proof of existence of an equivalent BV theory for a theory that classically only differs in auxiliary fields content \cite{BBH}, 

On the other hand, one could phrase equivalence in the BV-BFV formalism, requiring existence of the BV-BFV structure and equivalence of the respective BV and BFV cohomologies (an explicit example of this, is the strong equivalence of General Relativity and BF theory in three dimensions \cite{CaSc2019}). This existence requirement might become an obstruction to BV-BFV equivalence\footnote{Observe that this is essentially the strictification requirement of the presymplectic BFV data, see Remark \ref{rem:sympBFV}.}. For our purposes, then, even assuming that by removal of generalised auxiliary fields one may prove some BV equivalence between  Palatini--Cartan and Einstein--Hilbert theories, the results of \cite{CS2017} and of the present paper show a discrepancy of the theories in the BV-BFV sense.
\end{remark}
}

\subsection{The AKSZ construction}\label{s:AKSZbckgnd}
Let $X$ be a graded manifold and $N$ an ordinary manifold and let $\mu_N$ be the canonical Berezinian\footnote{Recall that a function on T[1]N is the same as a differential form on N. Integrating a function on $T[1]N$ against the canonical Berezinian $\mu_N$ is by definition the same as integrating the corresponding differential form on $N$, which we assume to be oriented.} on $T[1]N$.
% endowed with a measure $\mu_N$.
\begin{definition}[Transgression map]\label{transgressionmap1}
Consider the map
\begin{equation}\label{transgression}
	\T^{(\bullet)}_N \colon \Omega^\bullet(X) \longrightarrow \Omega^\bullet\left( \mathrm{Map}(T[1]N,X)\right)
\end{equation}
defined by $\T^{(\bullet)}_N\coloneqq p_*\mathrm{ev}^*$, where
\begin{equation}
	\xymatrix{
		\mathrm{Map}(T[1]N, X)\times T[1]N \ar[d]_{p}   \ar[r]_-{\mathrm{ev}}  & X\\
		\mathrm{Map}(T[1]N, X) & 
	}
\end{equation}
and we set $p_*=\intl_N \mu_N$.
We will call $\T^{(\bullet)}_N$ the {transgression map}, and its evaluation a {transgression}. 
\end{definition}
We endow the graded manifold $X$ with a function $S$ of degree $n$ and parity $n \mod 2$, together with a one-form $\alpha$ of degree $n-1$ and parity $n-1 \mod 2$, such that $\varpi=d\alpha$ is nondegenerate and $\{S,S\}=0$ with respect to the Poisson structure defined by $\varpi$. Then we say that $X$ has a Hamiltonian dg-manifold structure, with differential $\{S,\cdot\}$.

Observing that the de Rham differential $d_{N}$ on $N$ can be seen as a degree $1$ vector field on $\mathrm{Map}(T[1]N,X)$ we have
\begin{theorem}[\cite{AKSZ}]\label{AKSZtheorem}
Let $(X,S,\alpha)$ a dg-manifold as described above. Consider the data 
\begin{equation}\label{AKSZdata}
	\FF^{\AKSZ}(N;X,S,\alpha)\coloneqq\left(\calF^{\AKSZ}, S^{\AKSZ}, \Omega^{\AKSZ}, Q^{\AKSZ}\right)
\end{equation}
with $\calF^{\AKSZ}=\mathrm{Map}(T[1]N,X)$, $\Omega^{\AKSZ}\coloneqq \T^{(2)}_N(\varpi)$, the functional $S^{\AKSZ}\colon \calF^{\AKSZ}\to \mathbb{R}$,
\begin{equation}\label{e:AKSZ_generic-action}
	S^{\AKSZ}\coloneqq\T_N^{(0)}(S) + \iota_{d_{N}}\T_N^{(1)}(\alpha).
\end{equation}
and the cohomological vector field $Q^{\AKSZ}$ such that $\iota_{Q^{\AKSZ}}\Omega^{\AKSZ}=\delta S^{\AKSZ}$. Then, $\FF^{\AKSZ}(N;X,S,\alpha)$ defines a BV theory.
\end{theorem}

We will call $\calF^{\AKSZ}\coloneqq \mathrm{Map}(T[1]N,X)$ the \emph{AKSZ space of fields}. Introducing Darboux coordinates $\{p_i, q^i\}$ in $X$ so that $\alpha= p_i dq^i$, the space of AKSZ fields is composed of inhomogeneous differential forms $\mathfrak{P}, \mathfrak{Q}$ on $N$. Then, if we consider $X$ to be the space of sections of a bundle $E\to \Sigma$, that is to say $X=T^*[n-1]C^\infty(\Sigma,E)$, we can write 
\begin{equation}
\Omega^{\AKSZ} = \intl_{\Sigma\times N} \left[\langle \delta\mathfrak{P}, \delta\mathfrak{Q}\rangle 
		 \right]^{\mathrm{top}}
	\equiv  \intl_{\Sigma\times N} \left[\delta\mathfrak{P}_i \delta \mathfrak{Q}^i  \right]^{\mathrm{top}}
\end{equation}
where we have denoted by $\delta$ the deRham differential on spaces of maps and $C^\infty$-sections, and $\text{top}$ denotes the top-form parts of the inhomogeneous differential forms within brackets. We will drop the superscript $\text{top}$ in what follows.

Consider this elementary fact:
\begin{lemma}\label{lemmaAKSZ}
Let $A,B,C$ be graded manifolds, $\phi\colon B\to C$ an isomorphism of graded manifolds, and $\mu_A$ a measure on $A$. Consider the diagram
\begin{equation}
	\xymatrix{
		A\times B \ar[r]^{\mathrm{id}\times \phi} \ar[d]_{\pi_B} & A\times C \ar[d]_{\pi_C}\\
		B \ar[r]^\phi & C
	}
\end{equation}
Then, setting
${\pi_{B}}_* = \int \mu_A \cdot$ and ${\pi_{C}}_* = \int \mu_A \cdot $, we have $\phi^*\circ {\pi_{C}}_* = {\pi_{B}}_*\circ (\mathrm{id}\times \phi)^*$.
\end{lemma}

\begin{theorem}\label{strongequivalenceAKSZ}
Let $(X,S_X,\alpha_X)$ and $(Y,S_Y,\alpha_Y)$ be equivalent Hamiltonian dg-manifolds, i.e. there exists a diffeomorphism $\phi\colon X\to Y$ such that $\varpi_X=\phi^*\varpi_Y$, and $S_X=\phi^*S_Y$. Then $\FF^{\AKSZ}(N;X,S_X,\alpha_X)$ and $\FF^{\AKSZ}(N;Y,S_Y,\alpha_Y)$ are strongly equivalent BV(-BFV) theories for every manifold $N$.
\end{theorem}

\begin{proof}
$\phi\colon X\to Y$ induces an isomorphism 
$$\tilde{\phi}\colon \mathrm{Maps}(T[1]N,X) \to \mathrm{Maps}(T[1]N,Y)$$
by precomposing maps with $\phi^*$ or ${\phi^{-1}}^*$.
Then, we can apply Lemma \ref{lemmaAKSZ} with $B=\mathrm{Maps}(T[1]N,X)$, $C= \mathrm{Maps}(T[1]N,Y)$ and $A=T[1]N$.
\end{proof}

\subsection{One-dimensional AKSZ construction}\label{sec:local_description}
Let $I \subset \mathbb{R}$ be an interval, and $\FF^\partial$ an exact BFV theory. We can construct a BV theory by applying Theorem \ref{AKSZtheorem} on the Hamiltonian dg-manifold underlying an exact BFV theory: 
$$\FF^{\AKSZ}(I;\FF^\partial)\coloneqq \FF^{\AKSZ}(I;\calF^\partial,S^\partial,\alpha^\partial).$$
The resulting space of fields reads
\begin{align*}
\calF^{\AKSZ}= \text{Map}(T[1]I, \calF^{\partial}).
\end{align*}
Since the target space $\calF^{\partial}$  is (locally) a graded vector space, we identify the space of AKSZ fields with
\begin{align*}
\calF^{\AKSZ}= \Omega^{\bullet}(I) \otimes \calF^{\partial}.
\end{align*}
In particular, when $\calF^{\partial}$ is modeled on sections of some bundle over a $(N-1)$-dimensional manifold $\Sigma$, we can view $\calF^{\AKSZ}$ as the space of sections of some (graded) bundle over $\Sigma \times I$. The space $\Omega^{\bullet}(I)$ splits into:
\begin{align*}
\Omega^{\bullet}(I) = C^{\infty}(I) \oplus \Omega^1(I)[-1]
\end{align*}
hence, to each field in $\calF^{\partial}$ we associate two new fields. For simplicity we denote the field in  $C^{\infty}(I)\otimes \calF^{\partial}$ with the same letter as the old one, and use another letter for the one in $\Omega^1 [-1](I)\otimes \calF^{\partial}$.

\begin{proposition}[\cite{CMR2} and \cite{CMR2012}] \label{prop:AKSZ-BV-BFV}
Let $\FF^\partial(\Sigma)=(\Fp{\Sigma}{},\Sp{\Sigma}{}, \varp{\Sigma}{},\Qp{\Sigma}{})$ be an exact BFV theory, with $\Fp{\Sigma}{}\coloneqq \Gamma(E\to \Sigma)$, and $\varp{\Sigma}{}=\delta\alp{\Sigma}{}$. Then, if $I\coloneqq [0,1]$ we have that $\FF^{\AKSZ}(I;\FF^\partial(\Sigma))$ is a 1-extended BV-BFV theory over $\FF^\partial(\Sigma)$ (see Definition \ref{def:BVBFV}).
\end{proposition}
\begin{proof}
Theorem \ref{AKSZtheorem} tells us that $\FF^{\AKSZ}(I;\FF^\partial(\Sigma))$ is a BV theory (up to boundary terms). If we parametrise fields in $\mathcal{F}^{\AKSZ}$ as 
$$
\mathfrak{P}= p(t) + q^\dag(t)dt \qquad \mathfrak{Q} = q(t) + p^\dag(t)dt
$$
we get
$$
\Omega^{\AKSZ} = \intl_{\Sigma\times I } \langle\delta \mathfrak{P}, \delta \mathfrak{Q}\rangle =  \intl_{\Sigma\times I} \left\{\langle\delta p, \delta p^\dag\rangle + (-1)^{|q|+1} \langle\delta q, \delta q^\dag\rangle\right\} dt
$$
and
$$
S^{\AKSZ} = \intl_{\Sigma\times I}  \langle p, d_Iq\rangle +[\T_I^{(0)}(S^\partial(\Sigma)) ]^{\mathrm{top}}.
$$
The transgressed integrand needs to be first-order in $dt$, which leaves us with 
$$
[\T_I^{(0)}(S^\partial(\Sigma))]^{\mathrm{top}} \equiv S^\partial(\Sigma)[p+q^\dag dt,q+p^\dag dt] = \pard{S^\partial(\Sigma)}{p}(q^\dag dt) + \pard{S^\partial(\Sigma)}{q}(p^\dag dt)
$$
Then $Q^{\AKSZ}$ splits in a transversal part plus a \emph{tangential} one: $Q^{\AKSZ} = Q^T + \hat{Q}$, where $Q^Tq^\dag= -\dot{p}$ and $Q^T p^\dag = \dot{q}$ is essentially just deRham differential on $I$, and $\hat{Q}$ is easily obtained:
$$
\hat{Q} p =  \pard{S^\partial(\Sigma)}{q} \equiv Q^\partial p \qquad \hat{Q} q = \pard{S^\partial(\Sigma)}{p}\equiv Q^\partial q
$$
$$
\hat{Q} p^\dag =  \frac{\delta^2S^\partial(\Sigma)}{\delta q \delta p}(p^\dag) + \frac{\delta^2S^\partial(\Sigma)}{\delta p^2}(q^\dag)  \qquad \hat{Q} q^\dag = \frac{\delta^2S^\partial(\Sigma)}{\delta p\delta q}(q^\dag) + \frac{\delta^2S^\partial(\Sigma)}{\delta q^2} (p^\dag).
$$
The boundary terms are easily seen in the given local chart, in fact:
$$
\iota_{Q^{\AKSZ}} \Omega^{\AKSZ} = \delta S^{\AKSZ} + \check{\alpha}
$$
but $\check\alpha$ only sees contributions from $\langle p, dq\rangle$ and, up to sign, we get $\check\alpha = \alpha^\partial(\Sigma)$, with $\delta \alpha^\partial(\Sigma) = \varpi^\partial_{\Sigma}$. Then, the projection of $Q^{\AKSZ}$ along the natural projection map from $\calF^{AKSZ}$ to the space of boundary fields, which coincides with $\calF^\partial(\Sigma)$, is precisely $\Qp{\Sigma}{}$, concluding the argument.
\end{proof}

\begin{remark}\label{rem:Henneaux-strict}
A similar statement to Proposition \ref{prop:AKSZ-BV-BFV} is presented in Henneaux--Bunster \cite[Theorem 18.4.5]{HT}, where one identifies the output of the above AKSZ construction with the (first-order) BV theory obtained by embedding in the BV formalism the generalised Hamiltonian formulation of a given field theory (see also \cite{DGH}). An analogous construction, already in the context of AKSZ theories,  was presented in  \cite{GrigorievDaamgard,BarnichGrigoriev2003}. The added observation of Proposition \ref{prop:AKSZ-BV-BFV} is the compatibility between BV for the bulk and BFV on the boundary, viz., what we call a 1-extended BV-BFV theory in Definition \ref{def:BVBFV}.
\end{remark}

We would like to show that this construction behaves well under equivalences of the relevant BFV data.
\begin{corollary}[Theorem \ref{strongequivalenceAKSZ}]
Let $\FF^\partial_1$ and $\FF^\partial_2$ be two strongly BFV-equivalent (exact) theories,
then $\FF^{\AKSZ}(I;\FF_1^\partial)$ and $\FF^{\AKSZ}(I;\FF_2^\partial)$ are strongly BV-equivalent.
\end{corollary}
\begin{proof}
A strong BV equivalence induces an isomorphism of the underlying dg-manifolds.
\end{proof}

\section{BFV theories of gravity} \label{s:BFVgrav}
\subsection{BFV Einstein--Hilbert Theory}\label{s:BFVEH}
The BFV theory for GR in the Einstein--Hilbert formalism (as described in \cite{CS2016b}) associates to any $(N-1)$-dimensional (pseudo)-Riemannian\footnote{In this paper we will mostly focus on the case where $\Sigma$ is a Riemannian manifold, seen as a spacelike boundary of a cylinder $\Sigma\times \mathbb{R}$. Generalisations of this to the timelike case are straightforward. The relevant BFV data can be found in \cite{CS2016b}.} manifold $\Sigma$ the graded $0$-symplectic manifold
\begin{equation}\label{boundaryfields1}
\Fp{\Sigma}{EH} = T^*\underbrace{\left(S_{nd}^2(T\Sigma) \times \mathfrak{X}[1](\Sigma) \times C^\infty[1](\Sigma)\right)}_{\{\bg, \xi^\partial,\xi^n\}},
\end{equation}
where $S^2_{nd}(\Sigma)$ denotes the space of nondegenerate symmetric tensor fields of rank two, with canonical exact symplectic structure:
\begin{equation}\label{symplecticstructure}
	\varp{\Sigma}{EH} = \delta \alp{\Sigma}{EH}
	=
	\delta \intl_{\Sigma} \langle\bpi, \delta\bg \rangle + \langle\varphi_\partial,\delta\xi^\partial\rangle + \langle\varphi_n,\delta\xi^n\rangle,
\end{equation}
and $\{\Pi,\varphi_\partial,\varphi_n\}$ denote variables in the cotangent fiber{, dual to $\{\gamma,\xi^\partial,\xi^n\}$ respectively, i.e. 
\begin{align*}
    \bpi &\in S^2(T^*\Sigma)\otimes\mathrm{Dens}(\Sigma),\\
    \varphi_\partial & \in \Omega^1(\Sigma)\otimes\mathrm{Dens}(\Sigma),\\ 
    \varphi_n & \in C^\infty(\Sigma)\otimes \mathrm{Dens}(\Sigma).
\end{align*}}

\begin{remark}\label{rem:inversemetrics}
The components $(\gamma)^{ab}$ of a $\bg\in S_{nd}^2(T\Sigma)$ can be thought of as the inverse of a (pseudo-)Riemannian metric on $\Sigma$, which we denote by $\gamma^{-1}$. With a slight abuse of notation\footnote{This is not really problematic, since its variation reads $\delta\sqrt{\bg}= \frac12 \sqrt{\bg}\bg^{ab}\delta\bg_{ab} = -\frac12 \sqrt{\bg}\bg_{ab}\delta\bg^{ab}$ and we can use either formula according to our needs. If we wanted to use the correct notation we should simply replace $\sqrt{\bg}$ with its reciprocal, in formula \eqref{Hamiltonianconstraint}.} we will denote by $\sqrt{\bg}=\sqrt{\det(\gamma_{ab})}$ { the square root of the determinant of the metric on $\Sigma$ that we denote by $\bg^{-1}$ everywhere else. In other words, $\sqrt{\bg}$ is the usual density associated to a metric, i.e. $\sqrt{\bg}\,\mathrm{d^{N-1}x}$ is a volume form on $\Sigma$. Observe that all fields in the fibres of the cotangent bundle \eqref{boundaryfields1} are sections of the respective dual bundles, tensored with densities.  The conjugate field to $\gamma$ is a section of the second symmetric tensor power of the cotangent bundle of $\Sigma$ tensored with densities on $\Sigma$, i.e. $\bpi\in S^2(T^*\Sigma)\otimes\mathrm{Dens}(\Sigma)$ is of the form $\Pi= \sqrt{\bg}\pi$, for $\pi\in S^2(T^*\Sigma)$. A similar decomposition holds for $\varphi_\partial,\varphi_n$.} 
\end{remark}

In addition to $\Fp{\Sigma}{EH}$ and $\varp{\Sigma}{EH}$, the BV-BFV procedure outputs a functional of degree 1 on $\Fp{\Sigma}{EH}$, called BFV action. It is given by the local expression
\begin{align}\label{ADMBoundaction}
\Sp{\Sigma}{EH}=&\intl_{\Sigma} \left\{ H_n\xi^n + \langle\bpi, L_{\xi^\partial}\bg\rangle  + \bphi_nL_{\xi^\partial}\xi^n - \bg(\bphi_\partial,d\xi^n)\xi^n + \left\langle\bphi_\partial,\frac12[\xi^\partial,\xi^\partial]\right\rangle\right\}
\end{align}
where we have denoted the \emph{Hamiltonian constraint density} by
\begin{equation}\label{Hamiltonianconstraint}
H_n(\bg,\bpi) = \left(\frac{1}{\sqrt{\bg}}\left(\mathrm{Tr}_{\bg}[\bpi^2] - \frac{1}{d-1}\mathrm{Tr}_{\bg}[\bpi]^2\right) + \sqrt{\bg}\left(R^\partial -2\Lambda\right)\right)
\end{equation}
with $R^\partial$ is the trace of the Ricci tensor with respect to the metric $\gamma^{-1}$, $\Lambda\in \mathbb{R}$ is the cosmological constant, $\mathrm{Tr}_{\bg}[\bpi^2]=\bg^{ab}\bg^{cd}\bpi_{bc}\bpi_{ad}$ and $\mathrm{Tr}_{\bg}\bpi$ is the density $\bg^{ab}\bpi_{ab}$. Observe that we can also denote the \emph{momentum constraint density} as the density-valued one-form
\begin{equation}\label{momentumconstraint}
H_\partial\colon \mathfrak{X}(\Sigma)\to \mathrm{Dens}(\Sigma)\qquad H_\partial(X) = \langle \bpi,L_{X}\bg\rangle
\end{equation}
for $X\in\mathfrak{X}(\Sigma)$. 

\begin{remark}
One can integrate the density of equation \eqref{Hamiltonianconstraint} against a function $\lambda\in C^\infty(\Sigma)$, or integrate the density in \eqref{momentumconstraint} to get local functionals on fields. Then $\lambda$ and $X$ play the role of Lagrange multipliers, to enforce the so-called Hamiltonian and momentum constraints.
\end{remark}

The Hamiltonian vector field $\Qp{\Sigma}{EH}$ of $\Sp{\Sigma}{EH}$ with respect to $\varp{\Sigma}{EH}$ is a differential on $C^\infty(\Fp{\Sigma}{EH})$, the BFV differential, and its cohomology in degree zero is the reduced phase space defined by the constraints $\{H_n,H_\partial\}$.

\begin{definition}
We define BFV Einstein--Hilbert theory associated to be the assignment
\begin{equation}
\Sigma \leadsto    \FF^\partial_{EH}(\Sigma)=(\Fp{\Sigma}{EH},\Sp{\Sigma}{EH},\varp{\Sigma}{EH},\Qp{\Sigma}{EH}).
\end{equation}
\end{definition}

\subsection{BFV Palatini--Cartan theory}\label{sec:PC-BFV_theory}
Let $\Sigma$ be an $(N-1)$-dimensional compact and orientable\footnote{For simplicity we orient $\Sigma$ and $V$, but the formalism generalizes to nonorientable $\Sigma$ as well, see \cite{CCS2020}.} smooth manifold and let $P \rightarrow \Sigma$ be an $SO(N-1,1)$-principal bundle. Let also $\mathcal{V}$ be the associated vector bundle where each fibre is isomorphic to $(V, \eta)$, an $N$-dimensional vector space with a pseudo-Riemannian inner product $\eta$ on it. We further identify $\mathfrak{so}(N-1,1) \cong \bigwedge^2 \mathcal{V}$ using $\eta$.
% %and we define a map Tr$: \bigwedge ^N V \rightarrow \mathbb{R}$ given by the volume form and such that Tr$(v_i, v_j, \dots, v_k)= \epsilon_{ij\dots k}$ where $\{ v_i\}_i$ is an $\eta$-orthonormal basis of $V$. To keep the notation light we will use the shorthand \todo{controllare}
% $$
% \int \mathrm{Tr}[\dots] \equiv \mathrm{Tr}\int \dots
% $$
% or omit the symbol of the trace when no confusion can arise. 

Furthermore we use the following notation:
\begin{align*}
    \Omega_{\partial}^{i,j}:= \Omega^i\left(\Sigma, \textstyle{\bigwedge^j} \mathcal{V}\right).
\end{align*}

The BFV data for PC theory has been described in \cite{CCS2020} for $N \geq 4$ and in \cite{CaSc2019} for $N=3${, the following discussion will be nontrivial for $N\geq 4$, see Remark \ref{rem:constraintsandconditions}}. The classical fields of the theory are then $ e \in \Omega_{nd}^1(\Sigma, \mathcal{V})$ --- i.e $ \Omega_{\partial}^{1,1}$ plus the nondegeneracy condition that the induced morphism $T\Sigma\to\mathcal{V}$ should be injective --- and the equivalence class of a connection $\omega \in \mathcal{A}(\Sigma)$ under the $e$-dependent relation $\omega  \sim \omega + v $ for $v$ such that {$W_{e^{N-3}}^{1,1}(v)=0$, where 
$$W_{e^{N-3}}^{1,1}\colon \Omega^{1,1}_\partial \to \Omega^{2,2}_\partial, \qquad W_{e^{N-3}}^{1,1}(v) = e^{N-3}\wedge v.$$}
We denote this equivalence class and the quotient space it belongs to by $[\omega] \in \mathcal{A}^{\mathrm{red}}(\Sigma)$. We further assume that the boundary metric
\begin{equation} \label{e:Boundary-metric}
g^\partial_{ij}:=\eta(e_i,e_j)
\end{equation}
is nondegenerate.

{
The symplectic manifold $F^\partial_{PC}$ of (degree-$0$) boundary fields for PC theory is then the total space of the fibre bundle $F^\partial_{PC} \longrightarrow \Omega^1_{\text{nd}}(\Sigma,\mathcal{V})$, with fiber $\mathcal{A}^{\mathrm{red}}$. The manifold $F^\partial_{PC}$ is obtained as the symplectic reduction of the naive boundary two-form $\check{\varpi}=\int e^{N-3}\delta e\delta \omega$, which is pre-symplectic since $\mathrm{ker}(\check{\varpi}) \simeq \mathrm{ker}(W^{1,2}_{e^{N-3}})\not=\{0\}$, as described in \cite{CS2019}. Instead of working with equivalence classes of connections, it is convenient to fix a nonvanishing section $\epsilon_n\in \Gamma(\mathcal{V})$ and enforce a condition called \emph{structural constraint}, which was introduced in \cite{CCS2020}: 
\begin{align}\label{e:constraint}
 (N-3)\epsilon_n e^{N-4} d_{\omega} e  \in \text{Im} W_{e^{N-3}}^{1,1},
\end{align}
in the space of boundary tetrads and connections. In order to do this one restricts to tetrads $e$ that are linearly independent from $\epsilon_n$. In general this implies working in patches over the space of the $e$-fields. However, if $g^\partial$ is space-like, we may choose once and for all $\epsilon_n$ to be time-like, which provides a global choice on the space of the $e$-fields.

\begin{remark}
Considering the boundary by itself, the constraint \eqref{e:constraint} is one of the possible ways to fix the the transformations in the kernel of the presymplectic form. If we take the bulk as well into account, it assumes a more fundamental role: it is the necessary and sufficient condition for the transversal equations of motions to be solvable.
Indeed, the (bulk) equations of motion split into a \emph{tangential} equation
    $$e^{N-3} d_\omega e = 0$$
and a transversal one
    $$(N-3)e_n e^{N-4} d_\omega e = e^{N-3} (d_\omega e)_n,$$
which tells us that $e_n e^{N-4} d_\omega e$ must be in the image of $W_{e^{N-3}}^{1,1}$. Then, upon using the tangential equation, we can replace $e_n$ with some fixed $\epsilon_n$. See Section \ref{sec:Interpretation}.
\end{remark}

Denoting by $\mathsf{S}\subset \Omega_{nd}^1(\Sigma,\mathcal{V})\times \mathcal{A}(\Sigma)$ the submanifold defined by Equation \eqref{e:constraint}, we have:
\begin{proposition}[\cite{CCS2020,CS2019}]\label{prop:symplecticsubspeace}
There exists a symplectomorphism 
$${F}^\partial_{PC}(\Sigma) \longrightarrow \mathsf{S}.$$
\end{proposition}
Effectively, then, one can work on $\mathsf{S}$. The main advantage of this explicit description of the symplectic space of boundary fields is that it allows to explicitly compute the \emph{symplectic} BFV data for PC theory (see Remark \ref{rem:sympBFV}). 

In order to write down the BFV data it is sufficient to fix the equivalence class $[\omega]\in\mathcal{A}^{\mathrm{red}}(\Sigma)$ using \eqref{e:constraint} as done in \cite{CCS2020}, however, when extending $F_{PC}^\partial$ to a graded manifold we can choose to modify the structural constraint by adding terms in the ghosts and antifields. This will turn out to be convenient in what follows.
\begin{definition}\label{def:BFVspaceoffields}
Consider the graded manifold 
\begin{equation}\label{e:BFVspaceoffields}
    \check{\mathcal{F}}_{PC}(\Sigma):= \Omega^{1,1}_{\text{nd}}(\Sigma,\mathcal{V})\times \mathcal{A}(\Sigma) \times T^* \left(\Omega_{\partial}^{0,2}[1]\oplus \mathfrak{X}[1](\Sigma) \oplus C^\infty[1](\Sigma)\right),
\end{equation}
where we denote fields by 
\begin{itemize}
    \item $e \in \Omega^{1,1}_{\text{nd}}(\Sigma,\mathcal{V})$ and $\omega \in \mathcal{A}(\Sigma)$ in degree zero, 
    \item $c \in\Omega_{\partial}^{0,2}[1]$, $\xi \in\mathfrak{X}[1](\Sigma)$ and $\lambda\in \Omega^{0,0}[1]$ in degree one, 
    \item $c^\dag\in\Omega_{\partial}^{3,2}[-1]$, $\lambda^\dag\in\Omega_{\partial}^{3,4}[-1]$ and $\xi^\dag\in\Omega_\partial^{1,0}[-1]\otimes\Omega_{\partial}^{3,4}$ in degree minus one,
\end{itemize}
together with a fixed section $\epsilon_n \in \Gamma(M,\mathcal{V})$, completing the image of $e$ to a basis of  $\mathcal{V}$. We define the \emph{BFV structural constraint} to be condition\footnote{This BFV structural constraint differs from the classical one of equation \eqref{e:constraint} by a term depending on
ghosts and their antifields, so it reduces to the classical structural constraint if we restrict it to the classical space of boundary fields (the body of $\check{\calF}_{PC}(\Sigma)$).}:
\begin{equation}\label{e:BFVstructuralConstraint}
    (N-3)\epsilon_n e^{N-4} d_{\omega} e  + \left(L_\xi^\omega \epsilon_n - [c,\epsilon_n]\right)^{(a)}  c^\dag_a   \in \text{Im} W_{e^{N-3}}^{1,1}
\end{equation}
where $\{(a),(n)\}$ denote components with respect to a basis $\{e_a, \epsilon_n\}$. We define the \emph{space of BFV fields} $\calF_{PC}^\partial(\Sigma)$ to be the space of solutions of the BFV structural constraints within $\check{\calF}_{PC}(\Sigma)$.
\end{definition}

%The space of BFV fields is then given by the bundle $\Fp{\Sigma}{PC}$ defined as
%\begin{equation}\label{e:PCBFVspace}
%\Fp{\Sigma}{PC} \longrightarrow \Omega_{nd}^1(\Sigma, \mathcal{V}),
%\end{equation}
%with local trivialisation on an open $\mathcal{U}_{\Sigma} \subset \Omega_{nd}^1(\Sigma, \mathcal{V})$
%\begin{equation}\label{LoctrivF1}
%\Fp{\Sigma}{PC}\simeq \mathcal{U}_{\Sigma} \times \mathcal{A}^{\mathrm{red}}(\Sigma) \times T^* \left(\Omega_{\partial}^{0,2}[1]\oplus \mathfrak{X}[1](\Sigma) \oplus C^\infty[1](\Sigma)\right),
%\end{equation}
%Using the same notation employed in \cite{CCS2020}, we denote the fields by 

In order to have a better expression of the BFV structure, following \cite[Section 5.2]{CCS2020}, we introduce the field $y^\dag \in \Omega_{\partial}^{3,3}[-1]$ such that the original fields $\lambda^\dag$ and $\xi^{\dag '}$ are recovered through  $\epsilon_n y^\dag = -\lambda^\dag$ and $e_a y^\dag =  -\xi_a^{\dag '}$. This also allows us to write a single expression for all $N \geq 3$.

To complete the definition of the BFV data for Palatini--Cartan theory we consider a degree $1$ functional} and a symplectic form\footnote{This version of the BFV data features a particularly simple action functional, at the price of not expressing the symplectic form in its Darboux form. An alternative can be found in \cite{CCS2020}.} given, respectively, by:
\begin{align}\label{e:action_C3}
\Sp{\Sigma}{PC}= \int_{\Sigma} & c e^{N-3} d_{\omega} e + \iota_{\xi} e e^{N-3} F_{\omega} + \epsilon_n \lambda e^{N-3} F_{\omega} +\frac{1}{2} [c,c] c^{\dag} - L^{\omega}_{\xi} c c^{\dag}\nonumber\\ &+ \frac{1}{2} \iota_{\xi}\iota_{\xi} F_{\omega}c^{\dag}
-[c, \epsilon_n \lambda  ]y^{\dag} + L^{\omega}_{\xi} (\epsilon_n \lambda )y^{\dag} + \frac{1}{2}\iota_{[\xi,\xi]}e y^{\dag},
\end{align}
\begin{align}\label{e:symplectic_form_C3}
\varp{\Sigma}{PC} = \int_{\Sigma} e^{N-3} \delta e \delta \omega + \delta c \delta c^\dag + \delta \omega \delta (\iota_\xi c^\dag) - \delta \lambda \epsilon_n \delta y^\dag+\iota_{\delta \xi} \delta (e y^\dag).
\end{align}
{
Note that each term of the integrals belongs to $\Omega^{N-1,N}$, which can be canonically identified, using $\sqrt{|\det \eta|}$, with the space of densities on $\Sigma$. A detailed explanation can be found in \cite{CCS2020}. However, we will not write down the factor $\sqrt{|\det \eta|}$ explicitly.

\begin{definition}\label{def:BFVdata}
We define BFV Palatini--Cartan theory to be the assignment
\begin{equation}
\Sigma\leadsto    \FF^\partial_{PC}(\Sigma)=(\Fp{\Sigma}{PC},\Sp{\Sigma}{PC},\varp{\Sigma}{PC},\Qp{\Sigma}{PC}).
\end{equation}
with $\Qp{\Sigma}{PC}$ the Hamiltonian vector field of $\Sp{\Sigma}{PC}$ with respect to $\varp{\Sigma}{PC}$.
\end{definition}

\begin{remark}\label{rem:diffwithCCS}
Notice that the data presented in Definition \ref{def:BFVdata} are equivalent to the BFV data presented in \cite{CCS2020}. The cohomological vector field of a function on (the presymplectic manifold) $\check{\mathcal{F}}_{PC}(\Sigma)$ is uniquely fixed by the tangency to the BFV structural constraint \eqref{e:BFVstructuralConstraint}. In \cite{CCS2020} the equivalent choice of the constraint \eqref{e:constraint} is made, and the resulting $Q$'s differ along $\omega$ by a term in the kernel of $W_{e^{N-3}}^{1,2}$.\footnote{Also note that this term vanishes on the body of $\check \calF_{PC}$ when imposing the classical constraints, i.e., the expressions paired to the ghosts in the first three terms of \eqref{e:action_C3}. The two BFV constructions, by \eqref{e:constraint} or by \eqref{e:BFVstructuralConstraint}, describe the characteristic distribution of the coisotropic submanifold determined by the constraints. They extend this distribution in a different way outside of it, but this is irrelevant.} We make this choice in this paper because it makes it easier to compare the AKSZ construction for the constrained space $\calF^\partial_{PC}(\Sigma)$ with a constrained version of PC, see Remark \ref{rem:constraintsandconditions}.
\end{remark}

\begin{remark}\label{rem:constraintsandconditions}
In \cite{CS2017} two of the authors showed that the natural BV data for PC theory (in dimension $N\geq 4$) cannot be extended to a BV-BFV theory (Definition \ref{def:BVBFV}). Since the degree-$0$ part of our AKSZ target will be the constrained space $\Fp{\Sigma}{PC}$, the extension to the cylinder\footnote{Equation \eqref{e:BFVstructuralConstraint} can be extended to the space of fields over a cylinder. We consider this point of view in Definition \ref{def:reducedPCfield}.} of Equation \eqref{e:BFVstructuralConstraint} provides an explicit realisation of (one of) the conditions that must be imposed on the bulk BV fields, in order to have a 1-extended BV-BFV theory. For $N=3$, the additional condition imposed by  \eqref{e:constraint} is void, explaining why 3d PC theory can be (fully) extended without additional requirements on the fields \cite{CaSc2019}. We will comment further on this in Section \ref{sec:Interpretation}.\end{remark}

\begin{remark}\label{rem:sympBFV}
The main difficulty in constructing BFV data for PC theory comes from the requirement that $(\mathcal{F}^\partial_{PC},\varpi^\partial_{PC})$ be a symplectic manifold, as opposed to pre-symplectic (cf. with \cite{AlkGri} and \cite{grigoriev2016presymplectic}). We stress that this requirement is essential for quantisation. To the best of our knowledge, a complete description of the \emph{symplectic} BFV structure for PC gravity was not available before \cite{CCS2020}. The complexity of the symplectic reduction arising in the classical description of the degree-$0$ boudary structure \cite{CS2019,CCS2020}, as well as the obstruction in the BV-BFV induction procedure \cite{CS2017}, are relevant features  peculiar to this formulation of gravity.
\end{remark}

\begin{remark}
The conventions we choose for the fields in \eqref{e:action_C3} and in \eqref{e:symplectic_form_C3} differ from those in \cite[Proposition 21]{CaSc2019}. In order to make contact between the formulas one has to perform the following change of variables:
\begin{align*}
e^{\dag} \mapsto y^{\dag} &\qquad \omega^{\dag} \mapsto c^{\dag} \\
c \mapsto - c &\qquad \xi \mapsto -\xi \qquad \xi^{n} \mapsto - \lambda.
\end{align*}
In the case $N \geq 4$ some signs differ from the ones in \cite{CCS2020} due to a sign convention for $\lambda$.
\end{remark}

}

\section{AKSZ EH}\label{s:AKSZEH}
We explore here the idea of reconstructing the $(d+1)$-dimensional BV extension of Einstein--Hilbert theory by means of the AKSZ construction, with target the BFV data for Einstein--Hilbert theory (as presented in Section \ref{s:BFVEH}, based on \cite{CS2016b}). In order to do this one looks at the space $\calF^{\AKSZ}_{EH}\coloneqq\mathrm{Maps}(T[1]I,\Fp{\Sigma}{EH})$, with $I$ an interval. In a chart, to consider the transgression map of Equation \eqref{transgressionmap1} means to look at fields composed of a 0-form and a 1-form on the interval $I$, with values in $\Fp{\Sigma}{EH}$ and fixed total degree. For the case at hand we will then have
\begin{subequations}\label{AKSZfields}\begin{eqnarray}
\mathfrak{G} = \bg(t) + \bpi^\dag(t) dt & \mathfrak{Z}^n = \xi^{n}(t) + \eta(t) dt & \mathfrak{Z}^a = \xi^{a}(t) + \beta^a(t) dt \\
\mathfrak{P} = \bpi(t) + \bg^\dag(t) dt & \mathfrak{H}_n = \bphi_n(t) + \xi^\dag_n(t) dt & \mathfrak{H}_a = \bphi_a(t) + \xi^\dag_a(t) dt
\end{eqnarray}\end{subequations}
where, for all $t\in I$, we parametrise $\calF^{\AKSZ}_{EH}$ with fields\footnote{The motivation for this particular choice of notation will be manifest very soon.} 
\begin{align*}
\bg(t)\in \mathrm{Map}(I,S_{nd}^2(T\Sigma)), \qquad & \bg^\dag(t)\in \mathrm{Map}(I,S^2[-1](T^*\Sigma)), \\
\bpi(t)\in \mathrm{Map}(I,S^2(T^*\Sigma)), \qquad & \bpi^\dag(t)\in \mathrm{Map}(I,S[-1]^2(T\Sigma)), \\
\eta(t), \beta^a(t)\in\mathrm{Map}(I,C^\infty(\Sigma)), \qquad &  \bphi_n(t),\bphi_a(t)\in \mathrm{Map}(I,\mathrm{Dens}[-1](\Sigma)) , \\
\xi^n(t), \xi^a(t)\in \mathrm{Map}(I,C^\infty[1](\Sigma)), \qquad &\xi^\dag_n(t),\xi^\dag_a(t)\in \mathrm{Map}(I,\mathrm{Dens}[-2](\Sigma)), 
\end{align*}
where we required $\gamma(t)$ to be nondegenerate for all $t\in I$. Now, observe that $\eta, \xi^n$ are functions on $\Sigma$ whereas $\beta^a, \xi^a$ can be considered as the components of vector (fields) tangent to $\Sigma$, which we will denote by $\beta$ and $\xi^\partial$. Similarly, we can promote $\bphi_a$ and $\xi^\dag_a$ into $\Sigma$-density valued one forms, which we will denote by $\bphi_\partial, \xi^\dag_\partial$. For simplicity of notation we will often use a unified index $\rho\in\{n,a\}$, so that $\bphi_\rho\in\{\bphi_n,\bphi_a\}$ and $\xi^\dag_\rho\in\{\xi^\dag_n,\xi^\dag_a\}$.

\begin{remark}
Notice once again that we are using (nondegenerate) sections of $S^2(TM)$ instead of actual metrics. Occasionally we will need to raise/lower indices using $\gamma$ and its ``inverse'' which we will denote by $\bg^{-1}$. See Remark \ref{rem:inversemetrics}.
\end{remark}

In what follows it will be useful to denote the \emph{Kinetic} part of the Hamiltonian functional (Equation \eqref{Hamiltonianconstraint}) as
\begin{equation}\label{Keq}
\mathcal{K}\coloneqq \frac{1}{\sqrt{\bg}}\left(\mathrm{Tr}_{\bg}[\bpi^2] - \frac{1}{d-1}\mathrm{Tr}_{\bg}[\bpi]^2\right),
\end{equation}
and the \emph{cosmological Einstein tensor} with respect to a metric $\bg^{-1}$ with cosmological constant $\Lambda$ will be
\begin{equation}\label{einsteintensor}
\mathbf{G}[\bg,\Lambda] = R[\bg]  + \left(\Lambda - \frac12 \mathrm{Tr}_{\bg} R[\bg]\right) \bg^{-1},
\end{equation}
where $R[\bg]$ is the Ricci--Riemann tensor of $\bg$. We also introduce a tensor-valued second order operator $\mathbf{D}_{\bg}$ that on functions $\phi\in C^\infty(\Sigma)$ acts as
\begin{equation}\label{BransDicketensor}
\mathbf{D}_{\bg}\phi = \gamma^{-1} \Delta^\partial \phi - \nabla^\partial\odot\nabla^\partial \phi
\end{equation}
where $\nabla^\partial$ denotes the Levi-Civita connection of $\bg$, and we denoted by $\Delta^\partial=\nabla^\partial\cdot \nabla^\partial$ the Laplace-Beltrami operator. In a coordinate chart this reads:
$$
[\mathbf{D}_{\bg}\phi]_{ab}={\bg}_{ab}{\bg}^{cd}\nabla^\partial_c\nabla^\partial_d\phi - \nabla^\partial_a\nabla^\partial_b\phi 
$$

In what follows we will also use the metric gradient, i.e. the vector field $\mathrm{grad}_\bg \phi = \bg(d\phi,\cdot)$. Since the covariant derivative $\nabla^\partial$ will not explicitly appear in what follows, we shall employ the symbol $\nabla$ to introduce a shorthand notation for the metric gradient: 
$$
\nabla_\bg\phi\equiv \frac12\mathrm{grad}_\bg \phi.
$$
Then, it is a matter of a straightforward calculation to show that
$$
\mathbf{D}_{\bg}\phi = \frac12 \left(\gamma^{-1} \mathrm{Tr}[\mathcal{L}_{\nabla_\bg \phi}\gamma] - \mathcal{L}_{\nabla_\bg \phi}\gamma^{-1}\right).
$$
 
With this in mind we can state the following:

\begin{theorem}\label{th:AKSZEH}
The AKSZ data $\FF_{EH}^{\AKSZ}(I;\FF^\partial_{EH}(\Sigma))$ are given by the $(-1)$-shifted symplectic manifold 
{\small
\begin{align*}
    \calF^{\AKSZ}_{EH}\simeq T^*[-1]\left(\mathrm{Map}(I,S_{nd}^2(T\Sigma) \times S^2(T^*\Sigma)\times C^\infty(\Sigma)\times\mathfrak{X}(\Sigma)\times\mathfrak{X}[1](\Sigma)\times C^\infty[1](\Sigma))\right)
\end{align*}
}%
\begin{equation}\label{AKSZBVFORM2}
\Omega^{\AKSZ}_{EH}
	= 
	\int\limits_{\Sigma\times I} \left\{ - \langle\delta\bg, \delta\bg^\dag\rangle + \langle\delta\bpi, \delta\bpi^\dag\rangle 
		+ \delta\xi^{\rho}\delta\xi^\dag_\rho + \delta\eta\delta\bphi_n + \delta\beta^a\delta\bphi_a\right\} dt
\end{equation}
and the AKSZ action functional : 
\begin{subequations}\begin{align}
S^{\AKSZ}_{EH}
	&= \intl_{\Sigma \times I} \Big\{\langle \bpi, \dot{\bg}\rangle   -  \langle\bphi_\rho, \dot{\xi}^\rho\rangle +H_n \eta + H_\partial(\beta) - \langle  \bg^\dag, L_{\xi^{\partial}}\bg\rangle  + \langle \bpi^\dag, L_{\xi^{\partial}}\bpi\rangle\\
	&    - \left(\pard{\mathcal{K}}{\bg}(\bpi^\dag) + \pard{\mathcal{K}}{\bpi}(\bg^\dag)\right) \xi^n  - \sqrt{\bg}\langle \bpi^\dag, \mathbf{G}[\bg,\lambda] \xi^n + \mathbf{D}_{\bg}(\xi^n) \rangle   \\
	&    + \langle \bphi_\partial, \nabla_\gamma\eta\xi^{n} - \eta\nabla_\gamma\xi^{n} + L_{\xi^\partial} \beta \rangle   - \bphi_n L_{\beta}\xi^n +\bphi_n L_{\xi^\partial}\eta \\
	& + \left\langle\xi^\dag_\partial, \frac12 [\xi^{\partial},\xi^{\partial}] 
		+ \xi^{n}\nabla_\gamma\xi^{n} \right\rangle   + \bpi^\dag(\bphi_\partial,d\xi^{n})\xi^{n} + \xi_n^\dag L_{\xi^{\partial}}\xi^{n} \Big\} dt,
\end{align}\end{subequations}
together with its Hamiltonian vector field $Q^{\AKSZ}_{EH}$.
\end{theorem}

\begin{proof}
The prescription of Theorem \ref{AKSZtheorem}, suggests that to construct the data in $\FF^{\AKSZ}(I;\FF^\partial_{EH})$ we need to compute
\begin{equation}
\Omega^{\AKSZ}_{EH}
	= \T_I^{(2)}\varp{\Sigma}{EH} 
		= \int\limits_{\Sigma\times I} \langle \delta\mathfrak{P}, \delta\mathfrak{G} \rangle 
		+ \langle \delta\mathfrak{H}_\rho, \delta\mathfrak{Z}^\rho\rangle. 
\end{equation}
By selecting the top-form part of the integrand and observing that $|dt|=1$ we get
\begin{equation}
\Omega^{\AKSZ}_{EH}= \intl_{\Sigma\times I} \left\{ 
	- \langle\delta\bg^\dag, \delta\bg\rangle 
	+ \langle \delta\bpi, \delta\bpi^\dag \rangle
	+ \delta\xi^\dag_\rho \delta\xi^{\rho} 
	+ \delta\bphi_n \delta\eta
	+ \delta\bphi_a\delta\beta^a
	\right\} dt
\end{equation}
where the sign comes from $\langle \delta(\bg^\dag dt), \delta \bg \rangle = - \langle \delta\bg^\dag, \delta \bg \rangle dt$, since $|\delta\bg|=1$, while $\delta\xi^\dag_\rho dt \delta\xi^\rho = \delta\xi^\dag_\rho \delta\xi^\rho dt$, since $|\delta\xi^\rho|=2$. $\Omega^{\AKSZ}$ is a $(-1)$-symplectic structure on $\mathrm{Maps}(T[1]I,\Fp{\Sigma}{EH})$, a BV $2$-form.

Now, from $\alp{\Sigma}{EH}$ we can construct a degree-$0$ functional on $\calF^{\AKSZ}_{EH}$ by first applying the transgression map, which yields the $1$-form 
$$\T^{(1)}_I\alpha^\partial\in\Omega^1(\mathrm{Maps}(T[1]I,\Fp{\Sigma}{EH})),$$ 
and then contracting it with the de Rham differential on $I$ seen as an odd cohomological vector field $d_I$. In a local chart this is tantamount to replacing $\delta \leadsto d_I:= dt\frac{d}{dt}$, so that
\begin{equation}
\iota_{d_I}\T^{(1)}_I\alp{\Sigma}{EH} = \intl_{\Sigma\times I} \left\{\langle\bpi, \dot{\bg}\rangle - \langle\bphi_\rho, \dot{\xi}^\rho\rangle\right\} dt.
\end{equation}
where the sign comes from the fact that $\langle{\bphi_\rho} dt,\,\dot{\xi}^\rho\rangle = -\langle{\bphi_\rho}, \dot{\xi}^\rho\rangle dt$.
Finally, we want to compute the AKSZ action functional 
$$S^{\AKSZ}_{EH}\coloneqq\T_I^{(0)}(\Sp{\Sigma}{EH}) + \iota_{d_{I}}\T_I^{(1)}(\alp{\Sigma}{EH}).$$
This calculation is completely analogous to the previous ones, and it is mostly straightforward. One needs to pay attention to the signs, so it is worthwhile to stress that
$$
\bpi^\dag dt\, (\bphi_\partial, d\xi^n) \xi^n = [\bpi^\dag]^{ab} dt \bphi_a \partial_b \xi^n \xi^n = [\bpi^\dag]^{ab} \bphi_a \partial_b \xi^n \xi^n dt = \bpi^\dag (\bphi_\partial, d\xi^n) \xi^n dt
$$
while
{\small
\begin{align*}
  -\bg(\xi^\dag_\partial dt, d\xi^n)\xi^n - \bg(\bphi_\partial, d \eta\, dt\,\xi^n + d \xi^n\eta dt) = \left(\langle\xi^\dag_\partial, \xi^n \nabla_{\bg}\xi^n\rangle + \langle\bphi_\partial , \nabla_{\bg}\eta \xi^n - \eta \nabla_{\bg}\xi^n\rangle\right)dt  
\end{align*}
}%
Finally, at first order in $dt$,
$$
H_n(\bg + \bpi^\dag dt, \bpi + \bg^\dag) \xi^n =  \pard{(H_n\xi^n)}{\bg}(\bpi^\dag dt) + \pard{(H_n \xi^n)}{\bpi}(\bg^\dag dt ),
$$ 
and $dt\, \xi^n = - \xi^n dt$. We write the formulas above as derivatives of the functional $H_n\xi^n$ to stress that total derivatives will appear, due to the term $R^\partial$ in $H_n$. Recalling the expression for $H_n$ of equation \eqref{Hamiltonianconstraint} and the definition of $\mathcal{K}$, $\mathbf{G}[\gamma,\Lambda]$ and $\mathbf{D}_\gamma$ of Equations \eqref{Keq},\eqref{einsteintensor} and \eqref{BransDicketensor}, the variation of $H_n \xi^n$ with respect to $\bg$ yields 
$$
\pard{(H_n\xi^n)}{\bg} = \pard{\mathcal{K}}{\bg}\xi^n + \pard{(\sqrt{\bg}R^\partial\xi^n)}{\bg} = \pard{\mathcal{K}}{\bg}\xi^n + \sqrt{\bg}\left(\mathbf{G}[\bg,\Lambda] \xi^n + \mathbf{D}_{\bg}(\xi^n)\right) + d(\dots).
$$
The total derivative term is exact with respect to the tangent differential $d_\Sigma$. It can be discarded, provided $\Sigma$ has no boundary (which we are assuming throughout), so:
\begin{multline}
H_n(\bg + \bpi^\dag dt, \bpi + \bg^\dag) \xi^n = \\
	-\left( \pard{\mathcal{K}}{\bpi}(\bg^\dag) + \pard{\mathcal{K}}{\bg}(\bpi^\dag)\right)\xi^n dt 
	- \sqrt{\bg}\left\langle \bpi^\dag,\mathbf{G}[\bg,\Lambda] \xi^n + \mathbf{D}_{\bg}(\xi^n)\right\rangle dt.
\end{multline}
\end{proof}

\begin{remark}
In order to compute the cohomological vector field $Q^{\AKSZ}_{EH}$ we enforce the Hamiltonian condition $\iota_{Q^{\AKSZ}_{EH}} \Omega^{\AKSZ}_{EH}= \delta S^{\AKSZ}_{EH}$ dropping all possible boundary terms. It  reads (we omit the expression for $Q^{\AKSZ}_{EH}\xi^\dag$ and $Q^{\AKSZ}_{EH}\bphi$):
\begin{subequations}\label{QAKSZEH}\begin{align}
Q^{\AKSZ}_{EH} \bg & =  \pard{H_n}{\bpi}\xi^n  + L_{\xi^\partial}\gamma\\
Q^{\AKSZ}_{EH} \bpi &= - \pard{ \mathcal{K}}{\bg}\xi^n - \sqrt{\bg}\left( \mathbf{G}[\bg,\Lambda] \xi^n  + \mathbf{D}_{\bg}(\xi^n)\right) + L_{\xi^\partial}\bpi {  + } \bphi \odot d\xi^n \xi^n\\
Q^{\AKSZ}_{EH} \eta& = {  - }\dot{\xi}{}^n + L_{\xi^{\partial}}\eta - L_\beta\xi^n\\
Q^{\AKSZ}_{EH} \beta & = {  - } \dot{\xi}{}^\partial + L_{\xi^\partial}\beta + \nabla_{\bg} \eta\,\xi^n - \eta\, \nabla_{\bg}\xi^n {  - } \nabla_{\bpi^\dag}\xi^n \xi^n\\
Q^{\AKSZ}_{EH} \xi^\partial &= \frac12[\xi^\partial,\xi^\partial] + \xi^n\nabla_{\bg}\xi^n\\
Q^{\AKSZ}_{EH} \xi^n & = L_{\xi^\partial}\xi^n \\
Q^{\AKSZ}_{EH}\bg^\dag & = \dot{\bpi} + \pard{ \mathcal{K}}{\bg}\eta + \sqrt{\bg}\left(\mathbf{G}[\bg,\Lambda] \eta + \mathbf{D}_{\bg}(\eta)\right) + L_\beta\bpi + L_{\xi^\partial} \bg^\dag  \\
	& + \xi^\dag_\partial \odot d\xi^n\xi^n -\bphi_\partial \odot d\eta \xi^n { +} \eta \bphi_\partial \odot d\xi^n \\
	& + \left[\frac{\delta^2\mathcal{K}}{\delta\bg^2}(\bpi^\dag) 
		+ \frac{\delta^2H_n}{\delta\bg\delta\bpi}(\bg^\dag)\right]\xi^n
		- \frac12\bg^{-1} \langle \bpi^\dag, \mathbf{G}[\bg,\lambda] \xi^n + \mathbf{D}_{\bg}(\xi^n) \rangle\\
		&- \sqrt{\bg}\left\langle \bpi^\dag, \pard{\mathbf{G}[\bg,\lambda]}{\bg} \xi^n + \pard{\mathbf{D}_{\bg}(\xi^n)}{\bg}\right\rangle\\ \label{Pishell}
Q^{\AKSZ}_{EH}\bpi^\dag & =\dot{\bg} + \pard{\mathcal{K}}{\bpi}\eta + L_\beta \bg + L_{\xi^\partial} \bpi^\dag 
	{ -} \left[\frac{\delta^2\mathcal{K}}{\delta\bpi^2}(\bg^\dag) + \frac{\delta^2\mathcal{K}}{\delta\bg\delta\bpi}(\bpi^\dag)\right]\xi^n.
\end{align}\end{subequations}
\end{remark}

\begin{remark}
We notice that the term $\sqrt{\bg}\mathbf{D}_{\bg}(\cdot)$ is the contribution to the field equations for a metric due to the presence of a Brans--Dicke ``dilaton'' field, whose role is played by the ghost $\xi^n$ in the BFV action $\Sp{\Sigma}{EH}$ and by $\eta$ in $S^{\AKSZ}_{EH}$.
\end{remark}

\subsection{Pushforward}\label{s:pushforward}
We would like to compute the BV pushforward of $\FF^{\AKSZ}$ along the symplectic submanifold $(\bpi,\bpi^\dag)\in\calF''=T^*[-1]\mathrm{Map}(I,S^2(T^*\Sigma))\subset \calF^{\AKSZ}_{EH}$.

\begin{remark}\label{effectivepushforward}
This is the same as evaluating $S_{\text{eff}}$ from equation \eqref{Seffdef}. Since $S^{\AKSZ}_{EH}$ is only quadratic in $\Pi$, the calculation reduces to computing the Batalin--Vilkovisky--Legendre transform $S_{\text{BVL}}$ of $S^{\AKSZ}_{EH}$ with respect to $\mathcal{L}$, as in Definition \ref{def:BVL}, plus a correction in the integration measure for the remaining (second-order) effective theory. Note that Equation \ref{BVL} is equivalent to setting to zero the r.h.s. of Equation \eqref{Pishell}, together with $\Pi^\dag=0$.
\end{remark}

Recall that we are assuming $\gamma(t)$ to be a nondegenerate section of $S^2(T\Sigma)$ for every $t$, i.e. it represents the inverse of a metric, and dually $\bpi(t)\in S^2(T^*\Sigma)$. We can use $\bg$ and its inverse (denoted $\bg^\flat$) to raise/lower indices: explicitly, if $\bg=\gamma^{ab}\partial_a\odot\partial_b$, we have $\bg^\flat=\gamma_{ab}dx^a\odot dx^b$, with $\gamma^{ab}\gamma_{bc}=\delta^a_c$. Then, for $X\in S^2(T^*\Sigma),\ Y\in S^2(T\Sigma)$ we define $(X^\sharp)^{ab} := \bg^{ac}\bg^{bd}X_{bc}$ and $(Y^\flat)_{ab} = \bg_{ac}\bg_{bd}Y^{cd}$.  

\begin{definition}
Consider the space of fields $\uF{\Sigma\times I}{}\subseteq \calF^{\AKSZ}$ as
\begin{equation}
    \uF{\Sigma\times I}{} \coloneqq T^*[-1]\left(\mathrm{Map}\left(I,S_{nd}^2(T\Sigma) \times T[1]\left(C^\infty(\Sigma)\times\mathfrak{X}(\Sigma)\right)\right)\right)
\end{equation}
parametrised by $(\gamma,\eta,\beta,\xi^n,\xi^\partial,\varphi_n,\varphi_\partial,\xi^\dag_n,\xi^\dag_\partial)$, with $\iota_{EH}\colon \uF{\Sigma\times I}{}\to \calF^{AKSZ}$ the inclusion map.
\end{definition}

\begin{theorem}\label{th:pushforward}
The BV-pushforward of $\FF^{\AKSZ}(I;\FF^\partial_{EH}(\Sigma))$ with respect to the Lagrangian submanifold $\mathcal{L}=\{(\bpi,\bpi^\dag)\in\calF''\ |\ \bpi^\dag=0\}$ is the BV theory given by 
$$\uFF{\Sigma\times I}\coloneqq(\uF{\Sigma\times I},\uS{\Sigma\times I}{}, \uV{\Sigma\times I}{})$$%,\uQ{\Sigma\times I}{}
where 
\begin{align}
\uS{\Sigma\times I}{}&=
	\intl_{\mathbb{R}}dt\intl_{\Sigma} 
		-\eta\sqrt{\bg}\left[\left( \langle K^\sharp, K\rangle  - \mathrm{Tr}(K)^2\right) 
		+ R^\partial - 2\Lambda\right] \\\notag
	& -\langle\bg^\dag, L_{\xi^\partial}\bg\rangle 
		- 2\langle K^\sharp, \bg^\dag\rangle \xi^n -\langle\bphi_\rho, \dot{\xi}^\rho\rangle \\\notag
	& + \langle \bphi_\partial, \nabla_\gamma\eta\xi^{n} 
		- \eta\nabla_\gamma\xi^{n} + L_{\xi^\partial} \beta \rangle  
		+\bphi_n \left(- L_{\beta}\xi^n +L_{\xi^\partial}\eta\right) \\\notag
	& + \left\langle\xi^\dag_\partial,  \xi^{n}\nabla_\gamma\xi^{n}  
		+ \frac12 [\xi^{\partial},\xi^{\partial}] \right\rangle  + \xi_n^\dag L_{\xi^{\partial}}\xi^{n}
\end{align}
with $K := \frac{\eta^{-1}}{2} \left(\dot{\bg} + L_\beta \bg\right)^\flat$, and
\begin{equation}
\uV{\Sigma\times I}{}=\iota_{EH}^*\Omega^{\AKSZ}_{EH}.
\end{equation}
\end{theorem}
\begin{proof}
As discussed in Remark \ref{effectivepushforward}, we are interested in finding the effective action one obtains by means of the perturbative expansion of the integral
\begin{equation}\label{e:pathintegralprop}
\exp\left(\frac{i}{\hbar}S_{\text{eff}}\right)\coloneqq\intl_{\mathcal{L}\subset\calF''} \exp\left(\frac{i}{\hbar}S^{\AKSZ}_{EH}\right)
\end{equation}
because $S^{\AKSZ}_{EH}\vert_{\mathcal{L}}$ is quadratic in $\Pi$, through the term $\mathcal{K}(\Pi)\eta$. Observing that 
\begin{align}
\pard{\mathcal{K}}{\bpi} & = \frac{2}{\sqrt{\bg}}\left(\bpi^\sharp - \frac{\bg}{d-1} \mathrm{Tr}(\bpi)\right)\\
\frac{\delta^2 \mathcal{K}}{\delta\bpi^2}(\bg^\dag) & = \frac{2}{\sqrt{\bg}}\left((\bg^\dag)^\sharp - \frac{\bg}{d-1} \mathrm{Tr}(\bg^\dag)\right),
\end{align}
we have that Equation \eqref{Pishell} reads
{\small
\begin{equation}
\frac{2}{\sqrt{\bg}}\left(\bpi^\sharp - \frac{\bg}{d-1} \mathrm{Tr}(\bpi)\right) = -\eta^{-1}\left(\dot{\bg} + L_\beta\bg { -} \frac{2}{\sqrt{\bg}}\left((\bg^\dag)^\sharp - \frac{\bg}{d-1} \mathrm{Tr}(\bg^\dag)\right)\xi^n\right) + F(\bpi^\dag)
\end{equation}}
where $\mathrm{Tr}(X) = \bg^{ab}X_{ab}$. We will use the symbol $\approx$ to denote the enforcing of Equation \eqref{Pishell} and of $\bpi^\dag=0$. Then, requiring $\bpi^\dag=0$ and defining
\begin{equation}\label{KdefAKS}
K := \frac{\eta^{-1}}{2} \left(\dot{\bg} + L_\beta \bg\right)^\flat, 
\end{equation}
we obtain that
\begin{equation}
\bpi\approx - \sqrt{\bg}\left( K - \mathrm{Tr}(K) \bg^\flat \right) + \eta^{-1} \bg^\dag \xi^n.
\end{equation}

It is easy to compute now
\begin{equation}
H_n \approx \sqrt{\bg}\left[\left( \langle K^\sharp, K\rangle  - \mathrm{Tr}(K)^2\right) + R^\partial - 2\Lambda\right] - 2\eta^{-1}\langle K^\sharp, \bg^\dag\rangle \xi^n
\end{equation}
which, together with
\begin{equation}
\langle \bpi, \dot{\bg} + L_\beta \bg\rangle \approx -2 \sqrt{\bg}\left[\left( \langle K^\sharp, K\rangle  - \mathrm{Tr}(K)^2\right) + R^\partial - 2\Lambda\right] + 2\eta^{-1}\langle K^\sharp, \bg^\dag\rangle \xi^n
\end{equation}
and
\begin{equation}
-\frac{\delta H_n}{\delta\bpi}(\bg^\dag)\xi^n \approx + 2\langle K^\sharp, \bg^\dag\rangle \xi^n,
\end{equation}
yields
\begin{align}\label{OnshellAKSZEH}
S^{\AKSZ}_{EH}
	&\approx \intl_{\mathbb{R}}dt\intl_{\Sigma} -\eta\sqrt{\bg}\left[\left( \langle K^\sharp, K\rangle  - \mathrm{Tr}(K)^2\right) + R^\partial - 2\Lambda\right] \\\notag
	& -\langle\bg^\dag, L_{\xi^\partial}\bg\rangle + 2\langle K^\sharp, \bg^\dag\rangle \xi^n -\langle\bphi_\rho, \dot{\xi}^\rho\rangle \\\notag
	& + \langle \bphi_\partial, \nabla_\gamma\eta\xi^{n} - \eta\nabla_\gamma\xi^{n} + L_{\xi^\partial} \beta \rangle  +\bphi_n \left(- L_{\beta}\xi^n +L_{\xi^\partial}\eta\right) \\\notag
	& + \left\langle\xi^\dag_\partial,  \xi^{n}\nabla_\gamma\xi^{n}  + \frac12 [\xi^{\partial},\xi^{\partial}] \right\rangle  + \xi_n^\dag L_{\xi^{\partial}}\xi^{n} =:\uS{\Sigma\times I}{}.
\end{align}
So, formula \eqref{OnshellAKSZEH} shows that $S_{\text{eff}} = \uS{\Sigma\times I}{} + O(\hbar)$. The $\hbar$ correction is the (logarithm of the) determinant of the operator defining the quadratic form $\mathcal{K}$, i.e. the determinant of the deWitt super metric\footnote{To be precise, $\mathsf{W}_\gamma$ is the inverse of the metric introduced by deWitt, due to our choice of working with inverse metrics $\gamma$.} \cite{deWitt}
$$
    \mathsf{W}_\gamma^{ijkl} = \frac{1}{\sqrt{\gamma}}\left(\gamma^{ik}\gamma^{jl} - \frac{1}{d-1} \gamma^{ij}\gamma^{kl}\right),
$$
or, more invariantly
\begin{equation}
    \mathsf{W}_\gamma(\Pi,\Pi) = \frac{1}{\sqrt{\gamma}} \left(\langle \Pi^\sharp,\Pi \rangle - \frac{1}{d-1}\mathrm{Tr}_\gamma[\Pi]^2\right),
\end{equation}
and it will have the effect of correcting the overall measure on the residual BV space of fields $\uF{\Sigma\times I}{}$.
\end{proof}

\begin{remark}
Up to boundary, we can compute $\uQ{\Sigma\times I}{}$ (denoted hereinafter by $\UQ$) to be:
\begin{subequations}\label{algebroidstructure}
\begin{align}
    \UQ\gamma & = L_{\xi^\partial}\gamma - \eta^{-1}(\dot{\gamma} + L_\beta\gamma) \xi^n\label{algebroida}\\
    \UQ\eta & = - L_{\beta}\xi^n +L_{\xi^\partial}\eta \label{algebroidb}\\
    \UQ\beta & = \nabla_{\gamma}\eta\xi^{n} 
		- \eta\nabla_{\gamma}\xi^{n} + L_{\xi^\partial} \beta\label{algebroidc}\\
	\UQ\xi^n &= L_{\xi^{\partial}}\xi^{n}\label{algebroidd}\\
	\UQ\xi^\partial &= \xi^{n}\nabla_{\gamma}\xi^{n}  
		+ \frac12 [\xi^{\partial},\xi^{\partial}] \label{algebroide}
\end{align}
\end{subequations}
and similarly for antifields.
\end{remark}

\subsection{Reconstruction of Einstein--Hilbert theory}
In this section we wish to show that the BV pushforward of the AKSZ theory constructed in Section \ref{s:pushforward} is strongly equivalent to Einstein--Hilbert theory in the BV formalism.

To do this, we begin by considering the following definitions:
\begin{subequations}\label{reconstruct}\begin{align}
\txi^{} &= -\eta^{-1}\xi^n (\partial_t + \beta) + \xi^\partial \\\label{greconstruct}
\tg &= -\eta^{-2} \partial_t\odot\partial_t - 2 \eta^{-2} \beta\odot\partial_t+ \bg - \eta^{-2} \beta\odot\beta
\end{align}\end{subequations}

\begin{lemma}\label{lem:algebroid}
We have the following relations
\begin{subequations}\begin{align}
\frac12[\txi^{},\txi^{}] & = \UQ \txi^{},\\
L_{\txi^{}}\tg & = \UQ \tg.
\end{align}\end{subequations}
\end{lemma}

\begin{proof}
It is a straightforward calculation to show
\begin{align*}
    \UQ\txi^{} & = \eta^{-2} (\UQ\eta)  \xi^{n}(\partial_t + \beta) + \eta^{-1}\UQ\xi^n (\partial_t + \beta) + \UQ\xi^\partial + \eta^{-1}\xi^n \UQ\beta\\
	& = \left(- \eta^{-2} \dot{\xi}{}^n   \xi^{n}  - \eta^{-2}L_\beta\xi^{n}  \xi^{n} + L_{\xi^\partial} (-\eta^{-1} \xi^{n}) \right)(\partial_t + \beta) \\
	& \quad - \eta^{-1}\xi^n \dot{\xi}{}^\partial + \eta^{-1}\xi^n L_{\xi^\partial}\beta + \frac12[\xi^\partial,\xi^\partial].
\end{align*}
Observe that the ``algebroid term'' (see Remark \ref{rem:algebroid}, below) $\xi^n\nabla_{\bg}\xi^n$ in $\UQ{\xi^\partial}$ cancels out with part of $\eta^{-1}\xi^n \UQ\beta$. On the other hand this coincides with 
\begin{align*}
\frac12[\txi^{},\txi^{}] & = \frac12 [ -\eta^{-1}\xi^n(\partial_t + \beta) + \xi^\partial, - \eta^{-1}\xi^n(\partial_t + \beta) + \xi^\partial] \\
& =  \left(\eta^{-2}\xi^n(\partial_t + \beta)\xi^n + L_{\xi^\partial} (-\eta^{-1}\xi^n)\right)(\partial_t + \beta)\\
	& \quad - \eta^{-1}\xi^n\dot{\xi}{}^\partial +\eta^{-1}\xi^nL_{\xi^\partial}\beta + \frac12[\xi^\partial,\xi^\partial],
\end{align*}
proving the first claim.
We compute
{%\small
\begin{align*}
L_{\txi^{}}\tg & = 
    - 2\eta^{-3}\dot{\xi}{}^n\partial_t\partial_t 
	+ 2\eta^{-3}L_{\xi^\partial}\eta\partial_t\partial_t 
	+ 2\eta^{-2}\dot{\xi}{}\partial_t 
	- 2\eta^{-4}\xi^nL_\beta\eta\partial_t\partial_t \\
& \quad 
	- 2\eta^{-2}\partial_t\left(\eta^{-1}\xi^n\beta\right)\partial_t
    -4\eta^{-4}\xi^n\partial_t\eta\beta\partial_t  
	+ 4\eta^{-3} L_{\xi^\partial}\eta\beta\partial_t 
	- 4\eta^{-4} L_\beta\eta\xi^n \beta\partial_t  \\
& \quad
	- 2\eta^{-2} L_{\xi^\partial}\beta\partial_t 
	- 2\eta^{-2}L_\beta\left(\eta^{-1}\xi^n\right)\partial_t\partial_t  
    + 2\eta^{-2}\partial_t\left(\xi^\partial 
    - \eta^{-1}\xi^n\beta\right)\beta \\
& \quad 
    - 2\eta^{-2}L_\beta\left(\eta^{-1}{\xi^n}{}\right)\beta\partial_t 
	- 2\eta^{-2}\partial_t\left( \eta^{-1}{\xi^n}{}\right) \beta\partial_t 
    + 2\eta^{-2}\xi^n \dot{\beta}\partial_t  
    - \eta^{-1}{\xi^n}{}\dot\gamma^{ab}\partial_a\partial_b \\
& \quad
	+ L_{\xi^\partial}\left(\gamma^{ab}\partial_a\partial_b\right) 
	- \eta^{-1}{\xi^n}{} L_\beta(\gamma^{ab})\partial_a\partial_b 
	+ 2\nabla_{\bg}\left(\eta^{-1}{\xi^n}{}\right) \partial_t 
	+ 2\nabla_{\bg}\left(\eta^{-1}{\xi^n\beta}{}\right)\partial_c \\
	& \quad 
-2\eta^{-4}\xi^n \dot\eta \beta\beta 
	+ 2\eta^{-3}L_{\xi}\eta \beta\beta
	- 2\eta^{-4}\xi^nL_\beta\eta \beta\beta
	+ 2\eta^{-3} \xi^n\dot\beta\beta
	-\eta^{-2} L_{\xi^\partial}\left(\beta\beta\right) \\
	& \quad 
	- 2\eta^{-2}L_\beta\left(\eta^{-1}{\xi^n}{}\right)(\beta\partial_t +\beta\beta)
\end{align*}
}where we recall that expressions like $L_{\xi^\partial}(\beta)$ denote the Lie derivative of the vector field $\beta=\beta^a\partial_a$ along $\xi^\partial$. On the other hand we have
\begin{align*}
\UQ(-\eta^{-2})\partial_t\partial_t 
	&= \left( - 2 \eta^{-3} \dot{\xi}{}^n 
		+ 2 \eta^{-3}L_{\xi^\partial}\eta 
		- 2\eta^{-3}L_\beta\xi^n\right)\partial_t\partial_t\\
\UQ(-2\eta^{-2}\beta)\partial_t 
	& = \left(-4\eta^{-3}\dot{\xi}{}^n\beta 
		+ 4\eta^{-3} L_{\xi^\partial}\eta 
		- 4\eta^{-3} L_\beta\xi^n\beta\right)\partial_t \\
	& +\left( 2\eta^{-2}\dot{\xi}{}^\partial 
		- 2\eta^{-2} L_{\xi^\partial}\beta 
		- 2\eta^{-2}\nabla_{\bg}\eta\xi^n 
		+ 2\eta^{-1} \nabla_{\bg}\xi^n\right)\partial_t\\
\UQ(\gamma^{ab})\partial_a\partial_b 
	&= -\eta^{-1}\xi^n \dot{\bg} 
		- \eta^{-1}\xi^nL_\beta\bg 
		+ L_{\xi^\partial}\bg\\
\UQ(-\eta^{-2}\beta\beta)
	&= - 2\eta^{-3}\dot{\xi}{}^n \beta\beta 
		+ 2\eta^{-3} L_{\xi^\partial}\eta\beta\beta 
		- 2\eta^{-3}L_\beta\xi^n\beta\beta\\
	& + 2\eta^{-2}\beta\dot{\xi}^\partial 
		- 2\eta^{-2} L_{\xi^\partial}\beta \beta
		- 2\eta^{-2}\left(\nabla_{\bg}\eta\xi^n - \eta\nabla_{\bg}\xi^n\right)\beta
\end{align*}
And it is a matter of a straightforward, but lengthy computation to show that the two expressions coincide. Indeed, subtracting one from the other we obtain
\begin{multline}
L_{\txi^{}}\tg - \UQ\tg = 2 \eta^{-3}(-\eta^{-1} \xi^n ( \dot{\eta} + L_\beta(\eta)) \beta\beta + 2\eta^{-3}\xi^n\dot{\beta}\beta - 2\eta^{-2}L_\beta(\eta^{-1})\xi^n (\beta\partial_t + \beta\beta) \\
- 2\eta^{-4}\xi^n L_\beta(\eta) \partial_t^2 + 2\eta^{-4}\dot{\eta}\xi^n \beta\partial_t - 2\eta^{-3}\xi^n \dot{\beta}\xi^n \partial_t - 4\eta^{-4}\xi^n (\dot{\eta} + L_\beta\eta) \beta\partial_t\\
+ 2\eta^{-3}\xi^n \dot{\beta} \partial_t + 2\eta^{-4} L_\beta\eta(\partial_t^2 + \beta\partial_t) + 2\eta^{-4}\dot{\eta}\xi^n(\beta\partial_t + \beta\beta) - 2\eta^{-3}\xi^n\dot{\beta}\beta \equiv 0
\end{multline}

\end{proof}

\begin{remark}\label{rem:algebroid}
Using Lemma \ref{lem:algebroid} we wish to interpret \eqref{reconstruct} as a map of Lie algebroids. Consider the (trivial) vector bundle over 
$$\mathrm{Map}(I,S_{nd}^2(T\Sigma) \times C^\infty(\Sigma)\times\mathfrak{X}(\Sigma))\simeq \mathcal{PR}^\Sigma(\Sigma\times I),$$ 
where $\mathcal{PR}^{\Sigma}(\Sigma\times I)$ denotes pseudo-Riemennian metrics on $\Sigma\times I$ such that their restriction to $\Sigma$ is nondegenerate, with fibre 
$$\mathrm{Map}(I,C^\infty(\Sigma) \times \mathfrak{X}(\Sigma))\simeq \mathfrak{X}(\Sigma\times I).$$
We  consider two different Lie  algebroid  structures on this vector  bundle. One is the action algebroid with bracket given by the bracket of $(d+1)$-vector fields, and anchor given by Lie derivatives on metrics. The other algebroid structure is given by formulas \eqref{algebroidstructure}, with \eqref{algebroida},\eqref{algebroidb} and \eqref{algebroidc} defining the anchor map, and \eqref{algebroidd} and \eqref{algebroide} specifying the bracket of sections. Observe that the morphism of algebroids \eqref{reconstruct} does not preserve constant sections, as the splitting of a generic vector field $\txi$ depends on the so-called \emph{lapse} $\eta$ and \emph{shift} $\beta$, which are coordinates on the base of the fibre bundle.  The latter algebroid encodes the algebraic relations of the constraints of Einstein--Hilbert theory\footnote{We stress that, as it is, the structure one can extract from the BFV differential $Q^\partial$ is that of a curved $L_\infty$ algebroid, due to the dependency on fields of negative degree. We thank A. Weinstein, C. Blohmann and N. L. Delgado for enlightening discussions on this matter.}, and was carefully studied by other means in \cite{BFW}. It was also mentioned as a motivating example for the notion of Hamiltonian Lie Algebroid, introduced in \cite{BW}. It is an interesting question to check whether this construction satisfies the Hamiltonian requirements for an algebroid.
\end{remark}

To proceed, we need to recall the BV data associated with Einstein--Hilbert theory, in the ADM formalism. Given a pseudo-Riemannian (inverse) metric $\tg$ on a manifold $M$, we can perform a $d+1$ decomposition and rewrite it as\footnote{In this paper we will assume that the manifold $M$ has a global product structure $M=\Sigma\times \mathrm{R}$, and the induced metric on $\Sigma$ will be Riemannian, i.e. the leaves $\Sigma_t$ are spacelike submanifolds of $M$. It is straightforward to generalise this to the \emph{timelike} case. The relevant formulas for EH theory in the BV-BFV formalism have been given in \cite{CS2016b}.}
$$
\begin{array}{c}
\tg^{\mu\nu}=\left(\begin{array}{cc} -\eta^{-2} &  -\eta^{-2}\beta^b \\ -\eta^{-2}\beta^a & \gamma^{ab}-\eta^{-2}\beta^a\beta^b\end{array}\right)
\end{array}
$$
In the case where $M$ has a boundary, we can define the second fundamental form of the boundary submanifold $K_{ab}$ and its trace $K$ by means of the boundary covariant derivative $\nabla^\partial$ (the Levi-Civita connection of $\gamma$) as follows
\begin{equation}\label{KdefADM}
K_{ab}=\frac12 \eta^{-1}(2\nabla^\partial_{(a}\beta_{b)} + \partial_t\gamma_{ab}) \qquad K=\gamma^{ab}K_{ab}
\end{equation}
where $t$ denotes a coordinate transverse to the boundary $\partial M$. Finally, notice that
$$
(L_\beta\gamma)^{cd}\gamma_{ac}\gamma_{bd} = - 2\nabla^\partial_{(a}\beta_{b)} \qquad (\dot{\gamma})^{cd}\gamma_{ac}\gamma_{bd} = - \partial_t\gamma_{ab}.
$$

\begin{definition}
Let $(\calF_{EH}(M), \Omega_{EH}(M))$ be the symplectic manifold
$$
\calF_{EH}(M) \coloneqq T^*[-1]\left( \mathcal{PR}^{\partial M}(M) \times \mathfrak{X}[1](M)\right)
$$
with its canonical symplectic structure, and $\mathcal{PR}^{\partial M}(M)$ denotes pseudo-Riemennian metrics on $M$ such that their restriction to $\partial M$ is nondegenerate. Consider the functional 
\begin{equation}\label{BVEHaction}
S_{EH}(M) = \intl_{M} \Big\{-\eta\sqrt{\gamma}(\epsilon(K_{ab}K^{ab} - K^2) +R^\partial -2\Lambda)\Big\} + \tg^\dag L_{\txi} \tg + \frac{1}{2}\iota_{[\txi,\txi]}\txi^{\dag}
\end{equation}
and denote by $Q_{EH}(M)$ the Hamiltonian vector field of $S_{EH}(M)$, up to boundary terms. Then, the assignment of the tuple 
$$\FF_{EH}=(\calF_{EH}(M), S_{EH}(M), \Omega_{EH}(M), Q_{EH}(M)))$$
to every (d+1)-dimensional manifold $M$ that admits a Lorentzian structure will be called Einstein--Hilbert theory in the BV formalism.
\end{definition}

\begin{remark}
The sign convention used above is obtained from the standard ADM decomposition by redefining $(\eta,\beta) \to (-\eta,-\beta)$. This matches our conventions below. This change is due to the choice of using inverse metrics for the first order formulation, instead of metrics (in fact $\Pi_{ab} \partial_t \gamma^{ab} = -\Pi^{ab} \partial_t \gamma_{ab}$).
\end{remark}

\begin{theorem}\label{th:EHstrongequivalence}
Einstein--Hilbert theory in the BV formalism $\FF_{EH}(\Sigma\times I)$ is strongly equivalent to $\uFF{\Sigma\times I}{}$. Explicitly, the isomorphism of the underlying symplectic dg-manifolds reads:
\begin{subequations}\begin{align}
\tg &= -\eta^{-2}\partial_t\partial_t - 2 \eta^{-2}\beta \partial_t + \gamma -\eta^{-2} \beta\beta\\
\txi^{}&= -\eta^{-1} \xi^n \partial_t + \xi^\partial - \eta^{-1} \xi^n \beta\\
\txi^{\dag} & = \xi^\dag_\partial - \left(\eta \xi^\dag_n + \iota_\beta\xi^\dag_\partial\right)dt\\\notag
\tg^\dag & = \left( \frac12 \eta^3 \varphi_n - \eta^2 \varphi_a \beta^a 
    - \gamma^\dag_{ab}\beta^a\beta^b + \eta\beta^a \xi^\dag_a \xi^n 
    + \frac12\eta^{-1}\xi^\dag_n \xi^n\right) dt^2 \\
       & + \left(\frac12 \eta^2 \varphi_a + \gamma^\dag_{ab}\beta^b - \frac12 \eta \xi^\dag_a\xi^n\right) dx^a dt - \gamma^\dag_{ab} dx^a dx^b
\end{align}
\end{subequations}
with inverse:
\begin{subequations}\begin{align}
\eta &=[-\tg^{tt}]^{-\frac{1}{2}}\\
\beta^a & = - [-\tg^{tt}]^{-1} \tg^{ta}\\
\bg^{ab} & = [-\tg^{tt}]^{-1} \tg^{ta} \tg^{tb}\\
\xi^{n} & = - [-\tg^{tt}]^{-\frac12 }\txi^{t}\\
\xi^{a} & = \txi^{a} + [-\tg^{tt}]\tg^{ta}\txi^{t}\\
\bg^\dag_{ab} & = \tg^\dag_{ab}\\
\bphi_a &=  2 [-\tg^{tt}]\tg^\dag_{at}  + 2 \tg_{ab}^\dag \tg^{tb} + \txi^{\dag}_a\txi^{t}\\
\bphi_n & = 2[-\tg^{tt}]^{\frac32 }\tg^\dag_{tt}  - 4[-\tg^{tt}]^{\frac12 } \tg^\dag_{ta} \tg^{ta} + 2[-\tg^{tt}]^{-\frac12 }\tg^\dag_{ab} \tg^{ta}\tg^{tb} \\
& + [-\tg^{tt}]^{\frac12 }\txi^{\dag}_n \txi^{t} -[-\tg^{tt}]^{-\frac12} g^{ta} \txi^{\dag}_a\txi^{a}\\
\xi^\dag_n & = - [-\tg^{tt}]^{\frac12 }\txi^{\dag}_n + [-\tg^{tt}]^{-\frac12}\txi^{\dag}_a \tg^{ta}\\
\xi^{\dag}_a & = \txi^{\dag}_a
\end{align}
\end{subequations}
\end{theorem}

\begin{proof}
We begin observing that the definitions of $K$ in \eqref{KdefAKS} and $K$ in \eqref{KdefADM} coincide up to sign, after identifying $\tg$ with the expression of Equation \eqref{greconstruct}. Since the expression $S_{ADM}:=-\eta\sqrt{\gamma}(\epsilon(K_{ab}K^{ab} - K^2) +R^\partial -2\Lambda)$ is quadratic in $K$, we conclude that the degree-zero part of \eqref{BVEHaction} and \eqref{OnshellAKSZEH} coincide. This means that the two theories are classically equivalent\footnote{{Strictly speaking this is only true on an open subset of the moduli space of solutions, due to the nondegeneracy condition on $\gamma(t)$ enforced on the whole cylinder by the AKSZ construction, while Definition \ref{BVEHaction} only requires the nondegeneracy of $\tg\vert_{\partial M}$.}}, and \eqref{greconstruct} is the map between second-order and first-order Einstein--Hilbert theory.

We endeavour now to find the explicit expression for $\tg^\dag$ and $\txi^{\dag}$ so that 
$$
\phi^*(\langle\tg^\dag, \delta \tg\rangle + \langle\txi^{\dag}, \delta \txi^{}\rangle) = - \langle\delta\bg^\dag, \delta\bg\rangle 
	+ \langle \bpi, \delta\bpi^\dag \rangle
	+ \xi^\dag_\rho \delta\xi^{\rho} 
	+ \bphi_n \delta\eta
	+ \bphi_a\delta\beta^a.
$$
It is straightforward to compute 
\begin{multline}
\phi^*(\langle\tg^\dag, \delta \tg\rangle + \langle\txi^{\dag}, \delta \txi^{}\rangle) = 
	- \left[(\phi^*\txi^{\dag})_n\eta^{-1} 
		+ (\phi^*\txi^{\dag})_a\eta^{-1}\beta^a\right] \delta \xi^n 
		+ (\phi^*\txi^{\dag})_a\delta\xi^a\\
	+\left[2(\phi^*\tg^\dag)_{tt} \eta^{-3} 
		+ 4 (\phi^*\tg^\dag)_{at} \eta^{-3} \beta^a 
		+ 2(\phi^*\tg^\dag)_{ab} \eta^{-3}\beta^a\beta^b \right] \delta \eta \\ 
		- \left[(\phi^*\txi^{\dag})_t \eta^{-2} \xi^n - (\phi^*\txi^{\dag})_a \eta^{-2} \beta^a \xi^n\right] \delta \eta  	
	+(\phi^*\tg^\dag)_{ab}\delta\gamma^{ab} \\
		+ \left[2 (\phi^*\tg^\dag)_{at}\eta^{-2} 
		- 2(\phi^*\tg^\dag)_{ab}\eta^{-2}\beta^b 
		+ (\phi^*\txi^{\dag})_a\eta^{-1}\xi^n\right]\delta\beta^a 
\end{multline}
which leaves us with the intermediate expression:
\begin{subequations}\label{intermediate}\begin{align}
	\xi^\dag_a &= (\phi^*\txi^{\dag})_a\\
	\xi^\dag_n &= - \left[(\phi^*\txi^{\dag})_t\eta^{-1} 
		+ (\phi^*\txi^{\dag})_a\eta^{-1}\beta^a\right]\\
	\bg^\dag_{ab}&=-(\phi^*\tg^\dag)_{ab}\\
	\bphi_a&=-2 (\phi^*\tg^\dag)_{at}\eta^{-2} 
		- 2(\phi^*\tg^\dag)_{ab}\eta^{-2}\beta^b 
		+ (\phi^*\txi^{\dag})_a\eta^{-1}\xi^n\\\notag
	\bphi_n&= 2(\phi^*\tg^\dag)_{tt} \eta^{-3} 
		+ 4 (\phi^*\tg^\dag)_{at} \eta^{-3} \beta^a 
		+ 2(\phi^*\tg^\dag)_{ab} \eta^{-3}\beta^a\beta^b \\
		&- (\phi^*\txi^{\dag})_t \eta^{-2} \xi^n
		- (\phi^*\txi^{\dag})_a \eta^{-2} \beta^a \xi^n
\end{align}\end{subequations}
Starting from the top and solving downwards, we easily get 
\begin{align}
(\phi^*\txi^{\dag})_a &= \xi^\dag_a\\
(\phi^*\txi^{\dag})_n & = - \eta \xi^\dag_n - \xi^\dag_a \beta^a\\
(\phi^*\tg^\dag)_{ab} &= -\gamma^\dag_{ab}\\
(\phi^*\tg^\dag)_{at} &= -\frac12\eta^2\bphi_a 
    + \bg^\dag_{ab}\beta^b 
    + \frac12 \eta \xi^\dag_a\xi^n\\
(\phi^*\tg^\dag)_{tt} & = \frac12 \eta^{3} \bphi_n 
    + \eta^{2}\bphi_a\beta^a 
    - \bg^\dag_{ab}\beta^a\beta^b 
    - \eta \xi^\dag_a\beta^a\xi^n 
    - \frac12 \eta\xi^\dag_n\xi^n
\end{align}
Alternatively, from \eqref{intermediate}, observing that the assignment \eqref{reconstruct} can be inverted to yield
$$
\phi^{-1}{}^*\eta= [-\tg^{tt}]^{-\frac12}, \qquad \phi^{-1}{}^*\beta^a = - [-\tg^{tt}]^{-1} \tg^{ta}, \qquad \phi^{-1}{}^*\gamma^{ab} = [-\tg^{tt}]^{-1} \tg^{ta}\tg^{tb}
$$
together with
$$
\phi^{-1}{}^*\xi^n = -[-\tg^{tt}]^{-\frac12}\txi^{t}; \qquad \phi^{-1}{}^*\xi^a = \txi^{a} + [-\tg^{tt}]^{-1}\tg^{ta}\txi^{t}
$$
we can similarly obtain the inverse: 
\begin{align*}
\phi^{-1}{}^*\eta &=[-\tg^{tt}]^{-\frac{1}{2}}\\
\phi^{-1}{}^*\beta^a & = - [-\tg^{tt}]^{-1} \tg^{ta}\\
\phi^{-1}{}^*\bg^{ab} & = [-\tg^{tt}]^{-1} \tg^{ta} \tg^{tb}\\
\phi^{-1}{}^*\xi^{n} & = - [-\tg^{tt}]^{-\frac12 }\txi^{t}\\
\phi^{-1}{}^*\xi^{a} & = \txi^{a} + [-\tg^{tt}]\tg^{ta}\txi^{t}\\
\phi^{-1}{}^*\bg^\dag_{ab} & = -\tg^\dag_{ab}\\\intertext{}
\phi^{-1}{}^*\xi^{\dag}_a & = \txi^{\dag}_a\\
\phi^{-1}{}^*\xi^\dag_n & = - \txi^{\dag}_n[-\tg^{tt}]^{\frac12 } + \txi^{\dag}_a [-\tg^{tt}]^{-\frac12}\tg^{ta}\\
\phi^{-1}{}^*\bphi_a &= - 2 \tg^\dag_{at} [-\tg^{tt}] + 2 \tg_{ab}^\dag \tg^{tb} - \txi^{\dag}_a\txi^{t}\\
\phi^{-1}{}^*\bphi_n & = 2\tg^\dag_{tt} [-\tg^{tt}]^{\frac32 } - 4 \tg^\dag_{ta}[-\tg^{tt}]^{\frac12 } \tg^{ta} + 2\tg^\dag_{ab}[-\tg^{tt}]^{-\frac12 } \tg^{ta}\tg^{tb} \\
& + [-\tg^{tt}]^{\frac12 }\txi^{\dag}_t \txi^{t} -[-\tg^{tt}]^{-\frac12} g^{ta} \txi^{\dag}_a\txi^{t}.
\end{align*}

Now, using again the intermediate expressions \eqref{intermediate} let us consider the following terms, coming from Equation \eqref{OnshellAKSZEH}:
{%\small
\begin{align*}
    \xi^\dag_n L_{\xi^\partial}\xi^n 
        &= -(\phi^*\txi^{\dag})_t \eta^{-1}L_{\xi^\partial}\xi^n 
            - (\phi^*\txi^{\dag})_a\eta^{-1}\beta^a L_{\xi^\partial}\xi^n\\
    \langle \xi^\dag_\partial,(\xi^n\nabla_{\bg}\xi^n 
            + \frac12[\xi^\partial,\xi^\partial]\rangle 
        &= \langle (\phi^*\txi^{\dag})_\partial,(\xi^n\nabla_{\bg}\xi^n 
                + \frac12[\xi^\partial,\xi^\partial]\rangle
    \end{align*}
    \begin{align*}
    \bphi_n & \left( L_{\xi^\partial}\eta - L_{\beta}\xi^n - \dot{\xi}^n\right) \\
        &= \eta^{-3}\left[2(\phi^*\tg^\dag)_{tt}  + 4 (\phi^*\tg^\dag)_{at}\beta^a 
            + (\phi^*\tg^\dag)_{ab}\beta^a\beta^b\right]\left(L_{\xi^\partial}\eta
            -L_\beta \xi^n -\dot{\xi}^n\right)\\
        & \quad - \eta^{-2}\left[(\phi^*\txi^{\dag})_n \xi^n
                + (\phi^*\txi^{\dag})_a \beta^a \xi^n\right] \left(L_{\xi^\partial}\eta 
                -L_{\beta}\xi^n - \dot{\xi}^n\right) 
    \end{align*}
    \begin{align*}
    \langle\bphi_\partial&,\left(\nabla_{\bg}\xi^n 
            -\eta\nabla_{\bg}\xi^n + L_{\xi^\partial}\beta\right)\rangle \\
        &=-2\eta^{-2}\left((\phi^*\tg^\dag)_{at} 
            +(\phi^*\tg^\dag)_{ab}\beta^b\right)\left((\nabla_{\bg}\eta)^a\xi^n 
            - \eta(\nabla_{\bg}\xi^n)^a +(L_{\xi^\partial}\beta)^a 
            - \dot{\xi}^a\right)\\
        & \quad + (\phi^*\txi^{\dag})_a \eta^{-1} \xi^n\left((L_{\xi^\partial}\beta)^a
            -\eta(\nabla_{\bg}\xi^n)^a - \dot{\xi}^a\right) 
            \end{align*}
    \begin{align*}
    \langle \bg^\dag, \eta^{-1} \left(\dot{\bg} + L_\beta\bg\right)\xi^n 
        - L_{\xi^\partial}\bg\rangle 
    & = -(\phi^*\tg^\dag)_{ab}\left(\eta^{-1}\left(\dot{\bg} 
        + L_\beta\bg\right)\xi^n 
        - L_{\xi^\partial}\bg\right)^{ab}.
\end{align*}}
Then, summing all terms on the left hand side and factoring $(\phi^*\txi^{\dag})_t$, $(\phi^*\txi^{\dag})_a$ and $(\phi^*\tg^{\dag})$, we obtain
{\small
\begin{align*}
    & (\phi^*\txi^{\dag})_t\left[
        -L_{\xi^\partial}(\eta^{-1}\xi^n) 
        +\eta^{-2}\xi^n\left( 
        L_{\beta}\xi^n 
        + \dot{\xi}^n\right) \right]\\
    & +(\phi^*\tg^\dag)_{tt}\left[2\eta^{-3}\left(L_{\xi^\partial}\eta 
        - L_\beta \xi^n 
        - \dot{\xi}^n\right)\right] +\left\langle (\phi^*\txi^{\dag})_\partial, 
        \frac12[\xi^\partial,\xi^\partial] \right\rangle\\
    & +\left\langle (\phi^*\txi^{\dag})_\partial,     - L_{\xi^\partial}(\eta^{-1}\xi^n)\beta
        + \eta^{-1}L_{\xi^\partial}\beta\xi^n 
        + \eta^{-2} \beta \xi^n L_\beta \xi^n 
        + \eta^{-2}\beta \xi^n \dot{\xi}^n 
        - \eta^{-1}\xi^n \dot{\xi}^\partial\right\rangle\\    &+(\phi^*\tg^\dag)_{ab}\left[-\eta^{-1}\dot{\bg}^{ab}\xi^n - \eta^{-1}(L_\beta\bg)^{ab}\xi^n + (L_{\xi^\partial}\bg)^{ab} + 4\eta^{-3} \beta^a\left(L_{\xi^\partial}\eta 
        - L_\beta \xi^n 
        - \dot{\xi}^n\right)\right]\\
        &+(\phi^*\tg^\dag)_{ab}\left[-2\eta^{-2}\left((\nabla_{\bg}\eta)^a\xi^n 
        - \eta(\nabla_{\bg}\xi^n)^a 
        +(L_{\xi^\partial}\beta)^a 
        - \dot{\xi}^a\right)\right]\\
    & +(\phi^*\tg^\dag)_{ab}\left[2\eta^{-3}\beta^a\beta^b\left(L_{\xi^\partial}\eta 
        - L_\beta \xi^n 
        - \dot{\xi}^n\right)    \right]
       \\
    & +(\phi^*\tg^\dag)_{ab}\left[ -2\eta^{-2}\beta^b\left((\nabla_{\bg}\eta)^a\xi^n 
        - \eta(\nabla_{\bg}\xi^n)^a 
        +(L_{\xi^\partial}\beta)^a 
        - \dot{\xi}^a\right)\right]
\end{align*}}

Which, using Lemma \ref{lem:algebroid}, can be shown to be
\begin{equation}
    \phi^*\left(\tg^\dag L_{\txi}\tg + \iota_{[\txi,\txi]}\txi^\dag\right)    
\end{equation}
leading to
\begin{equation}
    \phi^*S_{EH}(\Sigma\times I) = \uS{\Sigma\times I}{}.
\end{equation}
\end{proof}

\begin{remark}
{We would like to stress here that the results in this section are a ``strictification'' of the general construction of a solution of the classical master equation for the extended Hamiltonian, as presented by Henneaux and Bunster in \cite[Theorem 18.8]{HT}. Indeed, the Hamiltonian analysis for a field theory relies on a (possibly) non-reduced version of the strict BFV data we consider, where strict indicates that we require all spaces of fields to be smooth symplectic manifolds. The AKSZ construction yields a BV theory (Theorem \ref{th:AKSZEH}) which is effectively equivalent to the natural BV extension of Einstein--Hilbert theory (Theorems \ref{th:pushforward} and \ref{th:EHstrongequivalence}). It could be argued that this effective equivalence preserves the BV cohomology \cite{Henn,DGH,BBH}. However, note that the quantisation procedure outlined in \cite{CMR2} does indeed require the strict version of a BV-BFV structure\footnote{See \cite{MSW2019} for the comparison between strict and lax BV-BFV structures.}, and its existence is not to be taken for granted, as was shown in \cite{CS2017} and \cite{CS2016a}.
}\end{remark}

\section{AKSZ PC}\label{s:AKSZPC}

Following the construction outlined in Section \ref{sec:local_description}, starting from the BFV theory of Palatini--Cartan gravity (see Section \ref{sec:PC-BFV_theory}), we can construct the AKSZ space of fields $\calF_{PC}^{\AKSZ}$.
We will use the following notation:
{
\begin{subequations}\label{e:variables_PC_AKSZ}
\begin{align}
\mathfrak{e}& = e + f^\dag & \mathfrak{w}&= \omega + u^\dag \\
\mathfrak{c}& = c + w & \mathfrak{x} &=\xi + z \\
\mathfrak{l}& = \lambda + \mu & \mathfrak{c}^\dag&= k^\dag + c^\dag\\
\mathfrak{y}^\dag & = e^\dag + y^\dag & &
\end{align}
\end{subequations}
where
\begin{equation}\label{e:AKSZfields}
\begin{aligned}
 e & \in C^{\infty}(I)\otimes \Omega_{nd}^1(\Sigma, \mathcal{V}) 			&
 f^\dag & \in \Omega^1[-1](I) \otimes  \Omega^1(\Sigma, \mathcal{V}) 		\\ 
 \omega & \in  C^{\infty}(I)\otimes \mathcal{A}^{}(\Sigma) 			    & 
 u^\dag & \in \Omega^1[-1](I)\otimes \mathcal{A}^{}(\Sigma) 			\\
 c & \in \Omega^0[1]( I \times \Sigma, \textstyle{\bigwedge^2}\mathcal{V})				&
 w & \in \Omega^1[-1](I)\otimes \Omega^0[1]( \Sigma, \textstyle{\bigwedge^2}\mathcal{V})  \\ 
 \xi & \in C^{\infty}(I)\otimes \mathfrak{X}[1](\Sigma)					&
 z & \in \Omega^1[-1](I)\otimes \mathfrak{X}[1](\Sigma)					\\
 \lambda & \in C^\infty[1](I \times \Sigma)								&
 \mu & \in \Omega^1[-1](I)\otimes C^\infty[1](\Sigma)						\\
 k^\dag & \in C^{\infty}(I)\otimes \Omega^{N-1}[-1](\Sigma,\textstyle{\bigwedge^{N-2}} \mathcal{V}) 	&
 c^\dag & \in \Omega^1[-1](I)\otimes \Omega^{N-1}[-1](\Sigma,\textstyle{\bigwedge^{N-2}} \mathcal{V})\\
 e^\dag & \in C^{\infty}(I)\otimes\Omega^{N-1}[-1](\Sigma,\textstyle{\bigwedge^{N-1}} \mathcal{V})	&
 y^\dag & \in \Omega^1[-1](I)\otimes \Omega^{N-1}[-1](\Sigma,\textstyle{\bigwedge^{N-1}} \mathcal{V})
\end{aligned}
\end{equation}
such that, for some $\sigma \in C^{\infty}(I)\otimes \Omega^1(\Sigma, \mathcal{V})$ and $B\in \Omega^1[-1](I) \otimes  \Omega^1(\Sigma, \mathcal{V})$, they satisfy the \emph{structural AKSZ constraints}:
% for some $\mathfrak{s}\in\mathrm{Maps}(T[1]I,\Omega^1(\Sigma,\mathcal{V}))$,
% \begin{equation}\label{e:AKSZstructconstraint}
%     (N-3)\epsilon_n  \mathfrak{e}^{N-4} d_{\mathfrak{w}}\mathfrak{e}+  \left(L_{\mathfrak{x}}^{\mathfrak{w}} \epsilon_n - [\mathfrak{c},\epsilon_n]\right)^{(a)}  \mathfrak{c}^\dag_a  = \mathfrak{e}^{N-3} \mathfrak{s},
% \end{equation}
% which is equivalent to the following equations (we omit the (global) $N-3$ factors):
\begin{subequations}
\begin{align}
    \epsilon_n \left\{(N-4)f^\dag e^{N-5} d_\omega e + e^{N-4} d_{\omega} f^\dag + e^{N-4} [u^\dag, e]\right\} &\label{e:struct_constrAKSZ2}\\ \nonumber
        +\left(\iota_z d_{\omega}\epsilon_n-[w-\iota_\xi u^\dag, \epsilon_n]\right)^{(a)}  k^\dag_a+ X^{(a)}  c^\dag_a + \left(X^{(b)}f^\dag_b\right)^{(a)}k^\dag_a&= f^\dag e^{N-4} \sigma + e^{N-3} B; \\
    \epsilon_n  e^{N-4} d_{\omega} e + X^{(a)}  k^\dag_a  &= e^{N-3}\sigma; \label{e:struct_constrAKSZ1}
\end{align}
\end{subequations}
where $X= \left(L_\xi^\omega \epsilon_n - [c,\epsilon_n]\right)\in\Gamma(M,\mathcal{V})$, while $\epsilon_n\in\Gamma(M,\mathcal{V})$ is a fixed section, and the indices $\{(a),(n)\}$ denote components with respect to a basis $\{e_a, \epsilon_n\}$.

\begin{remark}
Observe that our target for the AKSZ construction for Palatini--Cartan theory is the BFV theory defined in Definition \ref{def:BFVdata}, whose space of fields $\calF_{PC}^\partial$ is defined by the structural constraint \eqref{e:BFVstructuralConstraint}. As a consequence, the BFV constraint must be imposed on the AKSZ fields at every point of $I$. As the AKSZ fields consists of a $0$- and and $1$-form component, along the interval, the structural constraints now has a $0$-form and a $1$- form component corresponding to \eqref{e:struct_constrAKSZ1} and \eqref{e:struct_constrAKSZ2}, respectively. Despite the apparent complexity of these two equations, it is worth noting that they fix certain components of the AKSZ fields $\omega$ and $u^\dag$. See Section \ref{sec:Interpretation} for an interpretation. 
\end{remark}

\begin{remark}
Recall that to define the BFV structure for PC theory we needed a fixed section $\epsilon_n\in\Gamma(\Sigma,\mathcal{V})$ (cf. Definition \ref{def:BFVspaceoffields}). Note that $\epsilon_n$ is not a field of the theory but is part of the structure that defines it (more like a coupling constant). For this reason, in the AKSZ construction $\epsilon_n$ does not depend on the coordinate $x^n$ of the interval $I$. In the following, we will regard $\epsilon_n$ as a given section of $\Gamma(M,V)$ satisfying $d_I(\epsilon_n)=0$.
\end{remark}
}

\begin{theorem}\label{thm:AKSZPC}
The AKSZ data $\mathfrak{F}_{PC}^{\AKSZ}(I; \mathfrak{F}_{PC}^{\partial})$ are given by the quadruple 
$$
\mathfrak{F}_{PC}^{\AKSZ}(I; \mathfrak{F}_{PC}^{\partial})=(\mathcal{F}_{PC}^{\AKSZ}, S_{PC}^{\AKSZ}, \varpi_{PC}^{\AKSZ}, Q_{PC}^{\AKSZ})
$$
where:
\begin{align*}
\mathcal{F}_{PC}^{\AKSZ} \simeq T^*[-1](\mathrm{Map} (I, \mathcal{F}^{\partial}_{PC}))
\end{align*}
\begin{align*}
    \varpi_{PC}^{\AKSZ}= \int_{I\times \Sigma}&  \delta (e^{N-3} f^\dag)  \delta \omega+e^{N-3} \delta e \delta u^\dag + \delta w \delta k^\dag+ \delta c \delta c^\dag + \delta u^\dag \delta(\iota_\xi k^\dag)\\
& + \delta \omega \delta(\iota_z  k^\dag)+ \delta \omega \delta(\iota_\xi c^\dag) - \delta \mu \epsilon_n \delta e^\dag - \delta \lambda \epsilon_n \delta y^\dag  \\
&+\iota_{\delta z}\delta (e e^\dag)	+\iota_{\delta \xi} \delta (f^\dag e^\dag)+\iota_{\delta \xi} \delta (e y^\dag);
\end{align*}
\begin{align}
 S_{PC}^{\AKSZ} =\int_{I\times \Sigma} &  w e^{N-3} d_{\omega} e
	 + (N-3) c e^{N-4} f^\dag d_{\omega} e
	 + c e^{N-3} [u^\dag, e]
	 + c e^{N-3} d_{\omega} f^\dag \nonumber\\
	&+ \iota_{z} e e^{N-3} F_{\omega} 
	 + \iota_{\xi}( e^{N-3} f^\dag) F_{\omega}
	 + \iota_{\xi} e e^{N-3} d_{\omega}u^\dag 
	 + \epsilon_n \mu e^{N-3} F_{\omega} \nonumber \\ 
	&+ (N-3) \epsilon_n \lambda e^{N-4} f^\dag F_{\omega}
	 + \epsilon_n \lambda e^{N-3} d_{\omega}u^\dag 
	 + [w,c] k^{\dag}
	 + \frac{1}{2} [c,c] c^{\dag} \nonumber \\ 
	&- \iota_{z} d_{\omega} c k^{\dag}
	 - [\iota_{\xi}u^\dag ,c] k^{\dag}
	 - \iota_{\xi} d_{\omega}w k^{\dag}
	 - \iota_{\xi} d_{\omega} c c^{\dag}
	 + \iota_{z}\iota_{\xi} F_{\omega}k^{\dag}\nonumber \\ 
	&+ \frac{1}{2} \iota_{\xi}\iota_{\xi} d_{\omega}u^\dag k^{\dag}
	 + \frac{1}{2} \iota_{\xi}\iota_{\xi} F_{\omega}c^{\dag}
	 - [w, \epsilon_n \lambda ]e^{\dag}
	 - [c, \epsilon_n \mu ]e^{\dag} 
	 - [c, \epsilon_n \lambda ]y^{\dag} \nonumber \\
	&+ \iota_{z}d_{\omega} (\epsilon_n \lambda)e^{\dag} 
	 + [\iota_{\xi}u^\dag ,\epsilon_n \lambda]e^{\dag}
	 + \iota_{\xi}d_{\omega} (\epsilon_n \mu )e^{\dag} 
	 + \iota_{\xi}d_{\omega} (\epsilon_n \lambda)y^{\dag}\nonumber \\
	&+ \iota_{[z,\xi]}e e^{\dag}
  	 + \frac{1}{2}\iota_{[\xi,\xi]}f^\dag e^{\dag}
	 + \frac{1}{2}\iota_{[\xi,\xi]}e y^{\dag}  \nonumber  \\
	&+  \frac{1}{N-2} e^{N-2} d_I \omega 
	 + c d_I k^\dag+ d_I \omega \iota_\xi k^\dag 
	 - \iota_{d_I \xi} e  e^\dag 
	 + d_I \lambda \epsilon_n e^\dag \nonumber.
\end{align}
{
and $Q_{PC}^{\AKSZ}$ is defined as $Q_{PC}^{\AKSZ}= Q_{PC}^{\text{dR}} + Q_{PC}^{\text{lift}}$ where ${Q}_{PC}^{\text{lift}}$ is the tangent lift of $Q^{\partial}_{PC}$ to $\mathrm{Map}(T[1]I,\Fp{\Sigma}{PC})$ and $Q_{PC}^{dR}$ is the lift of the de Rham differential $d_I$.

%is the Hamiltonian vector field of $S_{PC}^{\AKSZ}$ with respect to $\varpi^{\AKSZ}_{PC}$:\todo{Guardando il commento di Alberto sotto, forse conviene dire direttamente qui che $Q^{\AKSZ}_{PC}=\check{Q}_{BFV} + d_I$ Così evitiamo anche di dover dire cosa vuol dire calcolarlo a meno di termini di bordo.}
%$$\iota_{Q^{\AKSZ}_{PC}} \varpi^{\AKSZ}_{PC} = \delta S^{\AKSZ}_{PC}.$$
}
\end{theorem}

\begin{proof}
This is a straightforward application of the AKSZ prescription outlined in Section \ref{sec:local_description}. 
Using the transgression map we can build a symplectic form $\calF_{PC}^{\AKSZ}$
\begin{align}\label{e:sympl_form_AKSZ}
\varpi^{\AKSZ}_{PC} = \int_{I\times \Sigma}& \mathfrak{e}^{N-3} \delta \mathfrak{e} \delta \mathfrak{w} + \delta \mathfrak{c} \delta\mathfrak{c}^\dag + \delta \mathfrak{w} \delta (\iota_{\mathfrak{x}} \mathfrak{c}^\dag) - \delta \mathfrak{l} e_n \delta \mathfrak{y}^\dag +  \iota_{\delta\mathfrak{x}} \delta (\mathfrak{e}\mathfrak{y}^\dag)
\end{align}
from which we obtain the claimed expression using \eqref{e:variables_PC_AKSZ}.
Analogously the AKSZ action can be constructed using the transgression map from the boundary one-form $\alpha^{\partial}$ and from the boundary action $S^{\partial}$.  Namely we have 
\begin{align}\label{e:action_AKSZ}
S^{\AKSZ}_{PC}=\int_{I\times \Sigma}&  \frac{1}{N-2} \mathfrak{e}^{N-2} d_I \mathfrak{w} + \mathfrak{c} d_I \mathfrak{c}^\dag+ d_I \mathfrak{w} \iota_{\mathfrak{x}} \mathfrak{c}^\dag - \iota_{d_I \mathfrak{x}} \mathfrak{e}  \mathfrak{y}^\dag + d_I \mathfrak{l} \epsilon_n \mathfrak{y}^\dag \\
 &  \mathfrak{c} \mathfrak{e} d_{\mathfrak{w}} \mathfrak{e} + \iota_{\mathfrak{x}} \mathfrak{e} \mathfrak{e} F_{\mathfrak{w}} + \epsilon_n \mathfrak{l} \mathfrak{e} F_{\mathfrak{w}} +\frac{1}{2} [\mathfrak{c},\mathfrak{c}] \mathfrak{c}^{\dag} - L^{\mathfrak{w}}_{\mathfrak{x}} \mathfrak{c} \mathfrak{c}^{\dag}+ \frac{1}{2} \iota_{\mathfrak{x}}\iota_{\mathfrak{x}} F_{\mathfrak{w}}\mathfrak{c}^{\dag}\nonumber\\ & -[\mathfrak{c}, \epsilon_n \mathfrak{l} ]\mathfrak{y}^{\dag} + L^{\mathfrak{w}}_{\mathfrak{x}} (\epsilon_n \mathfrak{l})\mathfrak{y}^{\dag} + \frac{1}{2}\iota_{[\mathfrak{x},\mathfrak{x}]}\mathfrak{e} \mathfrak{y}^{\dag}. \nonumber
\end{align}
Again the claimed expression can be obtained straightforwardly from \eqref{e:variables_PC_AKSZ}.
\end{proof}
{
\begin{remark}
%The explicit expression of the vector field $Q_{PC}^{\AKSZ}$ is rather complicated and it can be computed in two equivalent ways. The definition we gave in Theorem \ref{thm:AKSZPC} coincides with $Q_{PC}^{\AKSZ}$ as the Hamiltonian vector field of $S^{\AKSZ}_{PC}$, i.e. up to boundary terms we have $\iota_{Q^{\AKSZ}_{PC}} \varpi^{\AKSZ}_{PC} = \delta S^{\AKSZ}_{PC}$. 
%Recall that one can uniquely fix Hamiltonian vector fields on a constrained symplectic submanifold of a presymplectic manifold by requiring them to be tangent to the constraints. This was done for the BFV space of fields in Section \ref{sec:PC-BFV_theory} (following \cite{CCS2020}), and will be assumed in what follows. 
The invariance of the constraints \eqref{e:struct_constrAKSZ1} and \eqref{e:struct_constrAKSZ2} with respect to $Q_{PC}^{\AKSZ}$ is guaranteed by the invariance of the structural constraint on the boundary \eqref{e:BFVstructuralConstraint} with respect to $Q^{\BFV}_{PC}$, and by the properties of the AKSZ construction. 
\end{remark}
}

%{
%\begin{remark}
%Despite the fact that $\varpi^{\AKSZ}_{PC}$ is symplectic by construction, we will explicitly see how the constraints \eqref{e:struct_constrAKSZ1} and \eqref{e:struct_constrAKSZ1} are necessary to obtain a symplectic form on the AKSZ fields \eqref{e:AKSZfields} in Appendix \ref{appendix:Nd-PC-comparison}. \todo{correggere etichetta con quella giusta, quando l'avremo. CI SERVE QUESTO CONTO?Secondo me NO}
%\end{remark}
%}

From Theorem \ref{AKSZtheorem} we know that $\mathfrak{F}_{PC}^{\AKSZ}(I; \mathfrak{F}_{PC}^{\partial})$ yields a BV theory on the manifold $I \times \Sigma $. Furthermore, by Proposition \ref{prop:AKSZ-BV-BFV} these data satisfy also the BV-BFV axioms of Equation \eqref{BVBFVeqts}.
{
\begin{definition}\label{def:ndPCAKSZ}
We call \emph{nondegenerate} AKSZ PC theory the data $\FF_{PC_\star}^{\AKSZ}$ obtained by restricting the space of fields of $\mathfrak{F}_{PC}^{\AKSZ}(I;\FF^\partial_{PC}(\Sigma))$ to those maps whose $\mu$ component (as defined by Equation \eqref{e:variables_PC_AKSZ}) is nonvanishing.
\end{definition}
}

In \cite{CS2017} two of the authors proved that, using the natural symmetries of PC theory, the resulting BV theory $\FF_{PC}$ does not satisfy the BV-BFV axioms (it is not a $1$-extended BV theory) unless additional requirements on the fields are enforced. Next section will be devoted to the comparison between $\FF_{PC}(\Sigma\times I)$ and $\FF_{PC_\star}^{\AKSZ}(I;\FF^\partial_{PC}(\Sigma))$.

\subsection{Comparison of BV data for PC theory}
We want to compare the AKSZ-BV theory of Theorem \ref{thm:AKSZPC} with the one proposed for PC-gravity by two of the authors \cite{CS2017}, which we briefly recall here. Let $M$ be an $N$-dimensional manifold with $N>2$.

\begin{definition}\label{def:standardPC-BV}
We call \emph{standard} BV theory for PC gravity the BV data 
$$\FF_{PC}(M)= (\calF_{PC}(M), S_{PC}(M), \varpi_{PC}(M), Q_{PC}(M))$$
where 
\begin{equation*}
\calF_{PC}(M)\coloneqq T^*[-1]\left(\Omega_{nd}^1(M, \mathcal{V}) \oplus\mathcal{A}(M) \oplus \mathfrak{X}[1](M) \oplus \Omega^0[1](M,\mathrm{ad}P)\right)
\end{equation*} 
and the fields in the base are denoted by $(\be, \bom, \bxi^{}, \bc)$, while the corresponding variables in the cotangent fibre are denoted by $(\be^{\dag}, \bom^{\dag}, \bxi^{\dag}, \bc^{\dag})$;
\begin{equation*}
\varpi_{PC}(M) = \int_M \delta \be \delta \be^{\dag} + \delta \bom \delta \bom^{\dag}+ \delta \bc \delta \bc^{\dag} + \iota_{\delta \bxi^{}} \delta\bxi^{\dag}; 
\end{equation*}
\begin{align*}
S_{PC}(M) &=\int_M \frac{1}{N-2}	\be^{N-2} F_{\bom}  + \left(\iota_{\bxi^{}} F_{\bom}  - d_{\bom} \bc \right)\bom^\dag - \left(L_{\bxi^{}}^{\bom}\be- [\bc,\be]\right)\be^\dag\\
&+\int_M \frac{1}{2}\left(\iota_{\bxi^{}}\iota_{\bxi^{}} F_{\bom}  - [\bc,\bc]\right)\bc^\dag +\frac12\iota_{[\bxi^{},\bxi^{}]}\bxi^{\dag}.
\end{align*}
\end{definition}

The explicit expression of the cohomological vector field $Q_{PC}$, defined by the equation $\iota_{Q_{PC}}\varpi_{PC}= \delta S_{PC}$, will be useful in the following:

\begin{align*}
Q_{PC} \be &=  L_{\bxi^{}}^{\bom}\be- [\bc,\be] \\
Q_{PC} \bom &= \iota_{\bxi^{}} F_{\bom}  - d_{\bom} \bc \\
Q_{PC} \bc &= \frac{1}{2}\iota_{\bxi^{}}\iota_{\bxi^{}} F_{\bom}  - \frac{1}{2}[\bc,\bc]\\
Q_{PC} \bxi^{} &= \frac{1}{2} [\bxi^{}, \bxi^{}]\\
Q_{PC} \be^{\dag} & = \be^{N-3} F_{\bom} + L_{\bxi^{}}^{\bom} \be^{\dag} - [\bc , \be^{\dag}]\\
Q_{PC} \bom^{\dag} &= \be^{N-3} d_{\bom} \be - d_{\bom} \iota_{\bxi^{}} \bom^{\dag} - [\bc, \bom^{\dag}] + \iota_{\bxi^{}}[\be, \be^{\dag}]-\frac{1}{2} d_{\bom} \iota_{\bxi^{}} \iota_{\bxi^{}} \bc^{\dag}\\
Q_{PC} \bc^{\dag} &= - d_{\bom} \bom^\dag - [\be, \be^{\dag}] - [\bc, \bc^{\dag}] \\
Q_{PC} \bxi^{\dag}_{\bullet} &= F_{\bom\bullet} \bom^\dag -  (d_{\bom\bullet}\be)\be^{\dag}+\iota_{\bxi^{}}F_{\bom\bullet} \bc^\dag + L_{\bxi^{}}^{\bom}\bxi^{\dag}_{\bullet} + (d_{\bom}\iota_{\bxi^{}} \bxi^{\dag})_{\bullet}.
\end{align*}
Here we used the symbol $\bullet$ to remind the reader that $\bxi^\dag$ is a one-form with values in densities, and on the right hand side we highlight the one-form part of the expression.

{
\begin{remark}
Throughout the analysis we should always keep in mind that, while Definition \ref{def:standardPC-BV} is valid for any manifold $M$ (possibly with boundary), the AKSZ theory obtained in Theorem \ref{thm:AKSZPC} is by construction defined on a manifold diffeomorphic to a cylinder: $M=\Sigma\times I$. Furthermore, as we will see  in this section, the fields in $\mathcal{F}^{\AKSZ}_{PC^\star}$ correspond to those in the standard BV theory for PC but with an additional constraint.
\end{remark}

The product structure of $M$ induces a splitting of fields:
\begin{equation}\label{e:PCfieldsplitting}\begin{aligned}
    \be &= \te + {\te}_ndx^n 
        &  \be^\dag 
            &= \te{}^\dag + {\te}^{\dag}_ndx^n \\
    \bom &= \tom + {\tom}_ndx^n 
        &  \bom^\dag 
            &= \tom{}^\dag + {\tom}^{\dag}_ndx^n \\
    \bxi &= \txi + {\txi}{}^n\partial_n
        &  {\bxi}^\dag &= \underline{\txi}{}^\dag + {\txi}^{\dag}_ndx^n 
\end{aligned}\end{equation}
More compactly we can write field components in the $x^n$ direction as $\underline{\te}_n = \te_n dx^n$, $\underline{\txi}{}^n= \txi{}^n\partial_n$ and so on. Observe that $\bxi{}^\dag$ is a one-form with values in densities on $M$, so we can identify two $dx^n$ contributions: we denote by $\underline{\txi}_n{}^{\!\!\!\dag}$ the $dx^n$-component (of the one-form part) of $\bxi{}^\dag$ and by $\underline{\txi}$ the rest, stressing that the image of $\underline{\txi}$ is nontrivial along $dx^n$. This decomposition allows us to define the maps 
\begin{equation}\label{e:W-maps}
{W}_{\te{}^{N-3}}^{i,j}\colon \Omega^i\left(M,\textstyle{\bigwedge^j}\mathcal V\right) \rightarrow \Omega^{i+N-3}\left(M,\textstyle{\bigwedge^{j+N-3}}\mathcal V\right); \qquad {W}_{\te{}^{N-3}}^{i,j}(v) = \te{}^{N-3} \wedge v.
\end{equation}
%\todo[inline]{Fissare la notazione e la spiegazione}
%where $\underline{\txi}{}^n = \bxi{}^n\partial_n$ and for every other field $\ta$ we denote $\underline{\tta}_n={\widetilde{a}}_n dx^n$ so that $\iota_{\txi{}^n}\tta=0$ \todo{Not true for $\txi^\dag$}. Observe that $\bxi^\dag$ is a one-form on $M$ with values in densities, so  while $\bc^\dag= \underline{\tc}_n{}^{\!\!\!\dag}$. 

Let us now fix a nonzero section $\epsilon_n \in \Gamma(M,\mathcal{V})$ such that $d_I\epsilon_n=0$. We will then restrict the field $\te$ not to have components parallel to $\epsilon_n$. This is a restriction on the space of fields (it actually defines an open subspace). The nondegeneracy of $e$ implies that $(i)$ $\Tilde e$ and $\epsilon_n$ form a basis of $V$ at every point, and $(ii)$ $\te_n$ becomes a linear combination of $\te$ and $\epsilon_n$, with nonzero $\epsilon_n$-component. Denote by $X^{\{\mu\}}$ the components of a field $X$ with respect to the basis given by $\te$ and $\epsilon_n$ (i.e. $X = X^{\{b\}}\te_b+X^{\{n\}}{\epsilon}_n$).

Additionally, we consider the quantity
\begin{equation}\label{e:defWdag}
    \mathfrak{W}^\dag:= \underline{\tom}^{\dag}_n - \tom_a{}^{\!\!\!\dag} {\te}_n^{\{a\}}dx^n- \iota_{\txi}\underline{\tc}_n{}^{\!\!\!\dag}+{\tc}_{an}{}^{\!\!\!\dag}{\txi}{}^n{\te}_n^{\{a\}}dx^n.
\end{equation}
Its meaning will become manifest with the following:
\begin{definition}\label{def:PCStructuralConstraints}
We denote by $\iota_R\colon \calF_{PC}^{\res} \to \calF_{PC}$ the subspace of BV Palatini--Cartan fields defined by the following equations, which we call \emph{PC structural constraints}:
\begin{subequations}
    \begin{align}
     {\epsilon}_n \te{}^{(N-4)} d_{\tom} \te - {\epsilon}_n \te{}^{(N-4)} {W}_{\te{}^{N-3}}^{-1}(\mathfrak{W}^\dag) d \txi{}^n &\nonumber\\
    + ( [\tc, \epsilon_n]+ L_{\txi{}}^{\tom}(\epsilon_n)-d_{\tom_n} \epsilon_n]\txi^{n})^{\{a\}}(\tom_a{}^{\!\!\!\dag} -{\tc}_{an}{}^{\!\!\!\dag}{\txi}{}^n) &\in  \mathrm{Im}({W}_{\te{}^{N-3}}^{1,1}) \label{e:PCconstraints2}\\
    \mathfrak{W}^\dag &\in \mathrm{Im}({W}_{\te{}^{N-3}}^{1,1}) \label{e:PCconstraints1}
\end{align}
\end{subequations}
and by the condition that the metric $g_{\text{hor}}\coloneqq \te^*\eta$ is nowhere degenerate\footnote{{Notice that the condition on $g_{\text{hor}}$ will restrict the moduli space of solutions of the theory to an open subset.}}.
\end{definition}

\begin{remark}
The PC structural constraints \eqref{e:PCconstraints1} and \eqref{e:PCconstraints2} are invariant under the action of $Q_{PC}$. Thus they define a BV theory  
\begin{equation}\label{e:defFR}
    \FF_{PC}^{\res}:=\left(\calF_{PC}^{\res},\varpi_{PC}^{\res}=\iota_R^*\varpi_{PC},S_{PC}^{\res}=\iota_R^*S_{PC},Q_{PC}^{\res}\right)
\end{equation}
where $Q_{PC}^{\res}$ is the restriction of $Q_{PC}$ to $\calF_{PC}^{\res}$.
We will call this theory the \emph{restricted BV Palatini--Cartan theory}. 
The direct proof of the invariance of the constraints is lengthy and involved, yet we get this result for free as a corollary of the following theorem, which also specifies the relations between the three BV theories $\FF_{PC}^{\res}$, $\FF_{PC}$ and $\FF_{PC^\star}^{\AKSZ}$.
\end{remark}

%\begin{equation}
%    \xymatrix{
%      \hat{\calF}_{PC}^{\AKSZ} \ar[rr]^{\widetilde{\varphi}}  &   &   \calF_{PC}\\
%    \calF_{PC_\star}^{\AKSZ} \ar[u]_{\iota_{\AKSZ}} \ar[rru]^{\varphi} \ar[rr]^{\underline{\varphi}}  & &   \calF_{PC}^{\res}  \ar[u]^{\iota_R}
%    }
%\end{equation}

\begin{theorem}\label{thm:comparison}
Upon choosing the same section $\epsilon_n\in\Gamma(M,\mathcal{V})$ and the same signature for $g_{\text{hor}}$ in the three theories, the following diagram commutes
\begin{equation}
    \xymatrix{
        &   &   \calF_{PC}\\
    \calF_{PC_\star}^{\AKSZ} \ar[rru]^{\varphi} \ar[rr]^{\underline{\varphi}}  & &   \calF_{PC}^{\res}  \ar[u]^{\iota_R}
    }
\end{equation}
Moreover, $\underline{\varphi}$ is a symplectomorphism and we have
\begin{equation}
    \varpi^{\AKSZ}=(\underline{\varphi}\circ\iota_R)^*\varpi_{PC}; \qquad S^{\AKSZ}_{PC}= (\underline{\varphi}\circ\iota_R)^* S_{PC},
\end{equation}
so that $\underline{\varphi}$ and $\iota$ induce a strong BV equivalence and a BV inclusion, respectively:
\begin{align*}
\FF^{\AKSZ}_{PC_\star} \xrightarrow{\underline{\varphi}} \FF_{PC}^{\res} \qquad  \FF^{\AKSZ}_{PC_\star} \xrightarrow{\varphi} \FF_{PC}.
\end{align*} 
% where we define \emph{reduced} Palatini--Cartan theory to be the data 
% \begin{equation}\label{e:defFR}
%     \FF_{PC}^{\res}:=\left(\calF_{PC}^{\res},\varpi_{PC}^{\res}=\iota_R^*\varpi_{PC},S_{PC}^{\res}=\iota_R^*S_{PC},Q_{PC}^{\res}\right),
% \end{equation}
% where $Q_{PC}^{\res}$ is the restriction of $Q_{PC}$ to $\calF_{PC}^{\res}$.
\end{theorem}

%\begin{remark}
%Exploiting Proposition \ref{prop:equivalence+inclusion} we can prove Theorem \ref{thm:comparison} by showing that there exists a BV manifold $\FF_R=( \calF_{R}, S_{R}, \varpi_{R}, Q_{R})$, a strong BV equivalence $\varphi$ and a BV inclusion $\iota$ 
%\begin{align*}
%\calF^{\AKSZ}_{PC_\star} \xrightarrow{\varphi} \calF_{R} \xrightarrow{\iota} \calF_{PC}.
%\end{align*}
%{In particular we will give an explicit expression for the diffeomorphism $\varphi$ and characterise the fields in the target $\calF_{R}$, so that the inclusion $\iota$ becomes self-evident.}
%\end{remark}

\begin{remark}[Proof Strategy] In order to prove this, we will first show that there is an injective map $\varphi: \calF_{PC^\star}^{\AKSZ}\to \calF_{PC}$ such that $\varphi^*\varpi_{PC} = \varpi_{PC}^{\AKSZ}$ and $\varphi^*S_{PC} = S_{PC}^{\AKSZ}$. Note that, as a symplectomprhism, $\varphi$ is then an immersion. Then we will show that $\calF_{PC}^{\res}$ is the image of this map, so that the PC structural constraints \eqref{e:PCconstraints1} and \eqref{e:PCconstraints2} are satisfied if and only if the AKSZ structural constraints \eqref{e:struct_constrAKSZ1} and \eqref{e:struct_constrAKSZ2} are. The fact that $\underline{\varphi}$ is a symplectomorphism preserving the action also proves indirectly that $\FF_{PC}^{\res}$ is a BV theory.
\end{remark}

\begin{proof}
Denoting by $\{\be,\bom,\bc,\bxi\}$ the fields in $\calF_{PC}$ (their antifields with a dagger), and following the notation of Equation \eqref{e:variables_PC_AKSZ} for the variables in $\mathcal{F}^{\AKSZ}_{PC_\star}$, we define the map $\varphi: \calF^{\AKSZ}_{PC_\star}\rightarrow \calF_{PC}$ in terms of the splitting \eqref{e:PCfieldsplitting} (with $\varphi^*$ implicit on the right hand sides):
\begin{subequations} \label{e:cyl-variables}
\begin{eqnarray}
\varphi^*\be = \te + \underline{\te}_n  
    &  \varphi^*\bom = \tom + \underline{\tom}_n 
        & \varphi^*\be^{\dag}= \te{}^{\dag} + \underline{\te}_n{}^{\!\!\!\dag} 
 \\  
\varphi^*\bom^{\dag} = \tom{}^{\dag} + \underline{\tom}_n{}^{\!\!\!\dag} 
    & \varphi^*\bc= \tc 
        & \varphi^*\bc^{\dag} = \underline{\tc}_n{}^{\!\!\!\dag} \\
\varphi^*\bxi = \txi  + \underline{\txi}{}^n
    & \varphi^*\bxi^{\dag}=\underline{\txi}{}^{\dag}+ \underline{\txi}_n{}^{\!\!\!\dag}
\end{eqnarray}
\end{subequations}
where, using again the underlined notation to signify that the quantity is contains $dx^n$ or $\partial_n$, and $a\in\{1,2,\dots, N-1\}$:
{\small 
\begin{subequations}\label{e:symplectomorphismNd}
\begin{eqnarray}
 \te = e + \lambda \mu^{-1} f^{\dag}
 & \underline{\te}_n =  \epsilon_n \underline{\mu}+\iota_{\underline{z}} e + \lambda \mu^{-1}z^a \underline{f}^{\dag}_a\\
 \tom = \omega  - \lambda \mu^{-1} u^{\dag} 
 & \underline{\tom}_n = \underline{w} - \iota_\xi \underline{u}^\dag  - \lambda \mu^{-1}z^a \underline{u}^{\dag}_a\\
 \te{}^\dag = e^\dag -\lambda \mu^{-1} y_n^\dag 
 & \underline{\te}_n{}^{\!\!\!\dag} =  e^{N-3} \underline{u}^\dag + \iota_{\underline{z}} e^\dag -\lambda \mu^{-1}z^a \underline{y}_a^\dag + (N-3)e^{N-4}\lambda \mu^{-1}f^{\dag}\underline{u}^\dag \label{e:symplectomorphismNd_edag}\\
 \tom{}^\dag= k^\dag 
 & \underline{\tom}_n{}^{\!\!\!\dag}= e^{N-3} \underline{f}^\dag + \iota_{\underline{z}} k^\dag + \iota_\xi\underline{c}^\dag \label{e:symplectomorphismNd_omdag}\\
 \tc= c - \lambda \mu^{-1} \iota_{\xi}u^{\dag}
 & \underline{\tc}_n{}^{\!\!\!\dag}= \underline{c}^\dag\\
 \txi{}^{a}= \xi^a + \lambda \mu^{-1}z^a 
 &  \underline{\txi}{}^{\dag}= e \underline{y}^\dag + \underline{f}^\dag e^\dag - \underline{u}^\dag k^\dag + \underline{c}^\dag \lambda \mu^{-1} u^{\dag}\label{e:symplectomorphismNd_yadag}\\
 \underline{\txi}{}^{n}= \txi{}^n\partial_n= \lambda \mu^{-1} \partial_n& \underline{\txi}_n{}^{\!\!\!\dag}= e_n \underline{y}^\dag + e^{N-3} f^\dag \underline{u}^\dag+ f^\dag\iota_{\underline{z}} e^\dag +u^\dag \iota_{\underline{z}} k^\dag + c^\dag \lambda \mu^{-1}z^a \underline{u}^{\dag}_a \label{e:symplectomorphismNd_yndag}
\end{eqnarray}
\end{subequations}}
The explicit, long but straightforward calculation needed to prove that $\varphi$ is an inclusion of symplectic manifolds preserving the action functionals, i.e. such that $\varpi_{PC}^{\AKSZ}=\varphi^*\varpi_{PC}$ and $S_{PC}^{\AKSZ}=\varphi^*S_{PC}$, is given in Appendix \ref{appendix:Nd-PC-comparison}.

We then need to prove that $\mathrm{Im}(\varphi) = \calF_{PC}^{\res}$. In other words we want to prove that the map defined in  \eqref{e:symplectomorphismNd} will map a solution of the constraints \eqref{e:struct_constrAKSZ1} and \eqref{e:struct_constrAKSZ2} into a solution of \eqref{e:PCconstraints1} and \eqref{e:PCconstraints2}. Applying \eqref{e:symplectomorphismNd} to the definition of $\mathfrak{W}^\dag$ as given in Equation \eqref{e:defWdag} we get: 
\begin{align*}
    \mathfrak{W}^\dag & = \underline{\tom}^{\dag}_n - \tom_a{}^{\!\!\!\dag} {\te}_n^{\{a\}}dx^n- \iota_{\txi}\underline{\tc}_n{}^{\!\!\!\dag}+{\tc}_{an}{}^{\!\!\!\dag}{\txi}{}^n{\te}_n^{\{a\}}dx^n \\
    &=  e^{N-3} \underline{f}^\dag + \iota_{\underline{z}} k^\dag + \iota_\xi\underline{c}^\dag - k_a^\dag \underline{z_a} - \iota_{\xi}\underline{c}^\dag - \underline{c}_a^\dag\lambda \mu^{-1}z^a +  {c}_a^\dag\lambda \mu^{-1}\underline{z_a} \\
    &= e^{N-3} \underline{f}^\dag,
\end{align*}
which is \eqref{e:PCconstraints1}.
On the other hand, constraint \eqref{e:PCconstraints2} is satisfied if \eqref{e:struct_constrAKSZ1} and \eqref{e:struct_constrAKSZ2} are, as it can be seen by direct inspection: using \eqref{e:symplectomorphismNd} we get\footnote{Note that $X^{\{a\}}=X^{(a)}-(X^{(b)} \lambda \mu^{-1}f_b^\dag)^{(a)}$.}
\begin{align*}
    {\epsilon}_n \te{}^{(N-4)} d_{\tom} \te & - {\epsilon}_n \te{}^{(N-4)} {W}_{\te{}^{N-3}}^{-1}(\mathfrak{W}^\dag) d \txi{}^n \\
    &+ ( [\tc, \epsilon_n]+ L_{\txi{}}^{\tom}(\epsilon_n)-d_{\tom_n} \epsilon_n]\txi^{n})^{\{a\}}(\tom_a{}^{\!\!\!\dag} -{\tc}_{an}{}^{\!\!\!\dag}{\txi}{}^n)\\
    &= {\epsilon}_n e^{N-4} d_{\omega} e + (N-4){\epsilon}_n e^{N-5} \lambda \mu^{-1}f^\dag d_{\omega} e + {\epsilon}_n e^{N-4} [\lambda \mu^{-1}u^\dag, e] \\
    &+ \epsilon_n e^{N-4} d_{\omega} (\lambda \mu^{-1}f^\dag)+ (N-4)\epsilon_n e^{N-5} \lambda \mu^{-1}f^\dag d_{\omega} (\lambda \mu^{-1}f^\dag) \\
    & +  \epsilon_n f^\dag e^{N-4}  d_{\omega} (\lambda \mu^{-1})+ (N-4) \epsilon_n f^\dag e^{N-5} \lambda \mu^{-1}f^\dag d_{\omega} (\lambda \mu^{-1})\\
    &-\left([c, \epsilon_n]-[c, \epsilon_n]^{(b)}\lambda \mu^{-1}f_b^\dag +[\lambda \mu^{-1}\iota_\xi u^\dag, \epsilon_n]\right)^{(a)}(k_a^\dag + c_a^\dag \lambda \mu^{-1})\\
      &+\left(L_{\xi}^{\omega}(\epsilon_n)-L_{\xi}^{\omega}(\epsilon_n)^{(b)}\lambda \mu^{-1}f_b^\dag +[\lambda \mu^{-1}\iota_\xi u^\dag, \epsilon_n]\right)^{(a)}(k_a^\dag + c_a^\dag \lambda \mu^{-1})\\
      &-[w-\iota_\xi u^\dag, \epsilon_n]^{(a)}k_a^\dag\lambda \mu^{-1}\\
      &= \epsilon_n  e^{N-4} d_{\omega} e + \left(L_{\xi}^{\omega}(\epsilon_n) -[c, \epsilon_n]\right)^{(a)}  k^\dag_a \\
      &+ \lambda \mu^{-1}\Big(\epsilon_n \left\{(N-4)f^\dag e^{N-5} d_\omega e + e^{N-4} d_{\omega} f^\dag + e^{N-4} [u^\dag, e]\right\} \\
   &+ \left(\iota_z d_{\omega}\epsilon_n-[w-\iota_\xi u^\dag, \epsilon_n]\right)^{(a)}  k^\dag_a\\
   &+ \left(L_{\xi}^{\omega}(\epsilon_n) -[c, \epsilon_n]\right)^{(a)}  c^\dag_a + \left(\left(L_{\xi}^{\omega}(\epsilon_n) -[c, \epsilon_n]\right)^{(b)}f^\dag_b\right)^{(a)}k^\dag_a\Big) = (\spadesuit).
\end{align*}
Using now the AKSZ constraints \eqref{e:struct_constrAKSZ1} and \eqref{e:struct_constrAKSZ2} we obtain 
\begin{align*}
    (\spadesuit) &= e^{N-3} \sigma + \lambda \mu^{-1} (f^\dag e^{N-4} \sigma + e^{N-3} B)\\
    & =\te^{N-3} ( \sigma + \lambda \mu^{-1} B).
\end{align*}
Comparing the first and the last line of this computation we get the desired constraint \eqref{e:PCconstraints2}.
Hence $\varphi$ defines a diffeomorphism $\underline{\varphi}\colon \calF^{\AKSZ}_{PC_\star} \to \calF_{PC}^{\res}$.
% This map is also injective. 
% \todo{Iniettivita puo essere dimostrata direttamente oppure mostrando un inverso. Intendi $\phi$ o $\underline{\phi}$?}
Indeed, the inverse of this map is readily found, as follows. 

%First, observe that we want to compare two theories for which three bases are given: standard PC with $\{e_a, e_n\}$, the reduced space $\calF_R$ with $\{e_a,\epsilon_n\}$ and AKSZ with $\{\te,\epsilon_n\}$, for $\epsilon_n$ a fixed vector field $\epsilon_n\in\Gamma(\mathcal{V})$.  Let us decompose 
%$$
%\epsilon_n = \epsilon_n^{[a]}e_a + \epsilon_n^{[n]}e_n; \qquad \te_n = \te_n{}^{\{a\}}\te_a + \te_n{}^{\{n\}}\epsilon_n
%$$
It is easy to find $k^\dag=\tom{}^\dag$, $\underline{c}^\dag=\underline{\tc}_n{}^{\!\!\!\dag}$, and $\txi{}^n=\lambda\mu^{-1}$. Then we can write $e = \te - {\txi}{}^{n}f^\dag$, so that ${\te}_n = \epsilon_n\mu + \iota_z \te$, and taking $\{\te_a,\epsilon_n\}$ as a basis, we have $z^a = \te_n{}^{a}$ and $\mu = \te_n{}^n$, which also implies $\lambda = \te_n{}^n \txi{}^n $ and $\xi^a = \txi{}^a - \te_n{}^a\txi{}^n$.

We now turn to equation \eqref{e:symplectomorphismNd_omdag} which can be rewritten as 
$$
e^{N-3} \underline{f}^\dag = \underline{\tom}_n{}^{\!\!\!\dag} - \iota_{\underline{z}}k^\dag - \iota_\xi\underline{c}^\dag
$$
Let us denote the known piece by $\underline{\tilde{A}}:=- \iota_{\underline{z}}k^\dag - \iota_\xi\underline{c}^\dag$, so that we have
$$
\begin{cases}
e = \te - \txi{}^n f^\dag\\
e^{N-3} \underline{f}^\dag = \underline{\tom}_n{}^{\!\!\!\dag} + \underline{\tilde{A}}
\end{cases} \Longrightarrow \te{}^{N-3} \underline{f}^\dag= \underline{\tom}_n{}^{\!\!\!\dag} + \underline{\tilde{A}}
$$
where we used that $f^\dag f^\dag = 0$. We see here that this equation can be solved only when 
$$
\underline{\tom}_n{}^{\!\!\!\dag} - \iota_{\underline{z}}\tom{}^\dag - \iota_\xi\underline{\tc}{}^\dag \in \mathrm{Im}({W}^{1,1}_{\te{}^{N-3}}).
$$
From the equations 
$$
\te^{\dag} = e^\dag - \txi{}^n y^\dag, \qquad \underline{\te}_n{}^{\!\!\!\dag}=e^{N-3}\underline{u}^\dag + \iota_{\underline{z}}e^\dag - \lambda \mu^{-1}z^a\underline{y}_a^\dag + (N-3)e^{N-4}\lambda \mu^{-1}f^{\dag}\underline{u}^\dag,
$$
using again $e = \te - \txi{}^n f^\dag$, we get 
$$
\te{}^{N-3}\underline{u}^\dag = \underline{\te}_n{}^{\!\!\!\dag} - \iota_{\underline{z}}\te{}^\dag.
$$
Since $\underline{\te}_n{}^{\!\!\!\dag} - \iota_{\underline{z}}\te{}^\dag\in \Omega^{(N-2,N-1)}$, on which the map ${W}_{\te{}^{N-3}}$ is surjective, we conclude that, up to components $p(\underline{u}^\dag)$ in the kernel of ${W}_{\te{}^{N-3}}$, we can find 
$$u^\dag = {W}_{\te{}^{N-3}}^{-1}(\underline{\te}_n{}^{\!\!\!\dag} - \iota_{\underline{z}}\te{}^\dag) + p\underline{u}^\dag$$
However, we know that $u^\dag$ must satisfy the constraint \eqref{e:struct_constrAKSZ2}, which (impliclty but uniquely) fixes $p\underline{u}^\dag$ as a function of $\te, \tom, f^\dag$. We can use this directly to solve 
$$
\omega = \tom + \txi{}^n u^\dag = \tom + \txi{}^n\left( {W}_{\te{}^{N-3}}^{-1}(\underline{\te}_n{}^{\!\!\!\dag} - \iota_{\underline{z}}\te{}^\dag) + p\underline{u}^\dag\right)
$$
Analogously we can find $\underline{w}$  and $c$ as follows
$$
\underline{w} = \underline{\tom}_n + \iota_{\txi}\left( {W}_{\te{}^{N-3}}^{-1}(\underline{\te}_n{}^{\!\!\!\dag} - \iota_{\underline{z}}\te{}^\dag) + p\underline{u}^\dag\right),
$$
$$
c = \tc +\txi{}^n\iota_{\txi}\left( {W}_{\te{}^{N-3}}^{-1}({\te}_n{}^{\!\!\!\dag} - \iota_{\underline{z}}\te{}^\dag) + p{u}^\dag\right).
$$
Finally, we can conclude the calculation with $y^\dag$ and $e^\dag$ by inverting \eqref{e:symplectomorphismNd_edag}, \eqref{e:symplectomorphismNd_yadag} and \eqref{e:symplectomorphismNd_yndag}: it is useful to notice that it is possible to invert an equation of the form $e^\dag(1+ \lambda X)=Y$ for some $X,Y$ as $e^\dag=Y(1-\lambda X)$. However, we will not write down in full these last equations as we will not need them in what follows.

The BV theory $\FF^{\AKSZ}_{PC_\star}$ is obviously strongly equivalent to its image under the symplectomorphism $\varphi$, which is $\FF_{PC}^{\res}$. Furthermore, since up to boundary terms $Q_{PC}$ is the Hamiltonian vector field of $S_{PC}$, and the same holds for $Q^{\AKSZ}_{PC}$ and $S_{PC}^{\AKSZ}$, we have that in the interior $M^\circ = M \backslash \partial M$ the compatibility $\varphi^* Q_{PC} = Q_{PC}^{\AKSZ} \varphi^*$ is a consequence of $\varphi^*\varpi_{PC}=\varpi^{\AKSZ}_{PC}$ and $\varphi^*S_{PC}=S^{\AKSZ}_{PC}$. However, this is a local condition that then extends to the whole of $(M,\partial M)$ and $\varphi$ is a BV inclusion.
\end{proof}

\begin{remark} \label{rmk:nondegenerateAKSZe}
The defining condition $\mu \neq 0$ and ${g}_{\text{hor}}=\te^*\eta$ nondegenerate given in Definition \ref{def:ndPCAKSZ}, used in Theorem \ref{thm:comparison}, are necessary in order to make $\be$ non degenerate in the bulk, to build the symplectomorphism \eqref{e:symplectomorphismNd} (since $\epsilon_n^{[n]} = \mu^{-1}$).
\end{remark}
}

\begin{remark} \label{rmk:freecomponents}
The number of free components of $\iota_R^*\tom$ is $\frac{3 N(N-1)}{2}$, since $\omega$ and $w$ have respectively $N(N-1)$ and $\frac{N(N-1)}{2}$ free components. The $\frac{N(N-1)(N-3)}{2}$ missing components are those fixed by the condition in Equation \eqref{e:PCconstraints2}. Correspondingly, also $\iota_R^*\tom{}^\dag$ has $\frac{3 N(N-1)}{2}$ independent components: $\frac{N(N-1)}{2}$ coming from $k^\dag$ and $ N(N-1)$ from $f^\dag$. 
\end{remark}

{
\subsection{An interpretation of the restricted theory} \label{sec:Interpretation}
We now want to shed some light on the interpretation of the restricted theory $\FF^{\res}_{PC}$ defined in Theorem \ref{thm:comparison}.  

Recall that among the Euler--Lagrange equations of the classical PC theory we have\footnote{We use boldface letters to denote fields in PC theory.} $\be^{N-3}d_{\bom} \be =0$, which, thanks to the nondegeneracy of $\be$, is equivalent to $d_{\bom} \be=0$, i.e., the torsion-free condition for $\bom$. Imposing this condition forces $\bom$ to correspond to the Levi-Civita connection for the metric $g_{\mu\nu}=\eta(\be_\mu,\be_\nu)$, which is used to recover the Einstein--Hilbert formulation of the theory. Note that this yields only a classical equivalence of the two theories, as the fluctuations might violate the condition $d_{\bom} \be=0$ at the quantum level. Only by forcing this condition on the space of fields (i.e., by freezing the fluctuations that might violate it) may one recover the quantum Einstein--Hilbert theory.

However, one can consider a whole family of theories between PC and EH where only some part of the condition $d_{\bom} \be=0$ is imposed on the fields, looking for a compromise.\footnote{Imposing too few conditions out of $d_\omega e=0$ would not solve the compatibility problem with the boundary. Imposing too many generates other problems (see, e.g., \cite[Section 4.3]{CS2019}, where the whole of $d_\omega e=0$ is imposed manually).} that retains the good feature of PC of dealing with differential forms but yields a compatible boundary BFV theory as in EH \cite{CS2016b}.

In particular, working on a cylinder $I\times\Sigma$, we may use the decomposition $\be= \underline{ \te}_n+ \te$,
 $\bom= \underline{\tom}_n + \tom$. By choosing once and for all a nonzero section $\epsilon_n\in\Gamma(M,\mathcal{V})$ and requiring the components of $\te$ to span a transversal hyperplane in $\mathcal{V}$ at each point, we may expand $\underline{\te}_n$ in the basis $(\epsilon_n,\te)$; moreover, we require $\Tilde e$ to define a nondegenerate metric $\eta(\te,\te)$ at each point.\footnote{A more physical requirement, as one would like the two boundary components of $I\times\Sigma$ to be space-like Cauchy surfaces, consists in choosing $\epsilon_n$ to be a time-like section and $\te$ to define a positive definite metric.} Observe that the splitting of fields $\be,\bom$ induced by the cylinder structure also allows the definition of the maps $W^{i,j}_{\te{}^{N-3}}$ given in Equation \eqref{e:W-maps}.
 
With these notations we may impose the constraint 
\begin{equation}\label{e:classical constraint}
\epsilon_n\te^{N-3}d_{\tom}\te\in\mathrm{Im}({W}_{\te{}^{N-3}}^{1,1}),
\end{equation}
which implements only some of the conditions in $d_{\bom} \be=0$. 

Another interpretation of the constraints goes through considering a reduction of the fields instead of a restriction. Indeed we can also think of Equation \eqref{e:classical constraint} as a classical constraint that freezes certain components of the connection. We need the following
\begin{definition}\label{def:reducedPCfield}
We define the space of reduced connections on a cylinder to be the quotient 
\begin{equation}
    \mathcal{A}^{\mathrm{red}}(\Sigma\times I)\coloneqq \mathcal{A}(\Sigma\times I)/\mathrm{ker}({W}_{\te{}^{N-3}}^{1,2}),
\end{equation}
and denote by $F_{PC}^{\res}$ the fiber bundle
\begin{equation}
    F_{PC}^{\res} \longrightarrow \Omega^1_{nd}(\Sigma\times I,\mathcal{V})
\end{equation}
with typical fiber $\mathcal{A}^{\mathrm{red}}(\Sigma\times I)$ obtained by reducing the fibers of the trivial bundle 
$$
\mathcal{A}(\Sigma\times I)\times\Omega^1_{nd}(\Sigma\times I,\mathcal{V}) \longrightarrow \Omega^1_{nd}(\Sigma\times I,\mathcal{V})
$$ 
by $\mathrm{ker}(W_{\te{}^{N-3}}^{1,2})$.
\end{definition}

\begin{proposition}\label{prop:PCstructural}
Consider the splitting $\be=\te + \underline{\te}_n$, with ${g}_{\text{hor}}\coloneqq \te{}^*\eta$ nondegenerate. Then for every $([\omega],\te,\underline{\te}_n)$ there exists a unique $\omega \in \mathcal{A}(\Sigma\times I)$ such that 
\begin{equation}\label{e:PCreductionconstaint}
(N-3)\epsilon_n \wedge \te{}^{N-4} \wedge d_{\omega} \te \in \mathrm{Im}(W_{\te{}^{N-3}}^{1,1}),
\end{equation}
which induces a section of the fibration:
\begin{equation}\label{e:PCfibsect}
    \mathcal{A}(\Sigma\times I)\times\Omega^1_{nd}(\Sigma\times I,\mathcal{V}) \longrightarrow F^{\res}_{PC},
\end{equation}
\end{proposition}

\begin{proof}
This is a straightforward adaptation of \cite[Theorem 17]{CCS2020}, which holds at every point in $I$.
\end{proof} 

Hence, imposing only some part of the equation $d_{\bom}\be=0$ produces an intermediate theory, that in view of Proposition \ref{prop:PCstructural} can be alternatively thought of as Palatini--Cartan theory for a tetrad and a \emph{reduced} connection. However, in both interpretations, fixing a condition only on the classical fields does not produce a symplectic submanifold of the space of BV fields.

If we want to consistently restrict the BV theory of Definition \ref{def:standardPC-BV} we first have to impose some condition on the antifields as well, in order to ensure that we have a nondegenerate BV form. One can show that \eqref{e:classical constraint} actually fixes the components of $\tom$ in the kernel of ${W}^{1,2}_{\te{}^{N-3}}$. As a consequence, we can get a symplectic submanifold if, in addition to \eqref{e:classical constraint}, we impose\footnote{This condition requires that the antifield of $\tom$ be in the dual of the complement of the kernel of $W_{\te{}^{N-3}}^{1,2}$, which is the image of $W_{\te{}^{N-3}}^{1,1}$ because $\tom$ is tangent to the slice $\Sigma\times \{t\}$, and its antifield is of the form $\tom^\dag_n dx^n$, with $\tom^\dag_n\in\Omega^{N-2}(M,\wedge^{N-2}\mathcal{V})$.}
\begin{equation}\label{e:complementary constraint}
{\tom}^{\dag}_n\in \mathrm{Im}({W}_{\te{}^{N-3}}^{1,1}). 
\end{equation}

The problem, though, is that \eqref{e:classical constraint} and \eqref{e:complementary constraint} do not define a $Q$-submanifold, which is needed to have a BV theory. However, one can easily check that condition \eqref{e:classical constraint} is compatible with gauge transformations and diffeomorphisms upon using the Euler--Lagrange equations. This implies that it should be possible to correct \eqref{e:classical constraint}---and concurrently \eqref{e:complementary constraint} because we want to preserve the condition that we get a symplectic submanifold---so as to obtain a $Q$-submanifold. The explicit solution to this problem is actually given by \eqref{e:PCconstraints1} and \eqref{e:PCconstraints2}.

\begin{remark}
Observe that this solution might not be unique, as the choice of a structural constraint we made in Definition \ref{def:BFVspaceoffields} was made to render the invariance of \eqref{e:BFVstructuralConstraint} more manifest. However, Theorem \ref{thm:comparison} tells us that a different choice of BFV structural constraint will provide a different extension of the constraint \eqref{e:PCreductionconstaint} in Palatini--Cartan theory.
\end{remark}

%We would like to shed some light here on the interpretation of the reduced theory $\FF^{\res}_{PC}$ defined in Theorem \ref{thm:comparison}. In order to do this, we will start from classical PC theory on a cylinder, for which we have the splitting of fields $\be = \te + {\te}_ndx^n$, $\bom = \tom + {\tom}_ndx^n$, which .

%Now, when considering the BV extension of PC theory (Definition \ref{def:standardPC-BV}), the natural question that arises is whether we can extend the constraint \eqref{e:PCreductionconstaint} to the vanishing locus of a function $F\colon \calF_{PC} \to \mathbb{R}$ such that
%$$
%L_{Q_{PC}} F = 0 \mod F.
%$$
%In other words, we ask ourselves if there is a version of the constraint that is $Q_{PC}$ invariant and that restricts to \eqref{e:PCreductionconstaint} when restricted to degree-$0$ fields. The answer is positive, and it is given by the PC structural constraints of Equations \eqref{e:defFR}, which then describe the conormal bundle to the image of the section of the fibration \eqref{e:PCfibsect}.

}

\subsection{Three dimensional case}
When $N=3$ some simplifications occur. Indeed, in this case the inclusion is actually an identity since there are no additional constraints on the field. Furthermore we know that the theory is strongly BV-equivalent, both in the bulk and on the boundary, to the topological $BF$ theory, denoted here by $\FF^{\AKSZ}_{BF'}$. Hence we can summarize the results in the following theorem.
\begin{corollary}\label{cor:3dcase}
The theories $\mathfrak{F}^{\AKSZ}_{PC_\star}$ and $\mathfrak{F}^{\AKSZ}_{BF'}$ are strongly BV equivalent.
\end{corollary}
\begin{proof}
The claim follows directly from Theorem \ref{strongequivalenceAKSZ} given the results of Theorem \ref{thm:comparison} and of \cite{CaSc2019}, which proves the strong equivalence (at all codimensions) of non-degenerate BF theory and PC gravity in three dimensions. 
\end{proof}
Pictorially we can describe the content of Corollary \ref{cor:3dcase} as follows
\begin{center}
\begin{equation}\label{Commdiag}
\begin{tikzcd}[ row sep= 3 em, column sep= 4 em]
\calF_{PC} \arrow[dd, xshift= -1em, bend right= 50, "B", dashed]\arrow[d, leftrightarrow, "\phi"] \arrow[r, leftrightarrow, "\psi"] &\calF_{BF'} \arrow[d, equal]\\
\calF_{PC^\star}^{\AKSZ} \arrow[d, bend right= 50, "B", dashed]&\calF_{BF'}^{\AKSZ} \arrow[d, bend right= 50, "B", dashed]\\
\calF_{PC}^{\partial}\arrow[u, bend right= 50, "A", dashed]\arrow[r, leftrightarrow, "\psi^{\partial}"]&\calF_{BF'}^{\partial}\arrow[u, bend right= 50, "A" , dashed]
\end{tikzcd}
\end{equation}
\end{center}
where the arrows $A$ represent the AKSZ constructions, the arrows $B$ represent the BV-BFV reductions, while $\psi$, $\psi^{\partial}$ and $\phi$ are the symplectomorphisms mentioned above. 

\appendix
\section{Lengthy calculations} \label{appendix:Nd-PC-comparison}
We prove here that the transformation \eqref{e:symplectomorphismNd} is a symplectomorphism between $\calF^{\AKSZ}$ and $\calF_R$ that preserves the action. In the computation we will use multiple times the following useful relation:
\begin{align*}
\epsilon^{[a]}_n = - z^a \epsilon^{[n]}_n, \qquad \epsilon^{[n]}_n = \mu^{-1}.
\end{align*}

\input{computation_sympl_short}

\input{computation_action_short}	
\printbibliography
\end{document}

%% file: computation_sympl_short.tex
We now prove that $\varphi^*\varpi_{PC} =  \varpi_{PC}^{\AKSZ}.$
\begin{align}\label{e:intermediate_symplNd}
\varphi^*\varpi&_{PC}= \varphi^*\int_M \delta \be \delta \be {}^{\dag} + \delta \bom \delta \bom {}^{\dag}+ \delta \bc \delta \bc {}^{\dag} + \iota_{\delta \bxi} \delta\bxi {}^{\dag}  \nonumber\\
	& =  \int_M \delta \te \delta \underline{\te}_n{}^{\!\!\!\dag} + \delta \underline{\te}_n \delta \te{}^{\dag}+ \delta \tom \delta \underline{\tom}_n{}^{\!\!\!\dag}+ \delta \underline{\tom}_n \delta \tom{}^{\dag}+ \delta \tc \delta \underline{\tc}{}^{\dag} + \delta \txi{}^{a} \underline{\txi}_a{}^{\!\!\!\dag} + \delta \txi{}^{n} \underline{\txi}_n{}^{\!\!\!\dag } \\
& =  \int_M   \delta e  \delta (e^{N-3} \underline{u}^\dag)
	 + \delta e\delta(\iota_{\underline{z}} e^\dag) 
	 - \delta e\delta(\lambda \epsilon_n^{[a]} \underline{y}_a^\dag) 
	 +\delta e\delta((N-3)e^{N-4} \lambda \epsilon_n^{[n]}f^{\dag} \underline{u}^\dag ) \nonumber\\ 
	& \qquad +\delta(\lambda \epsilon_n^{[n]} f^{\dag})\delta (e \underline{u}^\dag)  
	 + \delta(\lambda \epsilon_n^{[n]} f^{\dag})\delta(\iota_{\underline{z}} e^\dag) 
	 -\delta(\lambda \epsilon_n^{[n]} f^{\dag})\delta(\lambda \epsilon_n^{[a]} \underline{y}_a^\dag) \nonumber\\
	& \qquad +\delta(\lambda \epsilon_n^{[n]} f^{\dag})\delta(\lambda \epsilon_n^{[n]}f^{\dag}\underline{u}^\dag )  
	 + \delta(\iota_{\underline{z}} e)  \delta e^\dag 
	 -\delta(\iota_{\underline{z}} e)\delta(\lambda \epsilon_n^{[n]} y^\dag)
	 + \delta( \epsilon_n \underline{\mu})  \delta e^\dag \nonumber\\
	& \qquad -\delta( \epsilon_n \underline{\mu})\delta(\lambda \epsilon_n^{[n]} y^\dag) 
	 + \delta(\lambda \epsilon_n^{[a]} \underline{f}^{\dag}_a)\delta e^\dag 
	 -\delta(\lambda \epsilon_n^{[a]} \underline{f}^{\dag}_a)\delta(\lambda \epsilon_n^{[n]} y^\dag)\nonumber\\
	& \qquad + \delta \omega  \delta(e^{N-3} \underline{f}^\dag) 
 	 + \delta \omega\delta(\iota_{\underline{z}} k^\dag) 
 	 + \delta \omega\delta(\iota_\xi \underline{c}^\dag)  
 	 - \delta(\lambda \epsilon_n^{[n]} u^{\dag}) \delta(e^{N-3} \underline{f}^\dag) \nonumber\\
	& \qquad - \delta(\lambda \epsilon_n^{[n]} u^{\dag})\delta(\iota_{\underline{z}} k^\dag) 
	 - \delta(\lambda \epsilon_n^{[n]} u^{\dag})\delta(\iota_\xi \underline{c}^\dag)  
	 + \delta \underline{w} \delta k^\dag
	 - \delta (\iota_\xi \underline{u}^\dag )\delta k^\dag \nonumber\\
	& \qquad - \delta(\lambda \epsilon_n^{[a]} \underline{u}^{\dag}_a)\delta k^\dag 
	 + \delta c \delta \underline{c}^\dag
	 - \delta(\iota_{\xi} \lambda \epsilon_n^{[n]} u^{\dag})\delta \underline{c}^\dag  
	 + \iota_{\delta \xi} \delta (e \underline{y}^\dag) \nonumber\\
 	& \qquad + \iota_{\delta \xi}\delta (\underline{f}^\dag e^\dag)
 	 - \iota_{\delta \xi}\delta (\underline{u}^\dag k^\dag) 
 	 + \iota_{ \delta \xi}\delta(\underline{c}^\dag \lambda \epsilon_n^{[n]} u^{\dag}) 
 	 + \delta (\lambda \epsilon^{[a]}_n) \delta (e_a \underline{y}^\dag) \nonumber\\
 	 \intertext{}
	& \qquad + \delta (\lambda \epsilon^{[a]}_n)\delta (\underline{f}_a^\dag e^\dag)
	 + \delta (\lambda \epsilon^{[a]}_n)\delta (\underline{u}_a^\dag k^\dag) 
	 + \delta (\lambda \epsilon^{[a]}_n) \delta(\underline{c}_a^\dag \lambda \epsilon_n^{[n]} u^{\dag})\nonumber\\
	& \qquad + \delta (\lambda \epsilon^{[n]}_n) \delta (\be_n \underline{y}^\dag )
	 + \delta (\lambda \epsilon^{[n]}_n) \delta (e^{N-3} f^\dag \underline{u}^\dag)
	 + \delta (\lambda \epsilon^{[n]}_n) \delta ( f^\dag\iota_{\underline{z}} e^\dag)\nonumber\\
	& \qquad + \delta (\lambda \epsilon^{[n]}_n) \delta (u^\dag \iota_{\underline{z}} k^\dag) 
	 + \delta (\lambda \epsilon^{[n]}_n) \delta (c^\dag \lambda \epsilon_n^{[a]} \underline{u}^{\dag}_a)\nonumber
\end{align}

This expression should be compared with the symplectic form coming from the AKSZ construction:
\begin{align}\label{e:sympl_form_AKSZ_Nd_ul}\setcounter{terms}{0}
\varpi^{\AKSZ}_{PC} = \int_{I\times \partial M}& 
	   \delta (e^{N-3} f^\dag)  \delta \omega
	 + e^{N-3} \delta e \delta u^\dag 
	 + \delta w \delta k^\dag
	 + \delta c \delta c^\dag 
	 + \delta u^\dag \delta(\iota_\xi k^\dag)\nonumber \\
	& + \delta \omega \delta(\iota_z  k^\dag)
	 + \delta \omega \delta(\iota_\xi c^\dag) 
	 - \delta \mu \epsilon_n \delta e^\dag
	 - \delta \lambda \epsilon_n \delta y^\dag \nonumber \\
	& +\iota_{\delta z}\delta (e e^\dag)	
	 +\iota_{\delta \xi} \delta (f^\dag e^\dag)
	 +\iota_{\delta \xi} \delta (e y^\dag).
\end{align}
Almost all the terms in \eqref{e:sympl_form_AKSZ_Nd_ul} can be direclty found in $\varphi^*\varpi_{PC}$.
The remaining terms can be identified using the following relations:
\begin{align*}
\delta \underline{u}^\dag \delta(\iota_\xi k^\dag) &=- \delta (\iota_\xi \underline{u}^\dag) \delta k^\dag - \iota_{\delta \xi}\delta( \underline{u}^\dag  k^\dag);\\
-\delta( \lambda \epsilon_n )\delta \underline{y}^\dag & = \delta(e_a \lambda \epsilon^{[a]}_n )\delta 
\underline{y}^\dag+\delta(\te_n \lambda \epsilon^{[n]}_n )\delta \underline{y}^\dag \\
 & = \delta( \lambda \epsilon^{[a]}_n )\delta (e_a\underline{y}^\dag) - \delta e \delta( \lambda \epsilon^{[a]}_n \underline{y}_a^\dag) \\
 & \; +  \delta( \lambda \epsilon^{[n]}_n )\delta (\te_n \underline{y}^\dag) - \delta \underline{\te}_n \delta( \lambda \epsilon^{[n]}_n y^\dag);\\
\iota_{\delta \underline{z}}\delta (e e^\dag) &=  \delta e  \delta ( \iota_{\underline{z}}e^\dag) +  \delta ( \iota_{\underline{z}}e)\delta e^\dag
\end{align*}

All the other terms in \eqref{e:intermediate_symplNd} sum to zero because of the following identities:

\begin{align*}
& \delta(\lambda \epsilon_n^{[n]} f^{\dag})\delta (e^{N-3} \underline{u}^\dag) +(N-3)\delta e\delta(e^{N-4}\lambda \epsilon_n^{[n]}f^{\dag} \underline{u}^\dag ) \\
 & -\delta(\lambda \epsilon_n^{[n]} u^{\dag}) \delta(e^{N-3} \underline{f}^\dag)+ \delta (\lambda \epsilon^{[n]}_n) \delta (e^{N-3} f^\dag \underline{u}^\dag)=0;\\
& \delta(\lambda \epsilon_n^{[n]} f^{\dag})\delta(\iota_{\underline{z}} e^\dag)+ \delta(\lambda \epsilon_n^{[a]} \underline{f}^{\dag}_a)\delta e^\dag+ \delta (\lambda \epsilon^{[a]}_n)\delta (\underline{f}_a^\dag e^\dag)+ \delta (\lambda \epsilon^{[n]}_n) \delta ( f^\dag\iota_{\underline{z}} e^\dag)=0;\\
&-\delta(\lambda \epsilon_n^{[n]} u^{\dag})\delta(\iota_{\underline{z}} k^\dag)- \delta(\lambda \epsilon_n^{[a]} \underline{u}^{\dag}_a)\delta k^\dag+ \delta (\lambda \epsilon^{[a]}_n)\delta (\underline{u}_a^\dag k^\dag)+ \delta (\lambda \epsilon^{[n]}_n) \delta (u^\dag \iota_{\underline{z}} k^\dag ) =0;\\
&-\delta(\lambda \epsilon_n^{[n]} u^{\dag})\delta(\iota_\xi \underline{c}^\dag)-\delta(\iota_{\xi} \lambda \epsilon_n^{[n]} u^{\dag})\delta \underline{c}^\dag+ \iota_{ \delta \xi}\delta(\underline{c}^\dag \lambda \epsilon_n^{[n]} u^{\dag})=0;\\
&\delta (\lambda \epsilon^{[a]}_n) \delta(\underline{c}_a^\dag \lambda \epsilon_n^{[n]} u^{\dag})+ \delta (\lambda \epsilon^{[n]}_n) \delta (c^\dag \lambda \epsilon_n^{[a]} \underline{u}^{\dag}_a)=0;\\
&-\delta(\lambda \epsilon_n^{[n]} f^{\dag})\delta(\lambda \epsilon_n^{[a]} \underline{y}_a^\dag)-\delta(\lambda \epsilon_n^{[a]} \underline{f}^{\dag}_a)\delta(\lambda \epsilon_n^{[n]} y^\dag) =0;\\
&(N-3)\delta(\lambda \epsilon_n^{[n]} f^{\dag})\delta(e^{N-4}\lambda \epsilon_n^{[n]}f^{\dag}\underline{u}^\dag )=0. 
\end{align*}

%% file: computation_action_short.tex
We go on to show that the symplectomorphism $\varphi$ preserves the action i.e. {$\varphi^* S_{PC} = S^{\AKSZ}_{PC}$.} We do it by direct inspection\footnote{We denote with $\partial_n$ the de Rham differential on $I$ (previously denoted with $d_I$) in order to be uniform with the notation of the fields \eqref{e:cyl-variables}.}:
\begin{align}\label{e:phiS_AKSZ_Nd}
\varphi^* S_{PC}  = & \varphi^*\int_{M} \frac{1}{N-2}  \be{}^{N-2} F_{\bom} + \left( d_{\bom} \bc - \iota_{\bxi} F_{\bom} \right) \bom{}^{\dag} + \left(L_{\bxi}^{\bom} \be - [\bc, \be]\right) \be{}^{\dag}\\ 
 & + \frac{1}{2}\left([\bc,\bc] - \iota_{\bxi}\iota_{\bxi} F_{\bom} \right) \bc{}^{\dag} + \frac{1}{2}\iota_{[\bxi,\bxi]} \bxi {}^{\dag} \nonumber\\
 \intertext{}
= & \int_{M}  \te{}^{N-3} \underline{\te}_n F_{\tom} + \frac{1}{N-2} \te{}^{N-2} \underline{F_{\tom_n}}-  \left(\iota_{\txi} F_{\tom}+ F_{\tom_n}\txi_{n}  - d_{\tom} \tc \right)\underline{\tom}_n{}^{\!\!\!\dag}\nonumber\\
 & - \left(\iota_{\txi{}} \underline{F_{\tom_n}}  - \underline{d_{\tom_n}} \tc \right)\tom{}^{\!\!\!\dag} 
+ \left(L_{\txi}^{\tom}\te + d_{\tom_n} \te \txi_{n} + \te_n d \txi_{n}- [\tc,\te]\right)\underline{\te}_n{}^{\!\!\!\dag}\nonumber\\ 
&+ \left(\iota_{\txi}d_{\tom}\underline{\te}_n + \iota_{\underline{\partial}_n \txi}\te - \underline{d_{\tom_n}}(\te_n \txi_{n})- [\tc,\underline{\te}_n]\right)\te{}^\dag\nonumber\\
&- \left(\frac{1}{2}\iota_{\txi}\iota_{\txi} F_{\tom} + \iota_{\txi{}} F_{\tom_n}\txi_{n}  - \frac{1}{2}[\tc,\tc]\right)\underline{\tc}{}^\dag 
+\frac12 \iota_{[\txi,\txi]}\underline{\txi{}^\dag}+ \frac12 \iota_{[\txi,\txi]^n}\underline{\txi}{}^{\!\!\!\dag}\nonumber
\end{align}
\begin{align*}
= & \int_{M}   e^{N-3}  \epsilon_n \underline{\mu} F_{\omega} 
	 +  e^{N-3} \iota_{\underline{z}} e F_{\omega}
	 +  e^{N-3}  \lambda \epsilon_n^{[a]} \underline{f}^{\dag}_a F_{\omega} 
	 + (N-3) e^{N-4}\lambda \epsilon_n^{[n]} f^{\dag} \underline{e_n}  F_{\omega}  \\ 
	& - e^{N-3}  \underline{e_n}  d_{\omega}(\lambda \epsilon_n^{[n]} u^{\dag})
	 - e^{N-3}\lambda \epsilon_n^{[a]} \underline{f}^{\dag}_a d_{\omega}(\lambda\epsilon_n^{[n]} u^{\dag}) \\
	& - (N-3)e^{N-4}\lambda\epsilon_n^{[n]}f^{\dag}\underline{e_n}d_{\omega}(\lambda\epsilon_n^{[n]} u^{\dag})  
     +  \frac{1}{N-2} e^{N-2} \left(\underline{\partial}_n  \omega 
	     -  \underline{\partial}_n  \lambda \epsilon_n^{[n]} u^{\dag}   
	     +  d_{\omega} \underline{w} \right) \\
	& -  \frac{1}{N-2} e^{N-2} \left( d_{\omega}( \iota_\xi \underline{u}^\dag)  
	     + d_{\omega}(\lambda \epsilon_n^{[a]} \underline{u}^{\dag}_a) 
 	     + [\lambda \epsilon_n^{[n]} u^{\dag}, \underline{w} - \iota_\xi \underline{u}^\dag]   \right)
     +  e^{N-3} \lambda \epsilon_n^{[n]} f^{\dag} \underline{F_{\omega_n}} \\
\shortintertext{}
    & - \iota_{\xi}F_{\omega}(e^{N-3} \underline{f}^\dag
     + \iota_{\underline{z}} k^\dag +  \iota_\xi \underline{c}^\dag)  
     + \iota_{\xi} d_{\omega}(\lambda \epsilon_n^{[n]} u^{\dag})(e^{N-3} \underline{f}^\dag
     + \iota_{\underline{z}} k^\dag +  \iota_\xi \underline{c}^\dag) 
\\ & -  F_{\omega_a}\lambda \epsilon^{[a]}_n  (e^{N-3} \underline{f}^\dag + \iota_{\underline{z}} k^\dag +  \iota_\xi \underline{c}^\dag)
 	 + d_{\omega}(\lambda \epsilon_n^{[n]} u^{\dag})_a \lambda \epsilon^{[a]}_n (e^{N-3} \underline{f}^\dag + \iota_{\underline{z}} k^\dag +  \iota_\xi \underline{c}^\dag)   
\\ & -  F_{\tom_n}  \lambda \epsilon_n^{[n]}e^{N-3} \underline{f}^\dag 
 	 -  F_{\tom_n}   \lambda \epsilon_n^{[n]} \iota_{\underline{z}} k^\dag 
 	 -  F_{\tom_n}   \lambda \epsilon_n^{[n]} \iota_\xi \underline{c}^\dag   
\\ & + d_{\omega}c (e^{N-3} \underline{f}^\dag + \iota_{\underline{z}} k^\dag +  \iota_\xi \underline{c}^\dag) 
	 -  [\lambda \epsilon_n^{[n]} u^{\dag},c](e^{N-3} \underline{f}^\dag + \iota_{\underline{z}} k^\dag +  \iota_\xi \underline{c}^\dag) 
\\ & -  d_{\omega}(\iota_{\xi} \lambda \epsilon_n^{[n]} u^{\dag})(e^{N-3} \underline{f}^\dag + \iota_{\underline{z}} k^\dag +  \iota_\xi \underline{c}^\dag) 
 \shortintertext{}& -  \iota_{\xi}\underline{\partial}_n  \omega  k^\dag 
 +   \iota_{\xi}\underline{\partial}_n ( \lambda \epsilon_n^{[n]} u^{\dag}) k^\dag 
 -  \iota_{\xi}d_{\omega} \underline{w} k^\dag  +  \iota_{\xi}d_{\omega}(\iota_\xi \underline{u}^\dag) k^\dag 
 \\ &+   \iota_{\xi}d_{\omega}(\lambda \epsilon_n^{[a]} \underline{u}^{\dag}_a) k^\dag 
 +  \iota_{\xi}[\lambda \epsilon_n^{[n]} u^{\dag}, \underline{w} 
 -  \iota_\xi \underline{u}^\dag]  k^\dag  
 -  \underline{F_{\tom_{an}}} \lambda \epsilon^{[a]}_n k^\dag   
 + \underline{\partial}_n c k^{\dag} 
 \\ &-  \underline{\partial}_n (\iota_{\xi} \lambda \epsilon_n^{[n]} u^{\dag}) k^{\dag} 
 +  [ \underline{w},c]k^{\dag} -  [\iota_\xi \underline{u}^\dag,c]k^{\dag}
 -  [\lambda \epsilon_n^{[a]} \underline{u}^{\dag}_a,c]k^{\dag} 
 \\ &+  [ \underline{w} - \iota_\xi \underline{u}^\dag,\iota_{\xi} \lambda \epsilon_n^{[n]} u^{\dag}]k^{\dag} 
\shortintertext{} 
    & +  L_{\xi}^{\omega} e e^{N-3} \underline{u}^{\dag} 
     +   L_{\xi}^{\omega} e \iota_{\underline{z}} e^\dag 
     -  L_{\xi}^{\omega} e  \lambda \epsilon_n^{[a]} \underline{y}_a^\dag 
     + (N-3)  L_{\xi}^{\omega} e e^{N-4}\lambda \epsilon_n^{[n]}f^{\dag}\underline{u}^\dag\\
    & +  ((d_{\omega}e)_a \lambda \epsilon_n^{[a]}  
     -  d_{\omega}(e_a \lambda \epsilon_n^{[a]})) e^{N-3} \underline{u}^{\dag} 
     + ((d_{\omega}e)_a \lambda \epsilon_n^{[a]} 
     -  d_{\omega}(e_a \lambda \epsilon_n^{[a]})) \iota_{\underline{z}} e^\dag \\
    & +  d_{\omega}(e_a \lambda\epsilon_n^{[a]})\lambda\epsilon_n^{[a]} \underline{y}_a^\dag
     - (N-3) d_{\omega}(e_a \lambda \epsilon_n^{[a]})e^{N-4}\lambda \epsilon_n^{[n]}f^{\dag}\underline{u}^\dag 
     + L_{\xi}^{\omega} (\lambda \epsilon_n^{[n]} f^{\dag}) e^{N-3} \underline{u}^{\dag}\\ 
    & + L_{\xi}^{\omega} (\lambda \epsilon_n^{[n]} f^{\dag}) \iota_{\underline{z}} e^\dag 
     - L_{\xi}^{\omega} (\lambda \epsilon_n^{[n]} f^{\dag}) \lambda \epsilon_n^{[a]} \underline{y}_a^\dag 
     +  (d_{\omega}(\lambda \epsilon_n^{[n]} f^{\dag}))_a \lambda \epsilon_n^{[a]} e^{N-3} \underline{u}^{\dag} \\  
    & +  (d_{\omega}(\lambda \epsilon_n^{[n]} f^{\dag}))_a \lambda \epsilon_n^{[a]}\iota_{\underline{z}} e^\dag)  
     - [\iota_{\xi}\lambda \epsilon_n^{[n]} u^{\dag} ,e] (e^{N-3} \underline{u}^{\dag} + \iota_{\underline{z}} e^\dag) 
     + \partial_n e \lambda \epsilon_n^{[n]} e^{N-3} \underline{u}^{\dag} \\
    & + \partial_n e \lambda \epsilon_n^{[n]} \iota_{\underline{z}} e^\dag 
     +  [w  - \iota_{\xi}u^{\dag}, e] \lambda \epsilon_n^{[n]} e^{N-3} \underline{u}^{\dag} 
 +  [w  - \iota_{\xi}u^{\dag}, e] \lambda \epsilon_n^{[n]}\iota_{\underline{z}} e^\dag 
 \\ &+   \partial_n (\lambda \epsilon_n^{[n]} f^{\dag}) \lambda \epsilon_n^{[n]}(e^{N-3} \underline{u}^{\dag} + \iota_{\underline{z}} e^\dag) 
 +  \underline{e_n}d(\lambda \epsilon_n^{[n]}) e^{N-3} \underline{u}^{\dag} 
 +  \underline{e_n}d(\lambda \epsilon_n^{[n]})\iota_{\underline{z}} e^\dag
 \\ &+  \lambda \epsilon_n^{[a]} f^{\dag}_a d(\lambda \epsilon_n^{[n]})  e^{N-3} \underline{u}^{\dag} 
 +  \lambda \epsilon_n^{[a]} f^{\dag}_a d(\lambda \epsilon_n^{[n]})  \iota_{\underline{z}} e^\dag
 -  \underline{e_n}d(\lambda \epsilon_n^{[n]}) \lambda \epsilon_n^{[a]} \underline{y}_a^\dag 
\\ & + (N-3) \underline{e_n}d(\lambda \epsilon_n^{[n]})e^{N-4}\lambda \epsilon_n^{[n]}f^{\dag}\underline{u}^\dag
 -  [c,e]e^{N-3} \underline{u}^{\dag} 
 -  [c,e] \iota_{\underline{z}} e^\dag
 + [c,e]\lambda \epsilon_n^{[a]} \underline{y}_a^\dag 
 \\ &- (N-3)[c,e]e^{N-4}\lambda \epsilon_n^{[n]}f^{\dag}\underline{u}^\dag 
 -  [c, \lambda \epsilon_n^{[n]} f^{\dag}] e^{N-3} \underline{u}^{\dag}
 -  [c, \lambda \epsilon_n^{[n]} f^{\dag}]\iota_{\underline{z}} e^\dag 
 \\ &+  [\iota_{\xi} \lambda \epsilon_n^{[n]} u^{\dag},e](e^{N-3} \underline{u}^{\dag}
 + \iota_{\underline{z}} e^\dag)
\shortintertext{} &+ \iota_{\xi} d_{\omega}(  \epsilon_n \underline{\mu})e^\dag
 +  \iota_{\xi} d_{\omega}(\iota_{\underline{z}} e )e^\dag 
 +   \iota_{\xi} d_{\omega}(\lambda \epsilon_n^{[a]} \underline{f}^{\dag}_a)e^\dag
 +  d_{\omega_a}( \underline{e_n})\lambda \epsilon_n^{[a]} e^\dag
 \\ &+   d_{\omega_a}(\lambda \epsilon_n^{[a]} \underline{f}^{\dag}_a)\lambda \epsilon_n^{[a]} e^\dag 
 -  \iota_{\xi} [\lambda \epsilon_n^{[n]} u^{\dag}, \underline{e_n}] e^\dag 
 -  \iota_{\xi} d_{\omega}( \underline{e_n} )\lambda \epsilon_n^{[n]} y_n^\dag 
\\ & -  \iota_{\xi} d_{\omega}(\lambda \epsilon_n^{[a]} \underline{f}^{\dag}_a)\lambda \epsilon_n^{[n]} y_n^\dag
 -  \iota_{\underline{\partial}_n \xi } e e^\dag 
 +    e_a \underline{\partial}_n (\lambda \epsilon_n^{[a]}) e^\dag 
 - \iota_{\underline{\partial}_n \xi{} } \lambda \epsilon_n^{[n]} f^{\dag} e^\dag
 \\ &+ \iota_{\underline{\partial}_n \xi{} } e \lambda \epsilon_n^{[n]} y_n^\dag
 -  \underline{\partial}_n ( \te_n \lambda \epsilon_n^{[n]} )e^\dag 
 -  [\underline{w}  -  \iota_\xi \underline{u}^\dag, \te_n \lambda \epsilon_n^{[n]}]e^\dag 
 \\ &+  \underline{\partial}_n ( \te_n \lambda \epsilon_n^{[n]} )\lambda \epsilon_n^{[n]} y_n^\dag
 -  [c,  \epsilon_n \underline{\mu} ]e^\dag 
 -  [c,\iota_{\underline{z}} e]e^\dag 
 -  [c, \lambda \epsilon_n^{[a]} \underline{f}^{\dag}_a]e^\dag 
 \\ &+  [\iota_{\xi} \lambda \epsilon_n^{[n]} u^{\dag}, \underline{e_n} ]e^\dag  
 + [c, \underline{e_n} ]\lambda \epsilon_n^{[n]} y_n^\dag 
\shortintertext{} &- \frac{1}{2} \iota_{\xi}\iota_{\xi} F_{\omega} \underline{c}^{\dag} 
 +  \frac{1}{2} \iota_{\xi}\iota_{\xi}d_{\omega}(\lambda \epsilon_n^{[n]} u^{\dag})\underline{c}^{\dag} 
 -  \iota_{\xi} F_{\omega_a} \lambda \epsilon_n^{[a]} \underline{c}^{\dag} 
 \\ &+  \iota_{\xi}d_{\omega}(\lambda \epsilon_n^{[n]} u^{\dag})_a \lambda \epsilon_n^{[a]} \underline{c}^{\dag} 
 -  \iota_{\xi} F_{\tom_n}  \lambda \epsilon_n^{[n]} \underline{c}^{\dag} 
 +  \frac{1}{2}[c,c]\underline{c}^{\dag}
 -  [\iota_{\xi} \lambda \epsilon_n^{[n]} u^{\dag},c]\underline{c}^{\dag}
\shortintertext{}& -  \xi^b \partial_b \xi^a e_a \underline{y}^\dag
 -   \xi^b \partial_b (\lambda \epsilon_n^{[a]})e_a \underline{y}^\dag 
 -  \lambda \epsilon_n^{[b]}\partial_b \xi^a e_a \underline{y}^\dag  
 -  \lambda \epsilon_n^{[b]}\partial_b (\lambda \epsilon_n^{[a]}) e_a \underline{y}^\dag  
\\ & - \lambda \epsilon_n^{[n]}\partial_n \xi{a} e_a \underline{y}^\dag
 - \xi^b \partial_b \xi^a\underline{f}_a^\dag e^\dag -   \xi^b \partial_b (\lambda \epsilon_n^{[a]})\underline{f}_a^\dag e^\dag 
 -  \lambda \epsilon_n^{[b]}\partial_b \xi^a \underline{f}_a^\dag e^\dag
 \\ &-  \lambda \epsilon_n^{[b]}\partial_b (\lambda \epsilon_n^{[a]}) \underline{f}_a^\dag e^\dag   
 - \lambda \epsilon_n^{[n]}\partial_n \xi{a} \underline{f}_a^\dag e^\dag 
 - \xi^b \partial_b \xi^a  \underline{u}_a^\dag k^\dag
 -   \xi^b \partial_b (\lambda \epsilon_n^{[a]})\underline{u}_a^\dag k^\dag 
\\ & -  \lambda \epsilon_n^{[b]}\partial_b \xi^a \underline{u}_a^\dag k^\dag 
 +  \lambda \epsilon_n^{[b]}\partial_b (\lambda \epsilon_n^{[a]})\underline{u}_a^\dag k^\dag 
 - \lambda \epsilon_n^{[n]}\partial_n \xi^a\underline{u}_a^\dag k^\dag 
 \\ &-  \lambda \epsilon_n^{[n]}\partial_n (\lambda \epsilon_n^{[a]})\underline{u}_a^\dag k^\dag 
 -  \xi^b \partial_b \xi^a \underline{c}_a^\dag \lambda \epsilon_n^{[n]} u^{\dag}
 -  \xi^b \partial_b (\lambda \epsilon_n^{[a]})\underline{c}_a^\dag \lambda \epsilon_n^{[n]} u^{\dag}
 \\ &-  \xi^a \partial_a (\lambda \epsilon_n^{[n]}) \te_n \underline{y}^\dag
 -  \lambda \epsilon_n^{[a]}\partial_a (\lambda \epsilon_n^{[n]})\te_n \underline{y}^\dag
 - \lambda \epsilon_n^{[n]}\partial_n (\lambda \epsilon_n^{[n]})\te_n \underline{y}^\dag
 \\ &-  \xi^a \partial_a (\lambda \epsilon_n^{[n]})f^\dag e^{N-3} \underline{u}^\dag 
 -  \lambda \epsilon_n^{[a]}\partial_a (\lambda \epsilon_n^{[n]})f^\dag e^{N-3} \underline{u}^\dag 
 - \lambda \epsilon_n^{[n]}\partial_n (\lambda \epsilon_n^{[n]})f^\dag e^{N-3} \underline{u}^\dag
 \\ &-  \xi^a \partial_a (\lambda \epsilon_n^{[n]})f^\dag\iota_{\underline{z}} e^\dag 
 -  \lambda \epsilon_n^{[a]}\partial_a (\lambda \epsilon_n^{[n]}) f^\dag\iota_{\underline{z}} e^\dag
 - \lambda \epsilon_n^{[n]}\partial_n (\lambda \epsilon_n^{[n]})f^\dag\iota_{\underline{z}} e^\dag
\\ & +  \xi^a \partial_a (\lambda \epsilon_n^{[n]})u^\dag \iota_{\underline{z}} k^\dag 
 +  \lambda \epsilon_n^{[a]}\partial_a (\lambda \epsilon_n^{[n]})u^\dag \iota_{\underline{z}} k^\dag 
 + \lambda \epsilon_n^{[n]}\partial_n (\lambda \epsilon_n^{[n]})u^\dag \iota_{\underline{z}} k^\dag
 \\ &-  \xi^a \partial_a (\lambda \epsilon_n^{[n]})c^\dag \lambda \epsilon_n^{[a]} \underline{u}^{\dag}_a 
\end{align*}
We want to compare this with the AKSZ action:
\begin{align}
S^{\AKSZ}_{PC}=\int_{I\times \partial M} &  \underline{w} e^{N-3} d_{\omega} e
	+(N-3) c e^{N-4} \underline{f^\dag} d_{\omega} e
	+c e^{N-3} [\underline{u^\dag}, e]
	+c e^{N-3} d_{\omega} \underline{f^\dag} \nonumber\\
	+& \iota_{\underline{z}} e e^{N-3} F_{\omega} 
 	+ \iota_{\xi}( e^{N-3} \underline{f^\dag}) F_{\omega}
	+ \iota_{\xi} e e^{N-3} d_{\omega}\underline{u^\dag} 
	+ \epsilon_n \underline{\mu} e^{N-3} F_{\omega}\nonumber \\
	+&(N-3) \epsilon_n \lambda e^{N-4} \underline{f^\dag} F_{\omega}
	+ \epsilon_n \lambda e^{N-3} d_{\omega}\underline{u^\dag}  
	+ [\underline{w},c] k^{\dag}
	+\frac{1}{2} [c,c] \underline{c^{\dag}} \nonumber \\
	-& \iota_{\underline{z}} d_{\omega} c k^{\dag}
	- [\iota_{\xi}\underline{u^\dag} ,c] k^{\dag}
	- \iota_{\xi} d_{\omega}\underline{w} k^{\dag} 
	- \iota_{\xi} d_{\omega} c \underline{c^{\dag}}
	+ \iota_{\underline{z}}\iota_{\xi} F_{\omega}k^{\dag}\nonumber\\
	+& \frac{1}{2} \iota_{\xi}\iota_{\xi} d_{\omega}\underline{u^\dag} k^{\dag}
	+ \frac{1}{2} \iota_{\xi}\iota_{\xi} F_{\omega}\underline{c^{\dag}}
	-[\underline{w},  \epsilon_n \lambda]e^{\dag}
	-[c,  \epsilon_n \underline{\mu}]e^{\dag} 
	-[c,  \epsilon_n \lambda]\underline{y^{\dag}} \nonumber \\
	+& \iota_{\underline{z}}d_{\omega} ( \epsilon_n\lambda)e^{\dag} 
	+ [\iota_{\xi}\underline{u^\dag} , \epsilon_n\lambda]e^{\dag}
	+ \iota_{\xi}d_{\omega} ( \epsilon_n\underline{\mu})e^{\dag} 
	+ \iota_{\xi}d_{\omega} ( \epsilon_n\lambda)\underline{y^{\dag}}
	+ \iota_{[\underline{z},\xi]}e e^{\dag}\nonumber \\
  	+& \frac{1}{2}\iota_{[\xi,\xi]}\underline{f^\dag} e^{\dag}
	+ \frac{1}{2}\iota_{[\xi,\xi]}e \underline{y^{\dag}}  
	+ \frac{1}{N-2} e^{N-2} \underline{\partial_n} \omega  
	+ c \underline{\partial_n} k^\dag
	+ \underline{\partial_n} \omega \iota_\xi k^\dag \nonumber \\ 
	-& \iota_{\underline{\partial_n} \xi} e  e^\dag 
	+ \underline{\partial_n} \lambda \epsilon_n e^\dag. \label{e:S_AKSZ-compare_Nd}
\end{align}

We proceed as follows: we first check that all terms in \eqref{e:S_AKSZ-compare_Nd} can be found in \eqref{e:phiS_AKSZ_Nd}, then we show that all other terms in \eqref{e:phiS_AKSZ_Nd} sum to zero. \\
We can easily recognized many terms identically repeated in both expressions. 
Some other terms in \eqref{e:S_AKSZ-compare_Nd} can be recovered in \eqref{e:phiS_AKSZ_Nd}  using Leibniz rule and Cartan calculus.

\begin{align*}
(N-3) c e^{N-4} \underline{f}^\dag d_{\omega} e + c ^{N-3} d_{\omega} \underline{f}^\dag =  d_{\omega} c  ( e^{N-3} \underline{f}^\dag); 
\end{align*}
\begin{align*}
-\frac{1}{N-2} e^{N-2} d_{\omega}( \iota_\xi \underline{u}^\dag) +L_{\xi}^{\omega} e e^{N-3} \underline{u}^{\dag}= +\iota_{\xi} e e^{N-3} d_{\omega} \underline{u}^{\dag}; 
\end{align*}
\begin{align*}
 - \frac{1}{2}\iota_{[\xi,\xi]}\underline{u}^\dag k^\dag = - \iota_{\xi} d_{\omega}\iota_{\xi}\underline{u}^\dag k^\dag + \frac{1}{2} \iota_{\xi} \iota_{\xi}d_{\omega} \underline{u}^\dag k^\dag;
\end{align*}
\begin{align*}
\iota_{[\underline{z},\xi]}e e^{\dag} =  \iota_{\xi} d_{\omega} \iota_{\underline{z}} e e^{\dag} + L_{\xi}^{\omega} e \iota_{\underline{z}} e^{\dag}
\end{align*}

All the other relations involving terms of 	\eqref{e:phiS_AKSZ_Nd} are based on the expansion 
\begin{align*}
 \epsilon_n = e_a  \epsilon_n^{(a)} +e_n  \epsilon_n^{(n)}.
\end{align*}

It is a long but rather easy computation to show that the remaining terms in \eqref{e:S_AKSZ-compare_Nd} sum to zero.  This is done by making repeated use of Cartan calculus and Leibniz rule. Notice also that some terms containing expressions of the form $\epsilon_{n}^{[a]}\epsilon_{n}^{[b]} X_{ab}$ vanish by antisymmetry.